\documentclass[11pt,twoside]{amsart}

\usepackage{amssymb}
\usepackage{amsthm}
\usepackage{amsmath}
\usepackage{hyperref}
\usepackage{ifthen}
\usepackage{mathtools}

\usepackage{etoolbox}
\let\bbordermatrix\bordermatrix
\patchcmd{\bbordermatrix}{8.75}{4.75}{}{}
\patchcmd{\bbordermatrix}{\left(}{\left[}{}{}
\patchcmd{\bbordermatrix}{\right)}{\right]}{}{}

\usepackage{graphicx}
\usepackage{tikz}
\usetikzlibrary{arrows}
\usetikzlibrary{shapes,snakes}

\usepackage{ tikz }
\usetikzlibrary{ calc }

\theoremstyle{plain}
\newtheorem{theorem}{Theorem}
\newtheorem{lemma}[theorem]{Lemma}
\newtheorem{corollary}[theorem]{Corollary}
\newtheorem{proposition}[theorem]{Proposition}

\theoremstyle{definition}

\newtheorem{example}[theorem]{Example}

\newtheorem{problem}[theorem]{Problem}

\theoremstyle{remark}

\newcommand{\set}[1]{\left\{#1\right\}}

\def\tn{\textnormal}
\def\ld{\lambda}
\def\mN{\mathbb{N}}
\def\mR{\mathbb{R}}
\def\mS{\mathbb{S}}
\def\mZ{\mathbb{Z}}

\def\S{\mathcal{S}}

\def\tn{\textnormal}
\def\ep{\epsilon}

\def\het{\hat}
\def\a{\alpha}
\def\b{\beta}

\def\ce{\coloneqq}

\newcommand{\ignore}[1]{}
\DeclareMathOperator{\SA}{SA}
\DeclareMathOperator{\LS}{LS}
\DeclareMathOperator{\BZ}{BZ}
\DeclareMathOperator{\Las}{Las}

\DeclareMathOperator{\FRAC}{FRAC}
\DeclareMathOperator{\CLIQ}{CLIQ}
\DeclareMathOperator{\STAB}{STAB}

\DeclareMathOperator{\diag}{diag}
\DeclareMathOperator{\conv}{conv}
\DeclareMathOperator{\cone}{cone}

\title[Rank-Monotone Operations and Minimal Graphs $\LS_+$]{On Rank-Monotone Graph Operations \\ and Minimal Obstruction Graphs \\ for the Lov{\'a}sz--Schrijver SDP Hierarchy}

\author{Yu Hin (Gary) Au}
\thanks{Yu Hin (Gary) Au: Corresponding author. Department of Mathematics and Statistics, University of Saskatchewan, Saskatoon, Saskatchewan, S7N 5E6 Canada. E-mail: gary.au@usask.ca}

\author{Levent Tun{\c c}el}
\thanks{Levent Tun{\c c}el: Research of this author was supported in part by an NSERC Discovery Grant. Department of Combinatorics and Optimization, Faculty of Mathematics, University of Waterloo, Waterloo, Ontario, N2L 3G1 Canada. E-mail: levent.tuncel@uwaterloo.ca}

\date{\today}


\begin{document}

\maketitle 

\begin{abstract}
We study the lift-and-project rank of the stable set polytopes of graphs with respect to the Lov{\'a}sz--Schrijver SDP operator $\LS_+$, with a particular focus on finding and characterizing the smallest graphs with a given $\LS_+$-rank (the needed number of iterations of the $\LS_+$ operator on the fractional stable set polytope to compute the stable set polytope). We introduce a generalized vertex-stretching operation that appears to be promising in generating $\LS_+$-minimal graphs and study its properties. We also provide several new $\LS_+$-minimal graphs, most notably the first known instances of $12$-vertex graphs with $\LS_+$-rank $4$, which provides the first advance in this direction since Escalante, Montelar, and Nasini's discovery of a $9$-vertex graph with $\LS_+$-rank $3$ in 2006.
\end{abstract}

\section{Introduction}\label{sec1}

Given a simple, undirected graph $G = (V(G), E(G))$, we say that $S \subseteq V(G)$ is a \emph{stable set} if no two vertices in $S$ are joined by an edge. The \emph{(maximum) stable set problem}, which aims to find a stable set of maximum cardinality in a given graph $G$, is one of the most well-studied problems in combinatorial optimization. While this problem is $\mathcal{NP}$-hard, a standard approach for tackling the problem is to associate stable sets of $G$ with points in $\mathbb{R}^{V(G)}$, and model it as a convex optimization problem. Given a set $S \subseteq V(G)$, its \emph{incidence vector} $\chi_S \in \set{0,1}^{V(G)}$ is defined so that $[\chi_S]_i = 1$ if $i \in S$, and $[\chi_S]_i =0$ otherwise. Then we define the \emph{stable set polytope} of a given graph $G$ to be the convex hull of the incidence vectors of stable sets of $G$:
\[
\STAB(G) \ce \conv\left( \set{ \chi_S :~\tn{$S \subseteq V(G)$ is a stable set of $G$}}\right).
\]
Observe that if we let $\alpha(G)$ be the cardinality of a maximum stable set in $G$, then
\begin{equation}\label{eqSTABG}
\alpha(G) = \max \set{\sum_{i \in V(G)} x_i : x \in \STAB(G)}. 
\end{equation}
While~\eqref{eqSTABG} is a linear program, considering again that the underlying combinatorial problem is $\mathcal{NP}$-hard, it is a difficult task to find an explicit description (e.g., via listing its facets) of $\STAB(G)$ for a general graph $G$. This naturally leads to the pursuit of ``nice'' convex relaxations of $\STAB(G)$. Below we list several desirable characteristics of such a convex relaxation $P$:
\begin{itemize}
\item[(i)]
$P \cap \set{0,1}^{V(G)} = \STAB(G) \cap \set{0,1}^{V(G)}$. That is, a $0$-$1$ vector is in $P$ if and only if it is the incidence vector of a stable set in $G$.
\item[(ii)]
$P$ is tractable. That is, one can optimize a linear function over $P$ with arbitrary precision in polynomial time.
\item[(iii)]
$P$ is a ``strong'' relaxation. This is a relatively subjective measure, and can mean that important families of valid inequalities of $\STAB(G)$ are also valid for $P$, and/or that $\max \set{ \sum_{i \in V(G)} x_i : x \in P}$ is ``close'' to $\alpha(G)$.
\end{itemize}
One of the simplest convex relaxations of $\STAB(G)$ is the \emph{fractional stable set polytope}
\[
\FRAC(G) \ce \set{ x \in [0,1]^{V(G)}: x_i + x_j \leq 1, \forall \set{i,j} \in E(G)}.
\]
While $\FRAC(G)$ satisfies properties (i) and (ii), it is rather weak in general. For a stronger relaxation, we call $C \subseteq V(G)$ a \emph{clique} if every pair of vertices in $C$ is joined by an edge. Then notice that, for every clique $C$, the inequality
\[
\sum_{i \in C} x_i \leq 1
\]
is valid for $\STAB(G)$. Thus, if we define the \emph{clique polytope}
\[
\CLIQ(G) \ce \set{ x \in [0,1]^{V(G)} : \sum_{i \in C} x_ i \leq 1~\tn{for every clique $C \subseteq V(G)$}},
\]
then $\STAB(G) \subseteq \CLIQ(G) \subseteq \FRAC(G)$ for every graph $G$. (For the second containment. observe that every edge is a clique of size $2$.) However, while $\CLIQ(G)$ is a stronger relaxation than $\FRAC(G)$, it is not tractable in general. 

In this manuscript, we focus on semidefinite relaxations of $\STAB(G)$ produced by $\LS_+$, a lift-and-project operator devised by Lov{\'a}sz and Schrijver~\cite{LovaszS91} which we fully define in Section~\ref{sec2}. (The operator has also been referred to as $N_+$ in the literature.) Given a graph $G$, the $\LS_+$ operator generates a sequence of relaxations $\LS_+^k(G)$ which satisfies
\[
\FRAC(G) =: \LS_+^0(G) \supseteq \LS_+^1(G) \supseteq \LS_+^2(G) \supseteq \cdots  \supseteq \LS_+^{|V(G)|}(G) = \STAB(G).
\]
(We will usually refer to $\LS_+^1(G)$ as simply $\LS_+(G)$.) When $k = O(1)$, $\LS_+^k(G)$ can be described as the feasible region of a semidefinite program whose number of variables and constraints are polynomial in the size of the number of vertices and edges in $G$, and thus the relaxation is indeed tractable in this case. Moreover, the first relaxation $\LS_+(G)$ already satisfies many well-known families of valid inequalities of $\STAB(G)$, including (among others) the aforementioned clique inequalities, orthogonality constraints imposed by the Lov{\'a}sz theta body~\cite{Lovasz79}, as well as odd hole, odd antihole, and odd wheel constraints~\cite{LovaszS91}.

The hierarchy of relaxations generated by $\LS_+$ gives rise to the notion of the \emph{$\LS_+$-rank} of a graph $G$, which is defined to be the smallest integer $k$ where $\LS_+^k(G) = \STAB(G)$, and gives us a measure of how difficult the stable set problem is for the $\LS_+$ operator. It is well-known that a graph $G$ has $\LS_+$-rank $0$ (i.e., satisfies $\FRAC(G) = \STAB(G)$) if and only if $G$ is bipartite. Some families of graphs that are known to have $\LS_+$-rank $1$ (i.e., satisfy $\LS_+(G) = \STAB(G)$) include --- but are not limited to --- odd cycles, odd antiholes, odd wheels, and perfect graphs (which are defined to be the graphs with $\CLIQ(G) = \STAB(G)$). In the last decade, considerable progress has been made in finding a combinatorial characterization of graphs with $\LS_+$-rank $1$ --- see, for instance,~\cite{BianchiENT13, BianchiENT17, Wagler22, BianchiENW23}.

Nevertheless, since the maximum stable problem is $\mathcal{NP}$-hard, there has to be graphs with unbounded $\LS_+$-rank. The first family of graphs that have unbounded $\LS_+$-rank was obtained by Stephen and the second author~\cite{StephenT99}, who showed that the line graph of the complete graph on $2k+1$ vertices has $\LS_+$-rank $k$, giving a family of graphs $G$ whose $\LS_+$-rank is asymptotically $\Omega(\sqrt{|V(G)|})$. On the other hand, Lipt{\'a}k and the second author~\cite{LiptakT03} showed the following:

\begin{theorem}\label{thmNover3}
For every graph $G$, the $\LS_+$-rank of $G$ is at most $\left\lfloor \frac{|V(G)|}{3} \right\rfloor$.
\end{theorem}

This naturally raises the following question: For every integer $\ell \geq 1$, is there a graph on $3\ell$ vertices which has $\LS_+$-rank $\ell$? If these graphs exist, their extremal nature (in terms of being the smallest possible graphs with a given $\LS_+$-rank) may help reveal the critical structures that expose the limitations of these $\LS_+$-relaxations. This understanding could be extremely helpful when it comes to analyzing other convex relaxations of the maximum stable set problem, particularly those which are produced by other lift-and-project methods.

This direction of investigation was already set in the seminal paper~\cite{LovaszS91} and questions about the behaviour of $\LS_+$-rank under simple graph operations were also raised in~\cite{GoemansT01}. In the same general direction of research, Laurent~\cite{Laurent2001} analyzed the $\LS_+$-rank and related ranks in the context of the maximum cut problem by establishing nice behaviour (only in the context of maximum cut problems) of the underlying lift-and-project operators under graph minor operations; also see~\cite{Laurent2003b} for an analysis of the Lasserre operator. However, as it was illustrated in some depth in~\cite{LiptakT03}, the $\LS_+$-rank of a graph does not behave in a nice, uniform way under the usual graph minor operations for the stable set problem. Therefore, a deeper investigation is necessary to construct the kind of graph operations which would be helpful in discovering and understanding minimal obstructions to tractable convex relaxations of the stable set polytope obtained by $\LS_+$ or other convex optimization based lift-and-project hierarchies.
Overall, the importance of the quest to understand minimal obstructions to families of SDP relaxations in particular --- and convex relaxations in general --- has been raised by many others. For example, Knuth, in his well-known survey ``The Sandwich Theorem"~\cite{Knuth1994} poses six open problems in the general context of Lov\'{a}sz theta function. Two of the six open problems concern
$\LS_+(\FRAC(G))$. One of them asks for finding what we call below a $2$-minimal graph (answered in~\cite{LiptakT03}).

We say that a graph $G$ is \emph{$\ell$-minimal} if $|V(G)| = 3\ell$ and $G$ has $\LS_+$-rank $\ell$. It is known that $\ell$-minimal graphs exist for $\ell \in \set{1,2,3}$. For $\ell = 1$, it is easy to see that the $3$-cycle is the only $1$-minimal graph. The first $2$-minimal graph ($G_{2,1}$ in Figure~\ref{figKnownEG}) was found by Lipt{\'a}k and the second author~\cite{LiptakT03}, who also conjectured that $\ell$-minimal graphs exist for all $\ell \in \mN$. Subsequently, Escalante, Montelar, and Nasini~\cite{EscalanteMN06} showed that there is only one other $2$-minimal graph ($G_{2,2}$ in Figure~\ref{figKnownEG}), while providing the first example of a $3$-minimal graph ($G_{3,1}$ in Figure~\ref{figKnownEG}).  (The logic behind the seemingly odd choice of vertex labels in the figures of this section will be explained in Section~\ref{sec4} when we introduce the vertex-stretching operation.)

\def\y{0.70}
\def\sc{2}
\def\x{180}
\def\z{360/4}

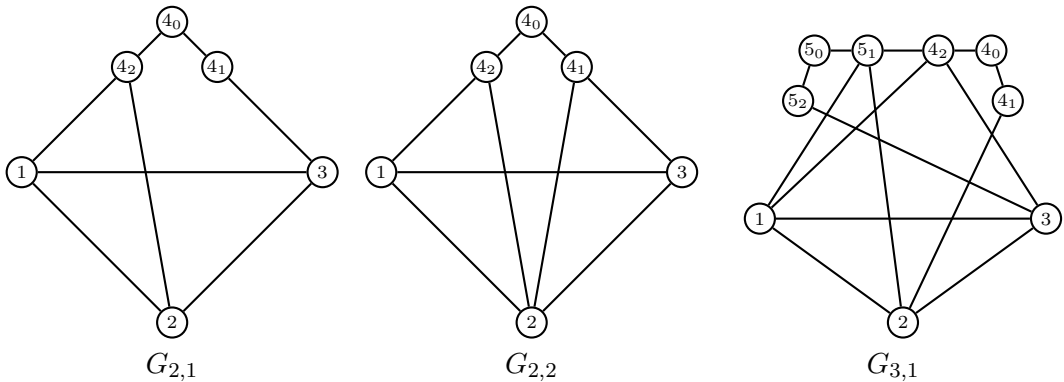
\begin{figure}[ht!]
\begin{center}
\begin{tabular}{ccc}

\begin{tikzpicture}
[scale=\sc, thick,main node/.style={circle, minimum size=4mm, inner sep=0.1mm,draw,font=\tiny\sffamily}]

\node[main node] at ({cos(\x+(0)*\z)},{sin(\x+(0)*\z)}) (1) {$1$};
\node[main node] at ({cos(\x+(1)*\z)},{sin(\x+(1)*\z)}) (2) {$2$};
\node[main node] at ({cos(\x+(2)*\z)},{sin(\x+(2)*\z)}) (3) {$3$};

\node[main node] at ({ \y* cos(\x+(3)*\z) + (1-\y)*cos(\x+(2)*\z)},{ \y* sin(\x+(3)*\z) + (1-\y)*sin(\x+(2)*\z)}) (4) {$4_1$};
\node[main node] at ({cos(\x+(3)*\z)},{sin(\x+(3)*\z)}) (5) {$4_0$};
\node[main node] at ({ \y* cos(\x+(3)*\z) + (1-\y)*cos(\x+(4)*\z)},{ \y* sin(\x+(3)*\z) + (1-\y)*sin(\x+(4)*\z)}) (6) {$4_2$};

 \path[every node/.style={font=\sffamily}]
(1) edge (2)
(2) edge (3)
(3) edge (1)
(4) edge (5)
(5) edge (6)
(4) edge (3)
(6) edge (1)
(6) edge (2);
\end{tikzpicture}
&
\begin{tikzpicture}[scale=\sc, thick,main node/.style={circle, minimum size=4mm, inner sep=0.1mm,draw,font=\tiny\sffamily}]
\node[main node] at ({cos(\x+(0)*\z)},{sin(\x+(0)*\z)}) (1) {$1$};
\node[main node] at ({cos(\x+(1)*\z)},{sin(\x+(1)*\z)}) (2) {$2$};
\node[main node] at ({cos(\x+(2)*\z)},{sin(\x+(2)*\z)}) (3) {$3$};

\node[main node] at ({ \y* cos(\x+(3)*\z) + (1-\y)*cos(\x+(2)*\z)},{ \y* sin(\x+(3)*\z) + (1-\y)*sin(\x+(2)*\z)}) (4) {$4_1$};
\node[main node] at ({cos(\x+(3)*\z)},{sin(\x+(3)*\z)}) (5) {$4_0$};
\node[main node] at ({ \y* cos(\x+(3)*\z) + (1-\y)*cos(\x+(4)*\z)},{ \y* sin(\x+(3)*\z) + (1-\y)*sin(\x+(4)*\z)}) (6) {$4_2$};

 \path[every node/.style={font=\sffamily}]
(1) edge (2)
(2) edge (3)
(3) edge (1)
(4) edge (5)
(5) edge (6)
(4) edge (2)
(4) edge (3)
(6) edge (1)
(6) edge (2);
\end{tikzpicture}

&

\def\x{270 - 360/5}
\def\z{360/5}
\begin{tikzpicture}[scale=\sc, thick,main node/.style={circle, minimum size=4mm, inner sep=0.1mm,draw,font=\tiny\sffamily}]
\node[main node] at ({cos(\x+(0)*\z)},{sin(\x+(0)*\z)}) (1) {$1$};
\node[main node] at ({cos(\x+(1)*\z)},{sin(\x+(1)*\z)}) (2) {$2$};
\node[main node] at ({cos(\x+(2)*\z)},{sin(\x+(2)*\z)}) (3) {$3$};

\node[main node] at ({ \y* cos(\x+(3)*\z) + (1-\y)*cos(\x+(2)*\z)},{ \y* sin(\x+(3)*\z) + (1-\y)*sin(\x+(2)*\z)}) (4) {$4_1$};
\node[main node] at ({cos(\x+(3)*\z)},{sin(\x+(3)*\z)}) (5) {$4_0$};
\node[main node] at ({ \y* cos(\x+(3)*\z) + (1-\y)*cos(\x+(4)*\z)},{ \y* sin(\x+(3)*\z) + (1-\y)*sin(\x+(4)*\z)}) (6) {$4_2$};

\node[main node] at ({ \y* cos(\x+(4)*\z) + (1-\y)*cos(\x+(3)*\z)},{ \y* sin(\x+(4)*\z) + (1-\y)*sin(\x+(3)*\z)}) (7) {$5_1$};
\node[main node] at ({cos(\x+(4)*\z)},{sin(\x+(4)*\z)}) (8) {$5_0$};
\node[main node] at ({ \y* cos(\x+(4)*\z) + (1-\y)*cos(\x+(5)*\z)},{ \y* sin(\x+(4)*\z) + (1-\y)*sin(\x+(5)*\z)}) (9) {$5_2$};

 \path[every node/.style={font=\sffamily}]
(1) edge (3)
(1) edge (2)
(2) edge (3)
(4) edge (5)
(5) edge (6)
(7) edge (8)
(8) edge (9)
(4) edge (2)
(6) edge (1)
(6) edge (3)
(7) edge (1)
(7) edge (2)
(9) edge (3)
(6) edge (7);
\end{tikzpicture}

\\
$G_{2,1}$ & $G_{2,2}$ & $G_{3,1}$ 
\end{tabular}
\end{center}
\caption{Known $2$- and $3$-minimal graphs due to~\cite{LiptakT03} and~\cite{EscalanteMN06}}\label{figKnownEG}
\end{figure}

In producing the first $3$-minimal graph, Escalante et al.~\cite{EscalanteMN06} also showed that there does not exist an $\ell$-minimal graph for any $\ell \geq 4$ if we restrict ourselves to graphs that can be obtained by starting with a complete graph and replacing every edge by a path of length at least $1$. (Let $K_n$ denote the complete graph on $n$ vertices. Notice that $G_{2,1}$ and $G_{3,1}$ can be respectively obtained from $K_4$ and $K_5$ by replacing some edges with paths of length $3$.) 
 
Recently, the authors~\cite{AuT24} discovered several family of graphs $G$ which has $\LS_+$-rank $\Omega(|V(G)|)$. One of them is the family of graphs $H_k$, which is defined as follows. Given $k \in \mN$, let $[k]$ denote the set $\set{1,2,\ldots,k}$. For every $k \geq 3$, let
\[
V(H_k) \ce \set{ i_0, i_1, i_2 : i \in [k]}
\]
and
\[
E(H_k) \ce \set{ \set{i_1, i_0}, \set{i_0, i_2} : i \in [k]} \cup \set{\set{ i_1, j_2} : i,j \in [k], i \neq j}.
\]

\def\y{0.70}
\def\sc{2}

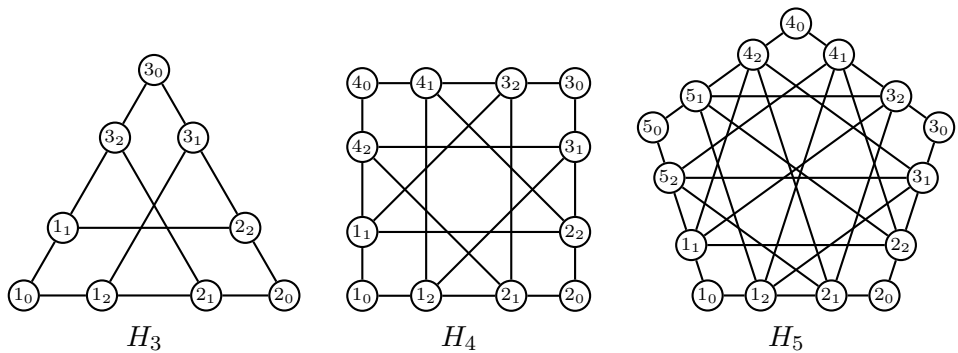
\begin{figure}[ht!]
\begin{center}
\begin{tabular}{ccc}

\def\x{270 - 180/3}
\def\z{360/3}

\begin{tikzpicture}[scale=\sc, thick,main node/.style={circle, minimum size=4mm, inner sep=0.1mm,draw,font=\tiny\sffamily}]
\node[main node] at ({ \y* cos(\x + (0)*\z) + (1-\y)*cos(\x+(-1)*\z)},{ \y* sin(\x+(0)*\z) + (1-\y)*sin(\x+(-1)*\z)}) (1) {$1_1$};
\node[main node] at ({cos(\x+(0)*\z)},{sin(\x+(0)*\z)}) (2) {$1_0$};
\node[main node] at ({ \y* cos(\x+(0)*\z) + (1-\y)*cos(\x+(1)*\z)},{ \y* sin(\x+(0)*\z) + (1-\y)*sin(\x+(1)*\z)}) (3) {$1_2$};

\node[main node] at ({ \y* cos(\x+(1)*\z) + (1-\y)*cos(\x+(0)*\z)},{ \y* sin(\x+(1)*\z) + (1-\y)*sin(\x+(0)*\z)}) (4) {$2_1$};
\node[main node] at ({cos(\x+(1)*\z)},{sin(\x+(1)*\z)}) (5) {$2_0$};
\node[main node] at ({ \y* cos(\x+(1)*\z) + (1-\y)*cos(\x+(2)*\z)},{ \y* sin(\x+(1)*\z) + (1-\y)*sin(\x+(2)*\z)}) (6) {$2_2$};

\node[main node] at ({ \y* cos(\x+(2)*\z) + (1-\y)*cos(\x+(1)*\z)},{ \y* sin(\x+(2)*\z) + (1-\y)*sin(\x+(1)*\z)}) (7) {$3_1$};
\node[main node] at ({cos(\x+(2)*\z)},{sin(\x+(2)*\z)}) (8) {$3_0$};
\node[main node] at ({ \y* cos(\x+(2)*\z) + (1-\y)*cos(\x+(3)*\z)},{ \y* sin(\x+(2)*\z) + (1-\y)*sin(\x+(3)*\z)}) (9) {$3_2$};

 \path[every node/.style={font=\sffamily}]
(2) edge (1)
(2) edge (3)
(5) edge (4)
(5) edge (6)
(8) edge (7)
(8) edge (9)
(1) edge (6)
(1) edge (9)
(4) edge (3)
(4) edge (9)
(7) edge (3)
(7) edge (6);
\end{tikzpicture}

&

\def\x{270 - 180/4}
\def\z{360/4}

\begin{tikzpicture}[scale=\sc, thick,main node/.style={circle, minimum size=4mm, inner sep=0.1mm,draw,font=\tiny\sffamily}]

\node[main node] at ({ \y* cos(\x+(0)*\z) + (1-\y)*cos(\x+(-1)*\z)},{ \y* sin(\x+(0)*\z) + (1-\y)*sin(\x+(-1)*\z)}) (1) {$1_1$};
\node[main node] at ({cos(\x+(0)*\z)},{sin(\x+(0)*\z)}) (2) {$1_0$};
\node[main node] at ({ \y* cos(\x+(0)*\z) + (1-\y)*cos(\x+(1)*\z)},{ \y* sin(\x+(0)*\z) + (1-\y)*sin(\x+(1)*\z)}) (3) {$1_2$};

\node[main node] at ({ \y* cos(\x+(1)*\z) + (1-\y)*cos(\x+(0)*\z)},{ \y* sin(\x+(1)*\z) + (1-\y)*sin(\x+(0)*\z)}) (4) {$2_1$};
\node[main node] at ({cos(\x+(1)*\z)},{sin(\x+(1)*\z)}) (5) {$2_0$};
\node[main node] at ({ \y* cos(\x+(1)*\z) + (1-\y)*cos(\x+(2)*\z)},{ \y* sin(\x+(1)*\z) + (1-\y)*sin(\x+(2)*\z)}) (6) {$2_2$};

\node[main node] at ({ \y* cos(\x+(2)*\z) + (1-\y)*cos(\x+(1)*\z)},{ \y* sin(\x+(2)*\z) + (1-\y)*sin(\x+(1)*\z)}) (7) {$3_1$};
\node[main node] at ({cos(\x+(2)*\z)},{sin(\x+(2)*\z)}) (8) {$3_0$};
\node[main node] at ({ \y* cos(\x+(2)*\z) + (1-\y)*cos(\x+(3)*\z)},{ \y* sin(\x+(2)*\z) + (1-\y)*sin(\x+(3)*\z)}) (9) {$3_2$};

\node[main node] at ({ \y* cos(\x+(3)*\z) + (1-\y)*cos(\x+(2)*\z)},{ \y* sin(\x+(3)*\z) + (1-\y)*sin(\x+(2)*\z)}) (10) {$4_1$};
\node[main node] at ({cos(\x+(3)*\z)},{sin(\x+(3)*\z)}) (11) {$4_0$};
\node[main node] at ({ \y* cos(\x+(3)*\z) + (1-\y)*cos(\x+(4)*\z)},{ \y* sin(\x+(3)*\z) + (1-\y)*sin(\x+(4)*\z)}) (12) {$4_2$};

 \path[every node/.style={font=\sffamily}]
(2) edge (1)
(2) edge (3)
(5) edge (4)
(5) edge (6)
(8) edge (7)
(8) edge (9)
(11) edge (10)
(11) edge (12)
(1) edge (6)
(1) edge (9)
(1) edge (12)
(4) edge (3)
(4) edge (9)
(4) edge (12)
(7) edge (3)
(7) edge (6)
(7) edge (12)
(10) edge (3)
(10) edge (6)
(10) edge (9);
\end{tikzpicture}

&

\def\x{270 - 180/5}
\def\z{360/5}

\begin{tikzpicture}[scale=\sc, thick,main node/.style={circle, minimum size=4mm, inner sep=0.1mm,draw,font=\tiny\sffamily}]
\node[main node] at ({ \y* cos(\x+(0)*\z) + (1-\y)*cos(\x+(-1)*\z)},{ \y* sin(\x+(0)*\z) + (1-\y)*sin(\x+(-1)*\z)}) (1) {$1_1$};
\node[main node] at ({cos(\x+(0)*\z)},{sin(\x+(0)*\z)}) (2) {$1_0$};
\node[main node] at ({ \y* cos(\x+(0)*\z) + (1-\y)*cos(\x+(1)*\z)},{ \y* sin(\x+(0)*\z) + (1-\y)*sin(\x+(1)*\z)}) (3) {$1_2$};

\node[main node] at ({ \y* cos(\x+(1)*\z) + (1-\y)*cos(\x+(0)*\z)},{ \y* sin(\x+(1)*\z) + (1-\y)*sin(\x+(0)*\z)}) (4) {$2_1$};
\node[main node] at ({cos(\x+(1)*\z)},{sin(\x+(1)*\z)}) (5) {$2_0$};
\node[main node] at ({ \y* cos(\x+(1)*\z) + (1-\y)*cos(\x+(2)*\z)},{ \y* sin(\x+(1)*\z) + (1-\y)*sin(\x+(2)*\z)}) (6) {$2_2$};

\node[main node] at ({ \y* cos(\x+(2)*\z) + (1-\y)*cos(\x+(1)*\z)},{ \y* sin(\x+(2)*\z) + (1-\y)*sin(\x+(1)*\z)}) (7) {$3_1$};
\node[main node] at ({cos(\x+(2)*\z)},{sin(\x+(2)*\z)}) (8) {$3_0$};
\node[main node] at ({ \y* cos(\x+(2)*\z) + (1-\y)*cos(\x+(3)*\z)},{ \y* sin(\x+(2)*\z) + (1-\y)*sin(\x+(3)*\z)}) (9) {$3_2$};

\node[main node] at ({ \y* cos(\x+(3)*\z) + (1-\y)*cos(\x+(2)*\z)},{ \y* sin(\x+(3)*\z) + (1-\y)*sin(\x+(2)*\z)}) (10) {$4_1$};
\node[main node] at ({cos(\x+(3)*\z)},{sin(\x+(3)*\z)}) (11) {$4_0$};
\node[main node] at ({ \y* cos(\x+(3)*\z) + (1-\y)*cos(\x+(4)*\z)},{ \y* sin(\x+(3)*\z) + (1-\y)*sin(\x+(4)*\z)}) (12) {$4_2$};

\node[main node] at ({ \y* cos(\x+(4)*\z) + (1-\y)*cos(\x+(3)*\z)},{ \y* sin(\x+(4)*\z) + (1-\y)*sin(\x+(3)*\z)}) (13) {$5_1$};
\node[main node] at ({cos(\x+(4)*\z)},{sin(\x+(4)*\z)}) (14) {$5_0$};
\node[main node] at ({ \y* cos(\x+(4)*\z) + (1-\y)*cos(\x+(5)*\z)},{ \y* sin(\x+(4)*\z) + (1-\y)*sin(\x+(5)*\z)}) (15) {$5_2$};

 \path[every node/.style={font=\sffamily}]
(2) edge (1)
(2) edge (3)
(5) edge (4)
(5) edge (6)
(8) edge (7)
(8) edge (9)
(11) edge (10)
(11) edge (12)
(14) edge (13)
(14) edge (15)
(1) edge (6)
(1) edge (9)
(1) edge (12)
(1) edge (15)
(4) edge (3)
(4) edge (9)
(4) edge (12)
(4) edge (15)
(7) edge (3)
(7) edge (6)
(7) edge (12)
(7) edge (15)
(10) edge (3)
(10) edge (6)
(10) edge (9)
(10) edge (15)
(13) edge (3)
(13) edge (6)
(13) edge (9)
(13) edge (12);
\end{tikzpicture}
\\
$H_3$ & $H_4$ & $H_5$ 
\end{tabular}
\caption{Several graphs in the family $H_k$}\label{figH_k}
\end{center}
\end{figure}

Figure~\ref{figH_k} illustrates the graphs $H_k$ for $k=3,4,5$. (Note that our vertex labels for $H_k$ are different from those in~\cite{AuT24}.) The authors~\cite[Theorem 2]{AuT24} 
proved the following.

\begin{theorem}\label{thmHk}
For every $k \geq 3$, The $\LS_+$-rank of $H_k$ is at least $\frac{3k}{16}$.
\end{theorem}

Theorem~\ref{thmHk} (and other results in~\cite{AuT24}) ended a 17-year lull in new hardness results for $\LS_+$-relaxations of the stable set problem, and provides renewed hope that $\ell$-minimal graphs do exist for $\ell \geq 4$. Indeed, one of the main contributions of this work is the discovery of what we believe to be the first known instance of a $4$-minimal graph ($G_{4,1}$ in Figure~\ref{figG41}).

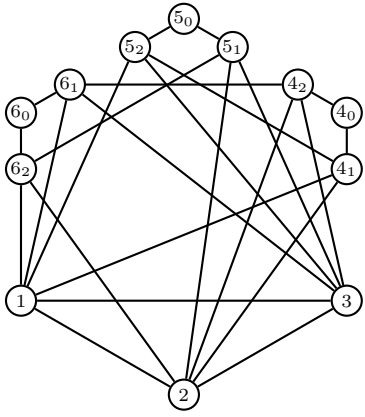
\begin{figure}[ht!]
\begin{center}
\def\x{270 - 360/6}
\def\z{360/6}
\def\y{0.7}
\def\sc{2.5}

\begin{tikzpicture}[scale=\sc, thick,main node/.style={circle, minimum size=4mm, inner sep=0.1mm,draw,font=\tiny\sffamily}]

\node[main node] at ({cos(\x+(0)*\z)},{sin(\x+(0)*\z)}) (1) {$1$};
\node[main node] at ({cos(\x+(1)*\z)},{sin(\x+(1)*\z)}) (2) {$2$};
\node[main node] at ({cos(\x+(2)*\z)},{sin(\x+(2)*\z)}) (3) {$3$};

\node[main node] at ({ \y* cos(\x+(3)*\z) + (1-\y)*cos(\x+(2)*\z)},{ \y* sin(\x+(3)*\z) + (1-\y)*sin(\x+(2)*\z)}) (4) {$4_1$};
\node[main node] at ({cos(\x+(3)*\z)},{sin(\x+(3)*\z)}) (5) {$4_0$};
\node[main node] at ({ \y* cos(\x+(3)*\z) + (1-\y)*cos(\x+(4)*\z)},{ \y* sin(\x+(3)*\z) + (1-\y)*sin(\x+(4)*\z)}) (6) {$4_2$};

\node[main node] at ({ \y* cos(\x+(4)*\z) + (1-\y)*cos(\x+(3)*\z)},{ \y* sin(\x+(4)*\z) + (1-\y)*sin(\x+(3)*\z)}) (7) {$5_1$};
\node[main node] at ({cos(\x+(4)*\z)},{sin(\x+(4)*\z)}) (8) {$5_0$};
\node[main node] at ({ \y* cos(\x+(4)*\z) + (1-\y)*cos(\x+(5)*\z)},{ \y* sin(\x+(4)*\z) + (1-\y)*sin(\x+(5)*\z)}) (9) {$5_2$};

\node[main node] at ({ \y* cos(\x+(5)*\z) + (1-\y)*cos(\x+(4)*\z)},{ \y* sin(\x+(5)*\z) + (1-\y)*sin(\x+(4)*\z)}) (10) {$6_1$};
\node[main node] at ({cos(\x+(5)*\z)},{sin(\x+(5)*\z)}) (11) {$6_0$};
\node[main node] at ({ \y* cos(\x+(5)*\z) + (1-\y)*cos(\x+(6)*\z)},{ \y* sin(\x+(5)*\z) + (1-\y)*sin(\x+(6)*\z)}) (12) {$6_2$};

 \path[every node/.style={font=\sffamily}]
(1) edge (2)
(1) edge (3)
(2) edge (3)
(4) edge (5)
(5) edge (6)
(7) edge (8)
(8) edge (9)
(10) edge (11)
(11) edge (12)
(4) edge (9)
(7) edge (12)
(10) edge (6)
(1) edge (4)
(2) edge (4)
(2) edge (6)
(3) edge (6)
(2) edge (7)
(3) edge (7)
(3) edge (9)
(1) edge (9)
(3) edge (10)
(1) edge (10)
(1) edge (12)
(2) edge (12);
\end{tikzpicture}
\caption{$G_{4,1}$, a $12$-vertex graph with $\LS_+$-rank $4$}\label{figG41}
\end{center}
\end{figure}

This paper is organized as follows. In Section~\ref{sec2}, we define the $\LS_+$ operator and introduce some of the tools and notation we will need for our subsequent analysis. Then, in Section~\ref{sec3}, we discuss what we call \emph{star-homomorphism} between graphs, and provide a template for constructing graph operations that are $\LS_+$-rank non-decreasing. Using this template, we define our \emph{vertex-stretching} operation in Section~\ref{sec4}, which generalizes similar graph operations studied previously~\cite{LiptakT03, AguileraEF14, BianchiENT17}. We then show in Section~\ref{sec5} that every $\ell$-minimal graph for $\ell \geq 2$ must be obtained from applying our vertex-stretching operation to a smaller graph, and in particular study the $\LS_+$-ranks of graphs obtained from stretching the vertices of a complete graph. In Section~\ref{sec6}, we prove that $G_{4,1}$ is indeed $4$-minimal and discuss some of the immediate consequences of the result, which includes the discovery of several other new $3$- and $4$-minimal graphs. We do remark that while the discovery of $G_{4,1}$ is largely guided by the structural results for $\ell$-minimal graphs presented in Sections $3$ to $5$, these results are not needed for the proof of $G_{4,1}$ having $\LS_+$-rank $4$. After that, we revisit the aforementioned families of graphs $H_k$ in Section~\ref{sec7}, and apply our results on vertex stretching to show that there exists a family of graphs $G$ with maximum degree $3$ and $\LS_+$-rank $\Omega(\sqrt{|V(G)|})$. Finally, we conclude our paper in Section~\ref{sec8} by mentioning some natural future research directions.

\section{Preliminaries}\label{sec2}

In this section, we define the lift-and-project operator $\LS_+$ due to Lov{\'a}sz and Schrijver~\cite{LovaszS91} and the convex relaxations of $\STAB(G)$ it produces, as well as go over the basic tools we will use in subsequent sections to analyze the $\LS_+$-rank of graphs.

\subsection{The $\LS_+$-operator}

Given a set $P \subseteq [0,1]^n$, we define the \emph{homogenized cone} of $P$ to be
\[
\cone(P) \ce \set{ \begin{bmatrix} \lambda \\ \lambda x \end{bmatrix} : \ld \geq 0, x \in P}.
\]
Notice that $\cone(P) \subseteq \mR^{n+1}$, and we will refer to the new coordinate with index $0$. Next, given a vector $x$ and an index $i$, we may refer to the $i$-entry in $x$ by $x_i$ or $[x]_i$. All vectors are column vectors by default, so $x^{\top}$, the transpose of a vector $x$, is a row vector. Next, let $\mathbb{S}_+^n$ denote the set of $n$-by-$n$ real symmetric positive semidefinite matrices, and $\diag(Y)$ be the vector formed by the diagonal entries of a square matrix $Y$. We also let $e_i$ be the $i^{\tn{th}}$ unit vector.

Given $P \subseteq [0,1]^n$, the operator $\LS_+$ first \emph{lifts} $P$ to the following set of matrices:
\[
\widehat{\LS}_+(P) \ce \set{ Y \in \mS_+^{n+1} : Ye_0 = \diag(Y), Ye_i, Y(e_0-e_i) \in \cone(P)~\forall i \in [n] }.
\]
It then \emph{projects} the set back down to the following set in $\mR^n$:
\[
\LS_+(P) \ce \set{ x \in \mR^n : \exists Y \in \widehat{\LS}_+(P), Ye_0 = \begin{bmatrix} 1 \\ x \end{bmatrix}}.
\]
Given $x \in \LS_+(P)$, we say that $Y \in \widehat{\LS}_+(P)$ is a \emph{certificate matrix} for $x$ if $Ye_0 = \begin{bmatrix} 1 \\ x \end{bmatrix}$. Also, given a set $P \subseteq [0,1]^n$, we define
\[
P_I \ce \conv\left( P \cap \set{0,1}^n \right)
\]
to be the \emph{integer hull} of $P$. The following foundational property of $\LS_+$ was shown in~\cite{LovaszS91}.

\begin{lemma}\label{lemLS+}
For every set $P \subseteq [0,1]^n$, $P_I \subseteq \LS_+(P) \subseteq P$.
\end{lemma}

\begin{proof}
For the first containment, let $x \in P \cap \set{0,1}^n$. Observe that $Y \ce \begin{bmatrix} 1 \\ x \end{bmatrix}\begin{bmatrix} 1 \\ x \end{bmatrix}^{\top} \in \widehat{\LS}_+(P)$, and so $x \in \LS_+(P)$. For the second containment, let $x \in \LS_+(P)$, and $Y \in \widehat{\LS}_+(P)$ be a certificate matrix for $x$. Since $Ye_0 = Ye_i + Y(e_0 - e_i)$ for any index $i \in [n]$ and that $\widehat{\LS}_+$ imposes that $Ye_i, Y(e_0 - e_i) \in \cone(P)$, it follows that $Ye_0 = \begin{bmatrix} 1 \\ x \end{bmatrix} \in \cone(P)$, and thus $x \in P$. 
\end{proof}

Therefore, $\LS_+(P)$ contains the same set of integral solutions as $P$. Also, if $P$ is tractable, then so is $\LS_+(P)$, and it is known that $P \supset \LS_+(P)$ unless $P = P_I$. Thus, $\LS_+(P)$ offers a tractable relaxation of $P_I$ that is tighter than the initial relaxation $P$.

Furthermore, we can apply $\LS_+$ multiple times to obtain yet tighter relaxations. Given $k \in \mN$, let $\LS_+^k(P)$ be the set obtained from applying $k$ successive $\LS_+$ operations to $P$. (We also let $\LS_+^0(P) \ce P$.) Then it is well-known that
\[
P_I = \LS_+^n(P) \subseteq \LS_+^{n-1}(P) \subseteq \cdots \subseteq \LS_+(P) \subseteq P.
\]
Thus, $\LS_+$ generates a hierarchy of progressively tighter convex relaxations which converge to $P_I$ in no more than $n$ iterations. The reader may refer to Lov{\'a}sz and Schrijver~\cite{LovaszS91} for a proof of this fact and some other properties of the $\LS_+$ operator.

\subsection{Analyzing the $\LS_+$-rank of a graph}

Recall that $\FRAC(G)$, the fractional stable set polytope of a graph $G$, offers a simple and tractable convex relaxation of $\STAB(G)$. Thus, we could apply $\LS_+$ to obtain stronger relaxations of $\STAB(G)$ than $\FRAC(G)$. Given an integer $k \geq 0$, define
\[
\LS_+^k(G) \ce \LS_+^k(\FRAC(G)),
\]
and let $r_+(G)$ denote the $\LS_+$-rank of $G$ (which, again, is the smallest integer $k$ where $\LS_+^k(G) = \STAB(G)$).

To show that a graph $G$ has $\LS_+$-rank at least $p$, the standard approach is to find a point $\bar{x}$ where $\bar{x} \not\in \STAB(G)$ and $\bar{x} \in \LS_+^{p-1}(G)$ --- this is the approach we will take when verifying that $r_+(G_{4,1}) \geq 4$. We do remark that verifying $\bar{x} \in \LS_+^{p-1}(G)$ tends to get progressively more challenging as $p$ increases, unless the symmetries of $G$ allow for an inductive argument (which is the case for the line graphs of odd cliques~\cite{StephenT99}, and to a lesser extent for $H_k$ and related graphs~\cite{AuT24}). Given $p \in \mN$, we also define
\[
\a_{\LS_+^p}(G) \ce \max \set{\bar{e}^{\top}x : x \in \LS_+^p(G)},
\]
where $\bar{e}$ denotes the vector of all-ones. Notice that if $\a_{\LS_+^p}(G) > \a(G)$, then $r_+(G) \geq p+1$. 

Next, the following is a well-known property of $\LS_+$.

\begin{lemma}\label{lemfacet}
Let $P \subseteq [0,1]^n$ be a polyhedron, and $F$ be a face of $[0,1]^n$. Then
\[
\LS_+(P \cap F) = \LS_+(P) \cap F.
\]
\end{lemma}

It follows from Lemma~\ref{lemfacet} that if $\bar{x} \in \LS_+^p(G)$ and $G'$ is an induced subgraph of $G$, then the vector obtained from $\bar{x}$ by removing entries not in $V(G')$ is in $\LS_+^p(G')$. This in turn implies that $r_+(G') \leq r_+(G)$ --- see, for instance,~\cite[Lemma 5]{AuT24} for a proof.

We next mention several other ways of bounding $r_+(G)$ using the $\LS_+$-rank of graphs that are related to $G$. Given a graph $G$ and $S \subseteq V(G)$, we let $G-S$ denote the subgraph of $G$ induced by the vertices $V(G) \setminus S$, and call $G-S$ the graph obtained by the \emph{deletion} of $S$. (When $S = \set{i}$ for some vertex $i$, we simply write $G-i$ instead of $G - \set{i}$.) Next, given $i \in V(G)$, let 
\[
\Gamma_G(i) \ce \set{ j \in V(G) : \set{i,j} \in E(G)}
\]
be the \emph{open neighborhood} of $i$ in $G$, and $\Gamma_G[i] \ce \Gamma_G(i) \cup \set{i}$ be the \emph{closed neighborhood} of $i$ in $G$. Then the graph obtained from the \emph{destruction} of $i$ in $G$ is defined as
\[
G \ominus i \ce G - \Gamma[i].
\]
Then we have the following.
 
\begin{theorem}\label{thmDeleteDestroy}
For every graph $G$,
\begin{itemize}
\item[(i)]
\cite[Corollary 2.16]{LovaszS91} $r_+(G) \leq \max \set{ r_+(G \ominus i) : i \in V(G) } + 1$;
\item[(ii)]
\cite[Theorem 36]{LiptakT03} $r_+(G) \leq \min \set{ r_+(G - i) : i \in V(G) } + 1$.
\end{itemize}
\end{theorem}
 
Recall that $r_+(G) = 0$ if and only if $G$ is bipartite (in which case $\FRAC(G) = \STAB(G)$). Thus, it follows immediately from Theorem~\ref{thmDeleteDestroy}(ii) that if $G$ is non-bipartite but $G- i$ is bipartite for some $i \in V(G)$, then $r_+(G) =1$ --- an example for such graphs is the odd cycles. Likewise, if $G$ is non-bipartite while $G \ominus i$ is bipartite for every $i \in V(G)$, then $r_+(G)=1$ as well --- such as when $G$ is an odd antihole (i.e., the graph complement of an odd cycle of length at least $5$), or an odd wheel (i.e., the graph obtained from joining a vertex to every vertex of an odd cycle of length at least $5$).
 
We say that a graph $G$ is \emph{perfect} if $\CLIQ(G) = \STAB(G)$. In terms of forbidden subgraphs, $G$ is perfect if and only if it does not contain an induced subgraph that is an odd hole (i.e., an odd cycle of length at least $5$) or an odd antihole~\cite{ChudnovskyRST06}. Since $\LS_+(G) \subseteq \CLIQ(G)$ in general~\cite{LovaszS91}, it follows that $r_+(G) \leq 1$ if $G$ is perfect.

The following is a restatement of~\cite[Lemma 5]{LiptakT03}.

\begin{proposition}\label{propCliqueCut}
Let $G$ be a graph, and $S_1, S_2, C \subseteq V(G)$ are mutually disjoint subsets such that
\begin{itemize}
\item
$S_1 \cup S_2 \cup C = V(G)$;
\item
$C$ induces a clique in $G$;
\item
There is no edge $\set{i,j} \in E(G)$ where $i\in S_1, j \in S_2$.
\end{itemize}
Then $r_+(G) = \max \set{ r_+(G-S_1), r_+(G-S_2)}$.
\end{proposition}

Thus, if $G$ has a cut clique (i.e., a clique $C$ where $G-C$ has multiple components), then the $\LS_+$-rank of $G$ is equal to that of one of its proper subgraphs.

Next, observe that $\FRAC(G)$ is \emph{lower-comprehensive} for every graph $G$ (i.e., if $x \in \FRAC(G)$, then $y \in \FRAC(G)$ for all vectors $y$ where $0 \leq y \leq x$). Also, one can show (such as using Lemma~\ref{lemfacet}) that $\LS_+$ preserves lower-comprehensiveness. Thus, it follows that $\LS_+^k(G)$ is lower-comprehensive as well for every $k \in \mN$. 

Finally, in addition to preserving lower-comprehensiveness, it is also clear from the definition of $\LS_+$ that the operator also preserves containment (i.e., if $P_1 \subseteq P_2$, then $\LS_+(P_1) \subseteq \LS_+(P_2)$). This implies the following.

\begin{lemma}\label{lem05subgraph}
Given graphs $G, H$ where $V(G) = V(H)$ and $E(G) \subseteq E(H)$,
\begin{itemize}
\item[(i)]
If $a^{\top}x \leq \b$ is valid for $\LS_+^p(G)$, then $a^{\top}x \leq \b$ is valid for $\LS_+^p(H)$;
\item[(ii)]
If $a^{\top}x \leq \b$ is not valid for $\LS_+^p(H)$, then $a^{\top}x \leq \b$ is not valid for $\LS_+^p(G)$.
\end{itemize}
\end{lemma}

\section{Star-homomorphic graphs}\label{sec3}

In this section, we introduce the notion of two graphs being star-homomorphic, and describe how the $\LS_+$-relaxations of such a pair of graphs are related. Given a graph $G = (V(G), E(G))$, we define the graph $G^{\dagger}$ where
\begin{align*}
V(G^{\dagger}) &\ce \set{i, \overline{i} : i \in V(G)},\\
E(G^{\dagger}) &\ce E(G) \cup \set{ \set{i, \overline{i}} : i \in V(G)}.
\end{align*}
In other words, we obtain $G^{\dagger}$ from $G$ by adding a new vertex $\overline{i}$ for every $i \in V(G)$, and then adding an edge between $\overline{i}$ and $i$. Figure~\ref{figGstar} provides an example of constructing $G^{\dagger}$ from $G$.

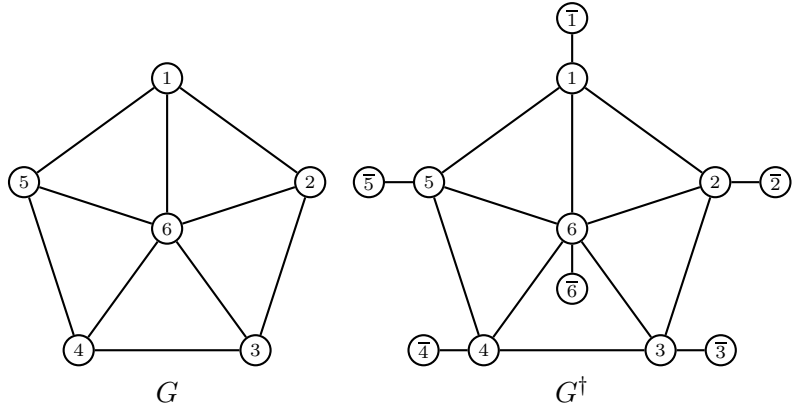
\begin{figure}[ht!]
\begin{center}

\def\sc{2}
\def\w{0.4}
\def\x{90}
\def\y{0.70}
\def\z{-360/5}
\begin{tabular}{cc}

\begin{tikzpicture}[scale=\sc, thick,main node/.style={circle, minimum size=4mm, inner sep=0.1mm,draw,font=\tiny\sffamily}]
\node[main node] at ({cos(\x+(0)*\z)},{sin(\x+(0)*\z)}) (1) {$1$}; 
\node[main node] at ({cos(\x+(1)*\z)},{sin(\x+(1)*\z)}) (2) {$2$};
 \node[main node] at ({cos(\x+(2)*\z)},{sin(\x+(2)*\z)}) (3) {$3$}; 
 \node[main node] at ({cos(\x+(3)*\z)},{sin(\x+(3)*\z)}) (4) {$4$}; 
 \node[main node] at ({cos(\x+(4)*\z)},{sin(\x+(4)*\z)}) (5) {$5$}; 
 \node[main node] at (0,0) (6) {$6$};

 \path[every node/.style={font=\sffamily}]
(1) edge (2)
(2) edge (3)
(3) edge (4)
(4) edge (5)
(5) edge (1)
(6) edge (1)
(6) edge (2)
(6) edge (3)
(6) edge (4)
(6) edge (5);
\end{tikzpicture}
&

\begin{tikzpicture}[scale=\sc, thick,main node/.style={circle, minimum size=4mm, inner sep=0.1mm,draw,font=\tiny\sffamily}]
\node[main node] at ({cos(\x+(0)*\z)},{sin(\x+(0)*\z)}) (1) {$1$}; 
\node[main node] at ({cos(\x+(1)*\z)},{sin(\x+(1)*\z)}) (2) {$2$};
 \node[main node] at ({cos(\x+(2)*\z)},{sin(\x+(2)*\z)}) (3) {$3$}; 
 \node[main node] at ({cos(\x+(3)*\z)},{sin(\x+(3)*\z)}) (4) {$4$}; 
 \node[main node] at ({cos(\x+(4)*\z)},{sin(\x+(4)*\z)}) (5) {$5$};
 \node[main node] at ({cos(\x+(0)*\z)},{sin(\x+(0)*\z)+\w}) (1b) {$\overline{1}$};  
  \node[main node] at ({cos(\x+(1)*\z)+\w},{sin(\x+(1)*\z)}) (2b) {$\overline{2}$};  
   \node[main node] at ({cos(\x+(2)*\z)+\w},{sin(\x+(2)*\z)}) (3b) {$\overline{3}$};  
    \node[main node] at ({cos(\x+(3)*\z)-\w},{sin(\x+(3)*\z)}) (4b) {$\overline{4}$};  
     \node[main node] at ({cos(\x+(4)*\z)-\w},{sin(\x+(4)*\z)}) (5b) {$\overline{5}$};  
     
 \node[main node] at (0,0) (6) {$6$};
  \node[main node] at (0,{-\w}) (6b) {$\overline{6}$};

 \path[every node/.style={font=\sffamily}]
(1) edge (2)
(2) edge (3)
(3) edge (4)
(4) edge (5)
(5) edge (1)
(6) edge (1)
(6) edge (2)
(6) edge (3)
(6) edge (4)
(6) edge (5)
(1) edge (1b)
(2) edge (2b)
(3) edge (3b)
(4) edge (4b)
(5) edge (5b)
(6) edge (6b);
\end{tikzpicture}\\
$G$ & $G^{\dagger}$
\end{tabular}
\caption{Constructing $G^{\dagger}$ from $G$}\label{figGstar}
\end{center}
\end{figure}
Also, given graphs $G$ and $H$, we say that $g : V(H) \to V(G)$ is a \emph{homomorphism} if, for all $i,j \in V(H)$,
\[
\set{i,j} \in E(H) \Rightarrow \set{g(i), g(j)} \in E(G).
\]
Furthermore, given graphs $G$ and $H$, if there exists a homomorphism $g : V(H) \to V(G^{\dagger})$, then we say that $H$ is  \emph{star-homomorphic to $G$ under $g$}.

\def\sc{2}
\def\w{0.4}
\def\x{90}
\def\y{0.70}
\def\z{-360/5}

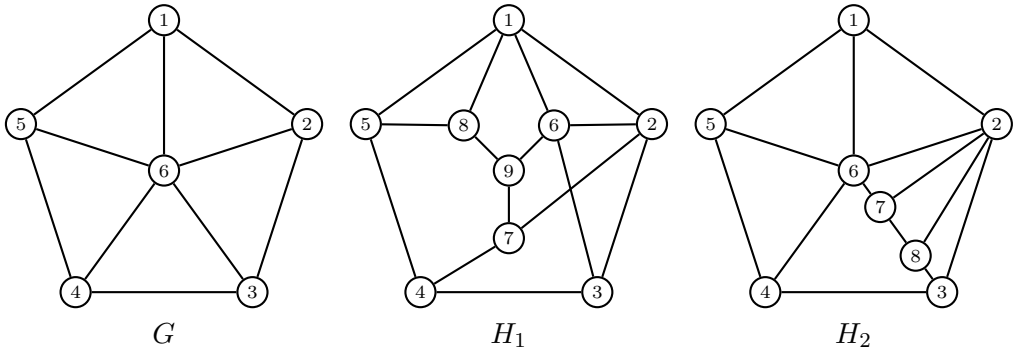
\begin{figure}[ht!]
\begin{center}
\begin{tabular}{ccc}

\begin{tikzpicture}[scale=\sc, thick,main node/.style={circle, minimum size=4mm, inner sep=0.1mm,draw,font=\tiny\sffamily}]
\node[main node] at ({cos(\x+(0)*\z)},{sin(\x+(0)*\z)}) (1) {$1$}; 
\node[main node] at ({cos(\x+(1)*\z)},{sin(\x+(1)*\z)}) (2) {$2$};
 \node[main node] at ({cos(\x+(2)*\z)},{sin(\x+(2)*\z)}) (3) {$3$}; 
 \node[main node] at ({cos(\x+(3)*\z)},{sin(\x+(3)*\z)}) (4) {$4$}; 
 \node[main node] at ({cos(\x+(4)*\z)},{sin(\x+(4)*\z)}) (5) {$5$}; 
 \node[main node] at (0,0) (6) {$6$};

 \path[every node/.style={font=\sffamily}]
(1) edge (2)
(2) edge (3)
(3) edge (4)
(4) edge (5)
(5) edge (1)
(6) edge (1)
(6) edge (2)
(6) edge (3)
(6) edge (4)
(6) edge (5);
\end{tikzpicture}
&

\begin{tikzpicture}[scale=\sc, thick,main node/.style={circle, minimum size=4mm, inner sep=0.1mm,draw,font=\tiny\sffamily}]
\node[main node] at ({cos(\x+(0)*\z)},{sin(\x+(0)*\z)}) (1) {$1$}; 
\node[main node] at ({cos(\x+(1)*\z)},{sin(\x+(1)*\z)}) (2) {$2$};
 \node[main node] at ({cos(\x+(2)*\z)},{sin(\x+(2)*\z)}) (3) {$3$}; 
 \node[main node] at ({cos(\x+(3)*\z)},{sin(\x+(3)*\z)}) (4) {$4$}; 
 \node[main node] at ({cos(\x+(4)*\z)},{sin(\x+(4)*\z)}) (5) {$5$}; 

 \node[main node] at (0.3,0.3) (6) {$6$};
  \node[main node] at (0,-0.45) (7) {$7$};
   \node[main node] at (-0.3,0.3) (8) {$8$};

 \node[main node] at (0,0) (9) {$9$};

 \path[every node/.style={font=\sffamily}]
(1) edge (2)
(2) edge (3)
(3) edge (4)
(4) edge (5)
(5) edge (1)
(6) edge (9)
(7) edge (9)
(8) edge (9)
(6) edge (1)
(6) edge (2)
(6) edge (3)
(7) edge (2)
(7) edge (4)
(8) edge (1)
(8) edge (5);
\end{tikzpicture}
&

\begin{tikzpicture}[scale=\sc, thick,main node/.style={circle, minimum size=4mm, inner sep=0.1mm,draw,font=\tiny\sffamily}]
\node[main node] at ({cos(\x+(0)*\z)},{sin(\x+(0)*\z)}) (1) {$1$}; 
\node[main node] at ({cos(\x+(1)*\z)},{sin(\x+(1)*\z)}) (2) {$2$};
 \node[main node] at ({cos(\x+(2)*\z)},{sin(\x+(2)*\z)}) (3) {$3$}; 
 \node[main node] at ({cos(\x+(3)*\z)},{sin(\x+(3)*\z)}) (4) {$4$}; 
 \node[main node] at ({cos(\x+(4)*\z)},{sin(\x+(4)*\z)}) (5) {$5$}; 
 \node[main node] at (0,0) (6) {$6$};
 \node[main node] at ({0.3*cos(\x+(2)*\z)},{0.3*sin(\x+(2)*\z)}) (7) {$7$}; 
  \node[main node] at ({0.7*cos(\x+(2)*\z)},{0.7*sin(\x+(2)*\z)}) (8) {$8$}; 

 \path[every node/.style={font=\sffamily}]
(1) edge (2)
(2) edge (3)
(3) edge (4)
(4) edge (5)
(5) edge (1)
(6) edge (1)
(6) edge (2)
(6) edge (7)
(7) edge (8)
(8) edge (3)
(7) edge (2)
(8) edge (2)
(6) edge (4)
(6) edge (5);
\end{tikzpicture}\\
$G$ & $H_1$ & $H_2$ 
\end{tabular}
\caption{Two graphs $H_1,H_2$ that are star-homomorphic to $G$}\label{figStarHomo}
\end{center}
\end{figure}

\begin{example}\label{egStarHomo}
Consider the graphs in Figure~\ref{figStarHomo}. Then $H_1$ is star-homomorphic to $G$ under $g_1$ where
\[
\begin{array}{r|c|c|c|c|c|c|c|c|c}
j \in V(H_1) & 1 & 2 & 3 & 4 & 5 & 6 & 7 & 8 & 9 \\
\hline
g_1(j) \in V(G^{\dagger}) & 1  & 2 & 3 & 4 & 5 & 6 & 6 & 6 & \overline{6} 
\end{array}
\]
Likewise, $H_2$ is star-homomorphic to $G$ under $g_2$ where
\[
\begin{array}{r|c|c|c|c|c|c|c|c}
j \in V(H_2) & 1 & 2 & 3 & 4 & 5 & 6 & 7 & 8 \\
\hline
g_2(j) \in V(G^{\dagger}) & 1  & 2 & 3 & 4 & 5 & 6 & 3 & 6 
\end{array}
\]
\end{example}

Given that $H$ is star-homomorphic to $G$ under $g$ and $x \in \mR^{V(G)}$, we let $\tilde{g}(x) \in \mR^{V(H)}$ be the vector where
\begin{equation}\label{eq_tildeg}
\left[\tilde{g}(x) \right]_j \ce \begin{cases}
x_i & \tn{if $g(j) = i$;}\\
1 - x_i & \tn{if $g(j) = \overline{i}$.}
\end{cases}
\end{equation}

\begin{example}\label{egStarHomo2}
To illustrate the function $\tilde{g}$, consider again the graphs $G$ and $H_1$ from Figure~\ref{figStarHomo}, as well as the star-homomorphism $g_1 : V(H_1) \to V(G^{\dagger})$ given in Example~\ref{egStarHomo}. Now let $x \ce (\frac{2}{5}, \frac{2}{5}, \frac{2}{5}, \frac{2}{5}, \frac{2}{5}, \frac{1}{5})^{\top} \in \mR^{V(G)}$. Then $\tilde{g}_1(x) \in \mR^{V(H_1)}$ where
\[
\begin{array}{r|c|c|c|c|c|c|c|c|c}
j \in V(H_1) & 1 & 2 & 3 & 4 & 5 & 6 & 7 & 8 & 9 \\
\hline
[ \tilde{g}_1(x)]_j  & \frac{2}{5}  & \frac{2}{5} & \frac{2}{5} &\frac{2}{5} & \frac{2}{5} & \frac{1}{5} & \frac{1}{5} & \frac{1}{5} & \frac{4}{5} 
\end{array}
\]
\end{example}

Notice the point $x$ given in Example~\ref{egStarHomo2} belongs to $\FRAC(G)$, and that $\tilde{g}_1$ maps $x$ to a point that belongs to $\FRAC(H_1)$. This is not a coincidence, but rather a general property of $\tilde{g}$ which follows readily from the definition of star-homomorphism:

\begin{lemma}\label{lemStarHomo}
Suppose $H$ is star-homomorphic to $G$ under $g$, and $x \in \mR^{V(G)}$. If $x \in \FRAC(G)$, then $\tilde{g}(x) \in \FRAC(H)$.
\end{lemma}

\begin{proof}
First, it is easy to see that $x \in [0,1]^{V(G)}$ implies $\tilde{g}(x) \in [0,1]^{V(H)}$. Now given edge $\set{j_1. j_2} \in E(H)$, $\set{ g(j_1), g(j_2)}$ is either an edge in $G$ or $\set{i, \overline{i}}$ for some $i \in V(G)$. In both cases, we see that $[\tilde{g}(x)]_{j_1} +[\tilde{g}(x)]_{j_2} \leq 1$.
\end{proof}

In fact, the implication in Lemma~\ref{lemStarHomo} is preserved under applications of $\LS_+$.

\begin{proposition}\label{propStarHomo}
Suppose $H$ is star-homomorphic to $G$ under $g$, and $x \in \mR^{V(G)}$. If $x \in \LS_+^p(G)$, then $\tilde{g}(x) \in \LS_+^p(H)$.
\end{proposition}

\begin{proof}
Suppose $x \in \LS_+^p(G)$. We prove that $\tilde{g}(x) \in \LS_+^p(H)$ by induction on $p$. The base case $p=0$ reduces to Lemma~\ref{lemStarHomo}. Next, suppose $p \geq 1$, and let $Y \in \widehat{\LS}_+^p(G)$ be a certificate matrix. For convenience, we also extend the function $\tilde{g}$ as follows: Given a real number $k \geq 0$, define $\tilde{g}_k : \mR^{\set{0} \cup V(G)} \to \mR^{\set{0} \cup V(H)}$ such that
\[
\left[\tilde{g}_k(x) \right]_j \ce \begin{cases}
x_0 & \tn{if $j=0$;}\\
x_i & \tn{if $j \in V(H)$ and $g(j) = i$;}\\
k - x_i & \tn{if $j \in V(H)$ and $g(j) = \overline{i}$.}
\end{cases}
\]
Notice that the function $\tilde{g}_k$ satisfies
\[
\tilde{g}_{\ld}\left( \begin{bmatrix} \ld \\ \ld x \end{bmatrix} \right) =  \begin{bmatrix} \ld \\ \ld \tilde{g}(x) \end{bmatrix}
\]
for every $\ld \geq 0$ and $x \in \mR^{V(G)}$. 

Since $Y \in \widehat{\LS}_+^p(G)$, $Ye_i, Y(e_0-e_i) \in \cone\left(\LS_+^{p-1}(G)\right)$ for every $i \in V(G)$. Thus, applying the inductive hypothesis we have $\tilde{g}_{x_i}(Ye_i), \tilde{g}_{1-x_i}(Y(e_0-e_i)) \in \cone\left(\LS_+^{p-1}(H)\right)$.

Next, define the matrix $U \in \mR^{ (\set{0} \cup V(G)) \times (\set{0} \cup V(H))}$ where
\[
Ue_j \ce 
\begin{cases}
e_0 & \tn{ if $j = 0$;}\\
e_{i} &\tn{if $j \in V(H)$ and $g(j) = i$;}\\
e_0 - e_{i} &\tn{if $j \in V(H)$ and $g(j) = \overline{i}$.}
\end{cases}
\]
Then, given $z \in \mR^{\set{0} \cup V(G)}$,  $U^{\top}z = \tilde{g}_{z_0}(z)$. 

Next, we claim that the matrix $Y' \ce U^{\top} Y U \in  \widehat{\LS}_+^p(H)$. First, since $Y = Y^{\top}, \diag(Y) = Ye_0$, and $Y \succeq 0$, it is easy to see that the corresponding properties also hold for $Y'$. We next show that $Y'e_j, Y'(e_0-e_j) \in \cone\left(\LS_+^{p-1}(H)\right)$ for every $j \in V(H)$. First, if $g(j) = i \in V(G)$, then $Y'e_j = \tilde{g}_{x_i}(Ye_i)$, and 
\[
Y'(e_0- e_j) =Y'e_0 - Y'e_j = \tilde{g}_1(Ye_0) - \tilde{g}_{x_i}(Ye_i) = \tilde{g}_{1-x_i}(Y(e_0-e_i)).
\]
To see the last equality, notice that 
\begin{align*}
&\left[ \tilde{g}_1(Ye_0) - \tilde{g}_{x_i}(Ye_i)\right]_{\ell} = \left[ \tilde{g}_{1-x_i}(Y(e_0-e_i))\right]_{\ell} \\
= &
\begin{cases}
1-x_i & \tn{ if $\ell = 0$;}\\
x_{i'} - Y[i',i] &\tn{ if $g(\ell) = i' \in V(G)$;}\\
1 - x_{i'} - x_i + Y[i',i] &\tn{ if $g(\ell) = \overline{i'}$ for some $i' \in V(G)$.}
\end{cases}
\end{align*}
Likewise, if $g(j) = \overline{i}$ for some $i \in V(G)$, then $Y'e_j = \tilde{g}_{1-x_i}(Y(e_0- e_i))$ and $Y'(e_0- e_j) = \tilde{g}_{x_i}(Ye_i)$. In all cases, it follows from the inductive hypothesis that $Y'e_j, Y'(e_0-e_j) \in \cone\left(\LS_+^{p-1}(H)\right)$. Therefore, $Y' \in \widehat{\LS}_+^p(H)$. Since $Y'e_0 = \tilde{g}_1(Ye_0) = \begin{bmatrix} 1 \\ \tilde{g}(x) \end{bmatrix}$, it follows that $\tilde{g}(x) \in \LS_+^p(H)$.
\end{proof}

Proposition~\ref{propStarHomo} helps establish a framework for bounding the $\LS_+$-rank of a graph by that of another.

\begin{lemma}\label{lemStarHomo2}
Let $G$ and $H$ be graphs where $H$ is star-homomorphic to $G$ under $g$. Assume that for every $x \in \mR^{V(G)}$, the following holds:
 \[
x \not\in \STAB(G) \Rightarrow \tilde{g}(x) \not\in \STAB(H).
\]
Then $r_+(H) \geq r_+(G)$.
\end{lemma}

\begin{proof}
Suppose $r_+(G) = p \geq 1$ (the claim is trivial if $p=0$). Then there exists $x \in \LS_+^{p-1}(G) \setminus \STAB(G)$. Then by the hypothesis and Proposition~\ref{propStarHomo}, $\tilde{g}(x) \in  \LS_+^{p-1}(H) \setminus \STAB(H)$, showing that $r_+(H) \geq p$.
\end{proof}

Finally, while our focus for this paper is the $\LS_+$ operator, we remark that the framework of star-homomorphism can be extended to analyze relaxations generated by other lift-and-project operators. Again, let $H$ be a graph that is star-homomorphic to $G$ under $g$. Then notice that the function $\tilde{g} : \mR^{V(G)} \to \mR^{V(H)}$ can be expressed as a composition of the following four elementary operations:
\begin{enumerate}
\item
Deleting a coordinate. E.g., $L : \mR^n \to \mR^{n-1}$ where
\[
L(x_1, x_2,  x_3, \ldots, x_n) = (x_2,  x_3, \ldots, x_n).
\]
\item
Swapping two coordinates. E.g., $L : \mR^n \to \mR^{n}$ where
\[
L(x_1, x_2, x_3, \ldots, x_n) = (x_2,x_1, x_3, \ldots, x_n).
\]
\item
Cloning a coordinate. E.g., $L : \mR^n \to \mR^{n+1}$ where
\[
L(x_1, x_2, \ldots, x_n) = (x_1,x_1, x_2, \ldots, x_n).
\]
\item
Flipping a coordinate. E.g., $L : \mR^n \to \mR^{n}$ where
\[
L(x_1, x_2, \ldots, x_n) = (1-x_1, x_2, \ldots, x_n).
\]
\end{enumerate}
Now, let $\mathcal{L}$ be a lift-and-project operator. If one can show that
\begin{equation}\label{eqpro0}
x \in \mathcal{L}(P) \Rightarrow L(x) \in \mathcal{L}(L(P)),
\end{equation}
for every function $L$ that belongs to one of the four categories above (where $L(P)$ denotes $\set{ L(z) : z \in P}$), then one can prove the analogous version of Proposition~\ref{propStarHomo} for $\mathcal{L}$. For instance, the ideas from the proof of Proposition~\ref{propStarHomo} show that~\eqref{eqpro0} holds for $\mathcal{L} \in \set{\LS_+, \LS, \LS_0}$ (where $\LS, \LS_0$~\cite{LovaszS91} are operators that generate polyhedral relaxations which are generally weaker than $\LS_+$). Also, it has been shown~\cite[Proposition 1]{AuT18} that the Lasserre operator $\Las$~\cite{Lasserre01} commutes with all automorphisms of the unit hypercube, a property that is also shared by the Sherali--Adams operator $\SA$~\cite{SheraliA90} and one of its PSD variants $\SA_+$~\cite{AuT16}. Thus, these operators satisfy~\eqref{eqpro0} as well, and much of what we show for $\LS_+$ in this section and the next section also applies to these operators.

\section{The (generalized) vertex-stretching operation}\label{sec4}

In this section, we introduce a graph operation that shows promise in producing relatively small graphs with high $\LS_+$-ranks, and study some of its properties. Given a graph $G$, vertex $v \in V(G)$, and non-empty sets $A_1, \ldots,  A_p \subset \Gamma_G(v)$ where $\bigcup_{\ell=1}^p A_{\ell}  = \Gamma_G(v)$, we define the \emph{stretching} of $v$ in $G$ by applying the following sequence of transformations to $G$: 
\begin{itemize}
\item
Replace $v$ by $p+1$ vertices: $v_0, v_1, \ldots, v_p$;
\item
For every $\ell \in [p]$, Join $v_{\ell}$ to $v_0$, as well as to all vertices in $A_{\ell}$.
\end{itemize}

We will also refer to the operation as \emph{$p$-stretching} when we would like to specify $p$ (which is necessarily at least $2$). For example, Figure~\ref{figstretching} shows the graph obtained from $2$-stretching vertex $5$ in $K_5$ (with $A_1 = \set{2,3,4}$ and $A_2 = \set{1,2,3}$). For another example, observe that in Figure~\ref{figStarHomo}, the graph $H_1$ can be obtained by $3$-stretching vertex $6$ in $G$.

\def\x{270 - 180/5}
\def\z{360/5}
\def\y{0.70}
\def\sc{2}

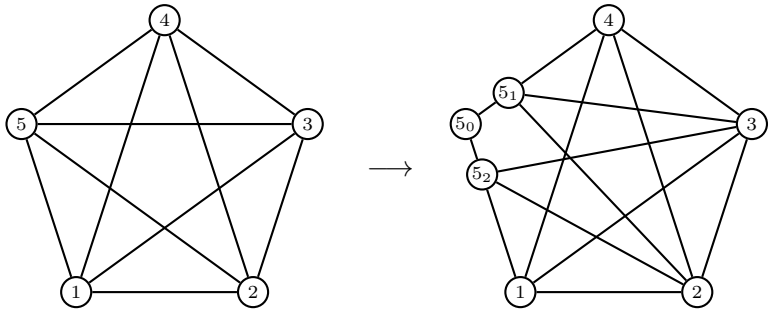
\begin{figure}[ht!]
\begin{center}
\begin{tabular}{cc}
\begin{tikzpicture}[scale=\sc, thick,main node/.style={circle, minimum size=4mm, inner sep=0.1mm,draw,font=\tiny\sffamily}]
\node[main node] at ({cos(\x+(0)*\z)},{sin(\x+(0)*\z)})  (1) {$1$};
\node[main node] at ({cos(\x+(1)*\z)},{sin(\x+(1)*\z)})  (2) {$2$};
\node[main node] at ({cos(\x+(2)*\z)},{sin(\x+(2)*\z)})  (3) {$3$};
\node[main node] at ({cos(\x+(3)*\z)},{sin(\x+(3)*\z)})  (4) {$4$};
\node[main node] at ({cos(\x+(4)*\z)},{sin(\x+(4)*\z)})  (5) {$5$};

 \path[every node/.style={font=\sffamily}]
(1) edge (2)
(1) edge (3)
(1) edge (4)
(2) edge (3)
(2) edge (4)
(3) edge (4)
(5) edge (2)
(5) edge (3)
(5) edge (4)
(5) edge (1);

\node at (1.5,0) () {$\longrightarrow$};
\end{tikzpicture}
&
\begin{tikzpicture}[scale=\sc, thick,main node/.style={circle, minimum size=4mm, inner sep=0.1mm,draw,font=\tiny\sffamily}]
\node[main node] at ({cos(\x+(0)*\z)},{sin(\x+(0)*\z)})  (1) {$1$};
\node[main node] at ({cos(\x+(1)*\z)},{sin(\x+(1)*\z)})  (2) {$2$};
\node[main node] at ({cos(\x+(2)*\z)},{sin(\x+(2)*\z)})  (3) {$3$};
\node[main node] at ({cos(\x+(3)*\z)},{sin(\x+(3)*\z)})  (4) {$4$};
\node[main node] at ({cos(\x+(4)*\z)},{sin(\x+(4)*\z)}) (6) {$5_0$};
\node[main node] at ({ \y* cos(\x+(4)*\z) + (1-\y)*cos(\x+(3)*\z)},{ \y* sin(\x+(4)*\z) + (1-\y)*sin(\x+(3)*\z)}) (5) {$5_1$};
\node[main node] at ({ \y* cos(\x+(4)*\z) + (1-\y)*cos(\x+(5)*\z)},{ \y* sin(\x+(4)*\z) + (1-\y)*sin(\x+(5)*\z)}) (7) {$5_2$};

 \path[every node/.style={font=\sffamily}]
(1) edge (2)
(1) edge (3)
(1) edge (4)
(2) edge (3)
(2) edge (4)
(3) edge (4)
(5) edge (2)
(5) edge (3)
(5) edge (4)
(7) edge (1)
(7) edge (2)
(7) edge (3)
(5) edge (6)
(6) edge (7);
\end{tikzpicture}
\end{tabular}
\end{center}
\caption{Demonstrating the vertex-stretching operation}\label{figstretching}
\end{figure}

We remark that our vertex stretching operation is a generalization of the type (i) stretching operation described in~\cite{LiptakT03} and later studied in~\cite{AguileraEF14} (which further requires that $p=2$ and $A_1 \cap A_2 = \emptyset$), as well as the $k$-stretching operation described in~\cite{BianchiENT17} (which further requires that $p=2$ and the vertices $A_1 \cap A_2$ induce a clique of size $k$ in $G$, with $k=0$ allowed). Also, when $A_1, \ldots, A_p$ are mutually disjoint and at most one of these $p$ sets has size greater than $1$, our vertex stretching specializes to an instance of type (ii) stretching from~\cite{LiptakT03}.

Given a graph $G$ and an integer $p \geq 2$, we define $\S_p(G)$ to be the set of graphs that can be obtained from $G$ by $p$-stretching one vertex. Notice that every graph $H \in \S_p(G)$ is star-homomorphic to $G$ under $g$ where
\begin{equation}\label{eqVertexStretchg}
g(j) \ce
\begin{cases}
j &\tn{if $j \in V(H \ominus v_0)$;}\\
v &\tn{if $j \in \set{v_1, \ldots, v_p}$;}\\
\overline{v} & \tn{if $j = v_0$.}
\end{cases}
\end{equation}
We also define $\S(G) \ce \bigcup_{p \geq 2} \S_p(G)$, and will show that $r_+(H) \geq r_+(G) $ for all $H \in \S(G)$ using Lemma~\ref{lemStarHomo2}. First, we need a tool that uses valid inequalities of $\STAB(G)$ to generate potential valid inequalities for $\STAB(H)$.

\begin{lemma}\label{lemVertexStretchFacet}
Let $H \in \S(G)$ be a graph obtained from $G$ by $p$-stretching vertex $v \in V(G)$, and let $a^{\top}x \leq \b$ be a valid inequality of $\STAB(G)$ where $a \geq 0$. If $d \in \mR_+^p$ satisfies $\sum_{\ell=1}^p d_{\ell} \geq a_v$ and
\begin{equation}\label{lemVertexStretchFaceteq0}
\max \set{ a^{\top}x : x \in \STAB(G), x_{i} = 0~\forall i \in \set{v} \cup \bigcup_{\ell \in T} A_\ell} \leq \b - a_v + \sum_{\ell \not\in T} d_{\ell}
\end{equation}
for all $\emptyset \subset T \subset [p]$, then
\begin{equation}\label{lemVertexStretchFaceteq1}
 \sum_{i \in V(H \ominus v_0)} a_{i} x_{i} + \sum_{\ell=1}^p d_{\ell} x_{v_{\ell}} + \left(\left( \sum_{\ell=1}^p d_\ell \right) - a_v \right) x_{v_0}\leq \b - a_v +  \sum_{\ell=1}^p d_{\ell}
\end{equation}
is valid for $\STAB(H)$.
\end{lemma}

\begin{proof}
Suppose $S \subseteq V(H)$ is an inclusion-wise maximal stable set in $H$. We define $S' \subseteq V(G)$ where
\[
S' \ce \begin{cases}
S \setminus \set{v_0} &\tn{if $v_0 \in S$;}\\
(S \setminus \set{v_1,\ldots, v_p} ) \cup \set{v} &\tn{if $\set{v_1, \ldots, v_p} \subseteq S$;}\\
S \setminus \set{v_1,\ldots, v_p} &\tn{if $1 \leq | \set{v_1, \ldots, v_p} \cap S| \leq p-1$.}
\end{cases}
\]
In all cases, $S'$ is a stable set in $G$, and $a^{\top}\chi_{S'} \leq \b$ implies that $\chi_S$ satisfies~\eqref{lemVertexStretchFaceteq1}. Note that the third case is when we require the assumption~\eqref{lemVertexStretchFaceteq0} with $T \ce \set{v_1, \ldots, v_p} \cap S$.
\end{proof}

Due to its similarity with the aforementioned vertex-stretching operations studied in~\cite{LiptakT03}, our vertex-stretching operation shares some similar structural properties, which we point out below.

\begin{proposition}\label{propVertexStretch}
Let $H \in \S(G)$ be a graph obtained from $G$ by $p$-stretching vertex $v \in V(G)$. Then we have the following.
\begin{itemize}
\item[(i)]
Suppose $a^{\top}x \leq \b$ is valid for $\STAB(G)$ where $a \geq 0$. Then 
\begin{equation}\label{propVertexStretcheq1}
 \sum_{i \in V(H \ominus v_0)} a_{i} x_{i} + \sum_{\ell=1}^p a_v x_{v_{\ell}} + (p-1) a_v x_{v_0}\leq \b + (p-1)a_v
\end{equation}
is valid for $\STAB(H)$.
\item[(ii)]
Let $g$ be as defined in~\eqref{eqVertexStretchg}. If $x \not\in \STAB(G)$, then $\tilde{g}(x) \not\in \STAB(H)$.
\item[(iii)]
$r_+(H) \geq r_+(G)$.
\end{itemize}
\end{proposition}

\begin{proof}
First, (i) follows readily from Lemma~\ref{lemVertexStretchFacet} with $d_{\ell} \ce a_v$ for every $\ell \in [p]$. Here, the condition~\eqref{lemVertexStretchFaceteq0} holds as the right hand side is at least $\b$ for all non-empty $T \subset [p]$. 

For (ii), first suppose $x \not \in \STAB(G)$. If $x \not\in [0,1]^{V(G)}$, then $\tilde{g}(x) \not\in [0,1]^{V(H)}$ 
(where $\tilde{g}$ is as defined in~\eqref{eq_tildeg}) and the claim follows. Otherwise, there is a facet $a^{\top}x \leq \b$ of $\STAB(G)$ where $a \geq 0$ that is violated by $x$. Now notice that for every $j \in V(H)$,
\[
[\tilde{g}(x)]_j = \begin{cases}
x_{j} & \tn{if $j \in V(H \ominus v_0)$;}\\
x_v & \tn{if $j \in \set{v_1, \ldots, v_p}$;}\\
1-x_v & \tn{if $j =v_0$.}
\end{cases}
\]
Then $\tilde{g}(x)$ violates~\eqref{propVertexStretcheq1}, and thus does not belong to $\STAB(H)$.

Finally, as we have shown that $H$ is star-homomorphic to $G$ under the function $g$ as defined in~\eqref{eqVertexStretchg}, (iii) follows directly from Lemma~\ref{lemStarHomo2}.
\end{proof}

We remark that Proposition~\ref{propVertexStretch} is a generalization of the corresponding results on types (i) and (ii) stretching from~\cite{LiptakT03}, and our proof uses many of the same ideas from similar arguments therein.

Next, we prove a result somewhat similar to Proposition~\ref{propVertexStretch}(i) that derives some facets of the stable set polytope of the stretched graph.

\begin{proposition}\label{propVertexStretch2}
Let $H \in \S(G)$ be a graph obtained from $G$ by $p$-stretching vertex $v \in V(G)$, and suppose $a^{\top}x \leq \b$ is a facet of $\STAB(G)$ where $a \geq 0$. For every $\ell \in [p]$, define $A_{\ell} \ce \Gamma_H(v_{\ell}) \setminus \set{v_0}$ and 
\begin{equation}\label{propVertexStretch2eq1}
d_{\ell} \ce a_v - \b + \max\set{ a^{\top}x : x \in \STAB(G), x_{i} = 0~\forall i \in  \set{v} \cup \bigcup_{j \in [p], j \neq \ell} A_{j} }.
\end{equation}
If the inequality~\eqref{lemVertexStretchFaceteq1} is valid for $\STAB(H)$, then it is a facet of $\STAB(H)$.
\end{proposition}

\begin{proof}
For convenience, let $n \ce |V(G)|$. Since $a^{\top}x \leq \b$ is a facet of $\STAB(G)$, there exist stable sets $S_1, \ldots, S_{n} \subseteq V(G)$ whose incidence vectors are affinely independent and all satisfy $a^{\top}x \leq \b$ with equality. Also, for every $\ell \in [p]$, let $D_{\ell}$ be a stable set that attains the maximum in the definition of $d_{\ell}$ in~\eqref{propVertexStretch2eq1}. We then define $S_1', \ldots, S_{n+p}'$ as follows. For every $i \in [n]$,
\[
S_i' \ce
\begin{cases}
S_i \cup \set{v_0} &\tn{if $v \not\in S_i$;}\\
(S_i \setminus \set{v}) \cup \set{v_1, \ldots, v_p} &\tn{if $v \in S_i$.}\\
\end{cases}
\]
We also define 
\[
S_{n+i}' \ce D_{i} \cup \set{ v_ j : j \in [p], j \neq i}
\]
for all $i \in [p]$.

Observe that $S_1', \ldots, S_{n+p}'$ must all be stable sets in $H$. Also, using the fact that incidence vectors of $S_1, \ldots, S_n$ are affinely independent and satisfy $a^{\top}x \leq \b$ with equality, we see that the incidence vectors of $S_1', \ldots, S_{n+p}'$ are affinely independent and all satisfy~\eqref{lemVertexStretchFaceteq1} with equality. 

Thus, if we know that~\eqref{lemVertexStretchFaceteq1} is valid for $\STAB(H)$, it must be a facet.
\end{proof}

The special case of $p=2$ in Proposition~\ref{propVertexStretch2} is particularly noteworthy:

\begin{corollary}\label{corVertexStretch2}
Let $H \in \S_2(G)$ be a graph obtained from $G$ by $2$-stretching vertex $v \in V(G)$, and suppose $a^{\top}x \leq \b$ is a facet of $\STAB(G)$ where $a \geq 0$. Define $A_1 \ce \Gamma_H(v_1) \setminus \set{v_0}, A_2 \ce \Gamma_H(v_2) \setminus \set{v_0}$, as well as the quantities
\begin{align*}
d_1 &\ce a_v - \b + \max\set{ a^{\top}x : x \in \STAB(G), x_{i} = 0~\forall i \in  \set{v} \cup A_2}, \\
d_2 &\ce a_v - \b + \max\set{ a^{\top}x : x \in \STAB(G), x_{i} = 0~\forall i \in  \set{v} \cup A_1}.
\end{align*}
If $d_1+d_2 \geq a_v$, then~\eqref{lemVertexStretchFaceteq1} is a facet of $\STAB(H)$.
\end{corollary}

\begin{proof}
This is largely a specialization of Proposition~\ref{propVertexStretch2} to the case $p=2$. Notice that the additional assumption of~\eqref{lemVertexStretchFaceteq1} being valid is not necessary in this case because $\set{1,2}$ has exactly two non-empty and proper subsets, and so the definition of  $d_1, d_2$ herein are enough to guarantee that the assumption~\eqref{lemVertexStretchFaceteq0} is met.
\end{proof}

Finally, we close this section by mentioning a ``reverse'' implication of Proposition~\ref{propVertexStretch}.

\begin{corollary}\label{corVertexStretch}
Let $H$ be a graph, and $v_0 \in V(H)$. If $\Gamma_H(v_0)$ is a stable set, then $r_+(G) \leq r_+(H)$,
where $G$ is the graph obtained from $H$ by contracting the set of vertices $\Gamma_H[v_0]$.
\end{corollary}

\begin{proof}
Let $v_1, \ldots, v_p \in V(H)$ be the neighbours of $v_0$, and define $A_i \ce \Gamma_H(v_i) \setminus \set{v_0}$ for every $i \in [p]$. Then there are two cases: First, if there exists $i \in [p]$ where $A_i = \bigcup_{j \in [p]} A_j$, then $G$ is isomorphic to the subgraph of $H$ induced by $(V(H) \setminus \Gamma_H[v_0]) \cup \set{v_i})$, and thus $r_+(G) \leq r_+(H)$.

Next, if there does not exist  $i \in [p]$ where $A_i = \bigcup_{j \in [p]} A_j$, then $H \in \S_p(G)$, and it follows that $r_+(G) \leq r_+(H)$.
\end{proof}

\def\x{270 - 180/6}
\def\z{360/6}
\def\sc{1.5}
\begin{figure}[ht!]
\begin{center}
\begin{tabular}{cc}

\begin{tikzpicture}[scale=\sc, thick,main node/.style={circle, minimum size=4mm, inner sep=0.1mm,draw,font=\tiny\sffamily}]
\node[main node] at ({cos(\x+(0)*\z)},{sin(\x+(0)*\z)}) (1) {$3$};
\node[main node] at ({cos(\x+(1)*\z)},{sin(\x+(1)*\z)}) (2) {$2$};
\node[main node] at ({cos(\x+(2)*\z)},{sin(\x+(2)*\z)}) (3) {$1$};
\node[main node] at ({cos(\x+(3)*\z)},{sin(\x+(3)*\z)}) (4) {$8$};
\node[main node] at ({cos(\x+(4)*\z)},{sin(\x+(4)*\z)}) (5) {$7$};
\node[main node] at ({cos(\x+(5)*\z)},{sin(\x+(5)*\z)}) (6) {$4$};
\node[main node] at (0,0) (7) {$5$};
\node[main node] at (0, 1.2) (8) {$6$};

 \path[every node/.style={font=\sffamily}]
(1) edge (2)
(2) edge (3)
(3) edge (4)
(4) edge (5)
(5) edge (6)
(6) edge (1)
(7) edge (3)
(7) edge (4)
(7) edge (5)
(7) edge (6)
(8) edge (4)
(8) edge (5);
\end{tikzpicture}

&

\begin{tikzpicture}[scale=\sc,  thick,main node/.style={circle, minimum size=4mm, inner sep=0.1mm,draw,font=\tiny\sffamily}]
\node[main node] at ({cos(\x+(0)*\z)},{sin(\x+(0)*\z)}) (1) {$3$};
\node[main node] at ({cos(\x+(1)*\z)},{sin(\x+(1)*\z)}) (2) {$2$};
\node[main node] at ({cos(\x+(2)*\z)},{sin(\x+(2)*\z)}) (3) {$1$};
\node[main node] at ({cos(\x+(5)*\z)},{sin(\x+(5)*\z)}) (6) {$4$};
\node[main node] at (0,0) (7) {$5$};
\node[main node] at (0, 1.2) (8) {$6$};

 \path[every node/.style={font=\sffamily}]
(1) edge (2)
(2) edge (3)
(3) edge (8)
(8) edge (6)
(6) edge (1)
(7) edge (3)
(7) edge (6)
(7) edge (8);
\end{tikzpicture}
\\
$H$ & $G$
\end{tabular}
\caption{An example in which contracting the closed neighborhood of a vertex increases the graph's $\LS_+$-rank}\label{figContract}
\end{center}
\end{figure}
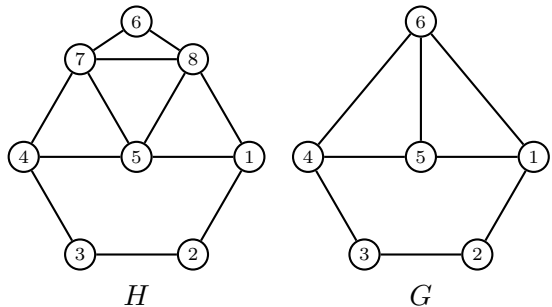

\begin{example}
In general, it is possible that contracting the closed neighborhood of a vertex results in an increase in the graph's $\LS_+$-rank. For example, notice that the graph $H$ in Figure~\ref{figContract} is the union of two $\LS_+$-rank-$1$ graphs whose intersection is the cut clique $\set{7,8}$, and thus it follows from Proposition~\ref{propCliqueCut} that $r_+(H) = 1$. However, contracting $\Gamma_H[6]$ in $H$ results in $G$, which is isomorphic to $G_{2,1}$, and so $r_+(G)=2$.
\end{example}

\section{$\LS_+$-minimal graphs via $2$-stretching cliques}\label{sec5}

In this section, we are interested in studying graphs with the fewest number of vertices with a given $\LS_+$-rank. Given $\ell \in \mN$, define $n_+(\ell)$ to be the minimum number of vertices on which there exists a graph $G$ with $r_+(G) = \ell$. It follows immediately from Theorem~\ref{thmNover3} that $n_+(\ell) \geq 3\ell$ for every $\ell \in \mN$. On the other hand, Theorem~\ref{thmHk} implies that $n_+(\ell) \leq 16\ell$. Thus, we know that $n_+(\ell) = \Theta(\ell)$ asymptotically.

Recall that a graph $G$ is $\ell$-minimal if $r_+(G) = \ell$ and $|V(G)| = 3\ell$. The following result establishes a close connection between $\ell$-minimal graphs and $2$-stretching vertices. 

\begin{theorem}
Let $H$ be an $\ell$-minimal graph where $\ell \geq 2$. Then there exists a graph $G$ where $H \in \S_2(G)$.
\end{theorem}

\begin{proof}
First, since $r_+(H) = \ell$, there exists a vertex $v_0$ where $r_+(H \ominus v_0) = \ell-1$. This implies that $|V(H \ominus v_0)| \geq 3\ell-3$, and thus $\deg(v_0) \leq 2$. If $\deg(v_0) = 1$, $H$ would contain a cut vertex, and Proposition~\ref{propCliqueCut} implies there would be a proper subgraph of $H$ with the same $\LS_+$-rank as $H$. Thus, we obtain that $\deg(v_0) =2$. Let $v_1, v_2$ denote the two neighbours of $v_0$. If $\set{v_1,v_2} \in E(H)$, then the edge would form a cut clique in $H$, and Proposition~\ref{propCliqueCut} again implies that there would a proper subgraph of $H$ with the same $\LS_+$-rank as $H$. Thus, $\set{v_1, v_2}$ must be a stable set in $H$.

Next, define $G$ to be the graph obtained from $H$ by contracting $\set{v_0, v_1, v_2}$, and label the new vertex $v$. We claim that $H$ can be obtained from $G$ by $2$-stretching $v \in V(G)$ with $A_1 \ce \Gamma_H(v_1) \setminus \set{v_0}$ and $A_2 \ce \Gamma_H(v_1) \setminus \set{v_0}$. To prove this, it suffices to show that $A_1, A_2$ are both proper subsets of $A_1 \cup A_2$. Equivalently, we need to show that $\Gamma_H(v_1) \setminus \Gamma_H(v_2)$ and $\Gamma_H(v_2) \setminus \Gamma_H(v_1)$ are non-empty.

Let $a^{\top}x \leq \b$ be a facet of $\STAB(H)$ of $\LS_+$-rank $\ell$. Then there must be stable sets $S_1, \ldots, S_{3\ell} \subseteq V(H)$ whose incidence vectors are affinely independent and all satisfy $a^{\top}x \leq \b$ with equality. Also, since $H$ is $\ell$-minimal, $a$ must have full support. Therefore, $S_i$ must be inclusion-wise maximal for all $i \in [3\ell]$, and hence belongs to one of the following cases:
\begin{enumerate}
\item
$S_i \cap \set{v_0, v_1, v_2} = \set{v_0}$;
\item
$S_i \cap \set{v_0, v_1, v_2} = \set{v_1}$;
\item
$S_i \cap \set{v_0, v_1, v_2} = \set{v_2}$;
\item
$S_i \cap \set{v_0, v_1, v_2} = \set{v_1, v_2}$.
\end{enumerate}

Since $\chi_{S_1}, \ldots, \chi_{S_{3\ell}}$ are affinely independent, one of these stable sets contains $v_0$ and belongs to Case (1), so assume without loss of generality that $v_0 \in S_1$. Now consider the matrix $A$ formed by the row vectors $(\chi_{S_2} - \chi_{S_1})^{\top}, (\chi_{S_3} - \chi_{S_1})^{\top}, \ldots, (\chi_{S_{3\ell}} - \chi_{S_1})^{\top}$. Since $S_1, \ldots, S_{3\ell}$ are affinely independent, $A$ must have linearly independent rows. This means that, if we focus on the submatrix $A'$ of $A$ which consists of just the three columns corresponding to $v_0, v_1$, and $v_2$, $A'$ must have rank $3$. This implies that there must exist at least one $S_i$ belonging to each of Cases (2), (3), and (4). 

Now consider a stable set $S_i$ that belongs to Case (2). Since $S_i$ is inclusion-wise maximal, it must contain a vertex that is adjacent to $v_2$ and not $v_1$, and so we have found a vertex that belongs to $\Gamma_H(v_2) \setminus \Gamma_H(v_1)$. The same argument applied to an $S_i$ from Case (3) gives a vertex that belongs to $\Gamma_H(v_1) \setminus \Gamma_H(v_2)$. This finishes the proof.
\end{proof}

Thus, for the remainder of this section, we will focus on $2$-stretching vertices, and study when that helps (and does not help) in generating $\ell$-minimal graphs. Since we will be studying graphs obtained from applying a sequence of $2$-stretching operations, we recursively define 
\[
\S_2^k(G) \ce \bigcup_{G' \in \S_2^{k-1}(G)} \S_2(G')
\]
for every graph $G$ and integer $k \geq 1$. That is, $\S_2^k(G)$ is the set of graphs that can be obtained from $G$ by a sequence of $k$ $2$-stretching operations. We also let $\S_2^0(G) \ce \set{G}$. The following is a basic property of the graphs in $\S_2^k(G)$.

\begin{lemma}\label{05stretchalpha}
Let $G$ be a graph, and let $H \in \S_2^k(G)$. Then $\a(H) = \a(G) +k$.
\end{lemma}

\begin{proof}
Let $H \in \S_2(G)$ be a graph obtained from $G$ by $2$-stretching vertex $v \in V(G)$. To prove our claim, it suffices to show that $\a(H) = \a(G) + 1$. Consider a set of vertices $S \subseteq V(G)$. If $v \in S$, then $S$ is a stable set in $G$ if and only if $(S \setminus \set{v}) \cup \set{v_1, v_2}$ is a stable set in $H$. If $v \not\in S$, then $S$ is a stable set in $G$ if and only if $S \cup \set{v_0}$ is a stable set in $H$. Thus, we see that $\a(H) = \a(G) + 1$.
\end{proof}

Recall the graphs $G_{2,1}, G_{2,2}$, and $G_{3,1}$. In Figure~\ref{figKnownEG}, we labelled the vertices of these graphs to highlight the fact that all three graphs can be obtained from applying a number of $2$-stretching operations to a complete graph. In fact, every known $\ell$-minimal graph to date --- the $3$-cycle and the three graphs in Figure~\ref{figKnownEG} --- belongs to $\S^{\ell-1}(K_{\ell+2})$. Thus, for the remainder of this section, we focus on graphs obtained from $2$-stretching vertices of a complete graph, and prove some results about the $\LS_+$-ranks of these graphs. 

Generally, to prove that a given graph $G$ has $\LS_+$-rank at least $p$, one needs to find a point $x \in \LS_+^{p-1}(G) \setminus \STAB(G)$. In particular, showing that $x \in \LS_+^{p-1}(G)$ often requires finding a corresponding certificate matrix $Y$ and showing that the matrix is positive semidefinite (among satisfying other properties imposed by $\LS_+$). Thus, we first provide a framework for easily and reliably verifying the positive semidefiniteness of specific matrices. Given a symmetric matrix $Y \in \mR^{n \times n}$, we say that $U, V \in \mZ^{n \times n}$ is a \emph{$UV$-certificate} of $Y$ if
\begin{itemize}
\item[(i)]
$k Y = U^{\top} U + V$ for some $k \in \mN$, and,
\item[(ii)]
$V$ is diagonally dominant (i.e.,  $V_{ii} \geq \sum_{j \neq i} |V_{ij}|$ for every $i \in [n]$).
\end{itemize}

Observe that the existence of a $UV$-certificate implies that $Y \succeq 0$ (these certificates are sum-of-squares certificates, and every rational matrix $Y \in \mS_+^n$ admits such certificates). Conversely, since every entry in $U$ and $V$ is an integer, verifying (i) and (ii) for a given $UV$-certificate only involves elementary numerical operations on whole numbers.

Next, we show that if we $2$-stretch a vertex in a complete graph, the result is always a graph with $\LS_+$-rank $2$. 

\begin{proposition}\label{prop51}
Let $n \geq 4$. Then $r_+(H) = 2$ for all $H \in \S_2(K_n)$.
\end{proposition}

\begin{proof}
Let $H \in \S_2(K_n)$, and assume without loss of generality that $H$ is obtained from $K_n$ by $2$-stretching vertex $n$. Also, let $G_n$ be the graph obtained from $2$-stretching vertex $n$ in $K_n$ with $A_1 \ce [n-1] \setminus \set{1}$ and $A_2 \ce [n-2]$. (For example, $G_4$ is the graph $G_{2,2}$ from Figure~\ref{figKnownEG} and $G_5$ is shown in Figure~\ref{figstretching}.) Now since the sets $A_1, A_2$ used in stretching vertex $n$ to obtain $H$ must satisfy $A_1, A_2 \subset [n-1]$ and $A_1 \cup A_2 = [n-1]$, we see that $H$ must be isomorphic to a subgraph of $G_n$. 

Next, we show that $\a_{\LS_+}(G_4) > 2$. Consider the certificate matrix 

{\footnotesize
\[
Y \ce 
\bbordermatrix{
&& 1 & 2 & 3 & 4_1 & 4_0 & 4_2 \cr
& 200 & 78 & 12 & 78 &78 &78 & 78 \cr
& 78 & 78 & 0 & 0& 39&39 & 0\cr
& 12 & 0 & 12& 0 &0 &12 & 0 \cr
& 78 & 0 &0 & 78& 0&39 & 39\cr
& 78 &39 &0 &0 & 78& 0 & 39\cr
& 78 &39 &12 &39 &0 & 78& 0\cr
& 78 &0 &0 &39 & 39& 0& 78 \cr
}
\]}

Note that the columns of $Y$ are labelled by the vertices in $G_4$ they correspond to (the rows of $Y$ follow the same order of indexing). Observe that $Y \succeq 0$ --- a $UV$-certificate for $Y$ is

{\scriptsize
\[
U \ce \begin{bmatrix}
0 & 0 & 0 & 0 & 0 & 0 & 0\\
-73 & -26 & -141 & -26 & 60 & 132 & 61\\
0 & 124 & 0 & -124 & -200 & 0 & 200\\
27 & -181 & 247 & -181 & 51 & 159 & 51\\
0 & -527 & 0 & 527 & -326 & 0 & 326\\
1 & -166 & -73 & -166 & 449 & -556 & 449\\
1224 & 482 & 60 & 482 & 482 & 485 & 482
\end{bmatrix},
~V \ce 
\begin{bmatrix}
2765 & 917 & 91 & 917 & 316 & -11 & 389\\
917 & 1308 & 3 & -212 & -136 & 10 & 29\\
91 & 3 & 601 & 3 & -280 & 71 & -139\\
917 & -212 & 3 & 1308 & 3 & 10 & -110\\
316 & -136 & -280 & 3 & 1328 & -155 & -45\\
-11 & 10 & 71 & 10 & -155 & 664 & -287\\
389 & 29 & -139 & -110 & -45 & -287 & 1207
\end{bmatrix},
\]}
which gives $7535 Y = U^{\top} U + V$. One can also check that $Ye_i, Y(e_0-e_i) \in \FRAC(G_4)$ for every $i \in V(G_4)$. This shows that $\bar{x} \ce \frac{1}{200} (78, 12, 78, 78, 78, 78)^{\top} \in \LS_+(G_4)$. Now since $\bar{e}^{\top} \bar{x} = 2.01 >2= \a(G_4)$, we see that $r_+(G_4) \geq 2$. Since $G_n$ contains $G_4$ as an induced subgraph for all $n \geq 4$, we conclude that $\a_{\LS_+}(G_n) > 2$. Then Lemma~\ref{lem05subgraph}(ii) implies that $\a_{\LS_+}(H) > 2$. Since $\a(H) = 2$, it follows that $r_+(H) \geq 2$.

Finally, notice that $H - n_0$ must be a perfect graph, so $r_+(H- n_0) \leq 1$ and consequently $r_+(H) \leq 2$. Thus, we conclude that $r_+(H) = 2$.
\end{proof}

We remark that, to obtain the numerical certificates in the proof above, we first found a matrix $Y_0$ by optimizing over $\LS_+(G_4)$ using CVX, a package for specifying and solving convex programs~\cite{CVX, GrantB08}, as well as the SDP solver SeDuMi~\cite{Sturm99}. Then we let $V_0 \ce \ld_0I$, where $\ld_0$ is the least eigenvalue of $Y_0$, and solved for a matrix $U_0$ where $U_0^{\top}U_0 + V_0 = Y_0 $. From there, we ``discretize'' $U_0, V_0$ in a somewhat trial-and-error manner until we obtain a suitable $UV$-certificate. The numerical certificates for the $4$-minimal graph $G_{4,1}$ presented in the proof of Theorem~\ref{thmG41} and Lemmas~\ref{lemG412},~\ref{lemG411}, and~\ref{lemG413} are found in a similar manner. It would be interesting to look into the prospect of generating these certificates via a more systematic and efficient approach.

Next, we also remark that in the proof for $r_+(G_{2,2}) \geq 2$ in~\cite{EscalanteMN06}, the following certificate matrix was given:

{\footnotesize
\[
\frac{1}{2688}
\bbordermatrix{
&& 4_2 & 1 & 3 & 4_1 & 4_0 & 2 \cr
& 2688 & 769 &769 &769 &769 &769 &1538 \cr
& 769 & 769 & 0 & 336 & 413 \frac{7}{13} & 0 & 0 \cr
& 769 &0 &769 & 0 & 336 & 384 & 0 \cr
& 769 &336 &0 &769 & 0 & 384 & 0 \cr
& 769 &413 \frac{7}{13} &336 &0 &769 & 0 & 896 \cr
& 769 & 0&384 & 384 &0 &769 &0 \cr
& 1538 &0 &0 & 0& 896& 0 &1538 \cr
}
\]}
However, the certificate is incorrect: $Y[2, 4_1] = \frac{896}{2688} = \frac{1}{3} > Y[0,4_1]$, and thus violates $Ye_{4_1} \in \cone(\FRAC(G_{2,2}))$. In fact, since the vector $\frac{1}{2688} (769, 769, 769, 769, 769, 1538)^{\top}$ contains only one entry greater than $\frac{1}{3}$, any certificate matrix for this vector cannot contain an off-diagonal entry of $\frac{1}{3}$ (which would have to appear in at least 2 columns in the certificate). Still, the claim that $r_+(G_{2,2}) = 2$ is correct, as shown in the proof of Proposition~\ref{prop51}.

Next, while all graphs in $\S_2(K_n)$ have $\LS_+$-rank $2$, we show that not all graphs in $\S_2^2(K_n)$ have $\LS_+$-rank $3$. Given a graph $G$, we say that a path in $G$ is \emph{sparse} if at most one of the vertices in the path has degree greater than $2$ in $G$. For example, in Figure~\ref{figSparsePath}, the graph on the left contains a sparse path $4_0, 4_1, 5_1, 5_0, 5_2$ of length $4$, while the graph on the right also contains a sparse path $4_0, 4_2, 3, 5_2, 5_0$ of length $4$. Then we have the following.

\def\x{270 - 360/5}
\def\z{360/5}
\def\y{0.7}
\def\sc{2}

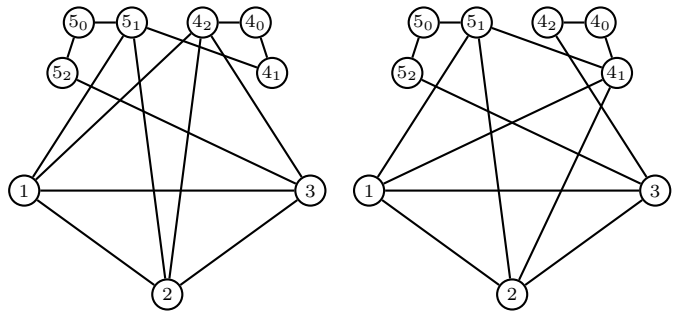
\begin{figure}[ht!]
\begin{center}
\begin{tabular}{cc}
\begin{tikzpicture}
[scale=\sc, thick,main node/.style={circle, minimum size=4mm, inner sep=0.1mm,draw,font=\tiny\sffamily}]
\node[main node] at ({cos(\x+(0)*\z)},{sin(\x+(0)*\z)}) (1) {$1$};
\node[main node] at ({cos(\x+(1)*\z)},{sin(\x+(1)*\z)}) (2) {$2$};
\node[main node] at ({cos(\x+(2)*\z)},{sin(\x+(2)*\z)}) (3) {$3$};

\node[main node] at ({ \y* cos(\x+(3)*\z) + (1-\y)*cos(\x+(2)*\z)},{ \y* sin(\x+(3)*\z) + (1-\y)*sin(\x+(2)*\z)}) (4) {$4_1$};
\node[main node] at ({cos(\x+(3)*\z)},{sin(\x+(3)*\z)}) (5) {$4_0$};
\node[main node] at ({ \y* cos(\x+(3)*\z) + (1-\y)*cos(\x+(4)*\z)},{ \y* sin(\x+(3)*\z) + (1-\y)*sin(\x+(4)*\z)}) (6) {$4_2$};

\node[main node] at ({ \y* cos(\x+(4)*\z) + (1-\y)*cos(\x+(3)*\z)},{ \y* sin(\x+(4)*\z) + (1-\y)*sin(\x+(3)*\z)}) (7) {$5_1$};
\node[main node] at ({cos(\x+(4)*\z)},{sin(\x+(4)*\z)}) (8) {$5_0$};
\node[main node] at ({ \y* cos(\x+(4)*\z) + (1-\y)*cos(\x+(5)*\z)},{ \y* sin(\x+(4)*\z) + (1-\y)*sin(\x+(5)*\z)}) (9) {$5_2$};

 \path[every node/.style={font=\sffamily}]
(1) edge (2)
(1) edge (3)
(2) edge (3)
(4) edge (5)
(5) edge (6)
(7) edge (8)
(8) edge (9)
(4) edge (7)
(1) edge (6)
(2) edge (6)
(3) edge (6)
(1) edge (7)
(2) edge (7)
(3) edge (9);
\end{tikzpicture}

&

\begin{tikzpicture}
[scale=\sc, thick,main node/.style={circle, minimum size=4mm, inner sep=0.1mm,draw,font=\tiny\sffamily}]
\node[main node] at ({cos(\x+(0)*\z)},{sin(\x+(0)*\z)}) (1) {$1$};
\node[main node] at ({cos(\x+(1)*\z)},{sin(\x+(1)*\z)}) (2) {$2$};
\node[main node] at ({cos(\x+(2)*\z)},{sin(\x+(2)*\z)}) (3) {$3$};

\node[main node] at ({ \y* cos(\x+(3)*\z) + (1-\y)*cos(\x+(2)*\z)},{ \y* sin(\x+(3)*\z) + (1-\y)*sin(\x+(2)*\z)}) (4) {$4_1$};
\node[main node] at ({cos(\x+(3)*\z)},{sin(\x+(3)*\z)}) (5) {$4_0$};
\node[main node] at ({ \y* cos(\x+(3)*\z) + (1-\y)*cos(\x+(4)*\z)},{ \y* sin(\x+(3)*\z) + (1-\y)*sin(\x+(4)*\z)}) (6) {$4_2$};

\node[main node] at ({ \y* cos(\x+(4)*\z) + (1-\y)*cos(\x+(3)*\z)},{ \y* sin(\x+(4)*\z) + (1-\y)*sin(\x+(3)*\z)}) (7) {$5_1$};
\node[main node] at ({cos(\x+(4)*\z)},{sin(\x+(4)*\z)}) (8) {$5_0$};
\node[main node] at ({ \y* cos(\x+(4)*\z) + (1-\y)*cos(\x+(5)*\z)},{ \y* sin(\x+(4)*\z) + (1-\y)*sin(\x+(5)*\z)}) (9) {$5_2$};

 \path[every node/.style={font=\sffamily}]
(1) edge (2)
(1) edge (3)
(2) edge (3)
(4) edge (5)
(5) edge (6)
(7) edge (8)
(8) edge (9)
(4) edge (7)
(1) edge (4)
(2) edge (4)
(3) edge (6)
(1) edge (7)
(2) edge (7)
(3) edge (9);
\end{tikzpicture}
\end{tabular}
\caption{Two graphs in $\S_2^2(K_5)$ with sparse paths}\label{figSparsePath}
\end{center}
\end{figure}

\begin{proposition}\label{propSparsePath}
Let $\ell \geq 2$. If $G \in \S_2^{\ell-1}(K_{\ell+2})$ contains a sparse path of length at least $3$, then $G$ is not $\ell$-minimal.
\end{proposition}

\begin{proof}
First, observe that each $2$-stretching operation adds $2$ vertices to a given graph, and so every $G \in \S_2^{\ell-1}(K_{\ell+2})$ has $\ell+2 + 2(\ell-1) = 3\ell$ vertices. We show that if $G$ contains a sparse path of length at least $3$, then $r_+(G) < \ell$.

Let $U$ be a set of vertices that induces a sparse path in $G$ where $|U| \geq 4$. Then there is a vertex  $u_1 \in U$ where  $\deg(v) = 2$ for all $v \in U \setminus \set{u_1}$. Now notice that in the graph $G - u_1$, there must exist $u_2 \in U \setminus \set{u_1}$ where $\deg(u_2) = 1$. Thus, Proposition~\ref{propCliqueCut} implies that $G - u_1$ and $G- \set{u_1,u_2}$ must have the same $\LS_+$-rank.

Again, since $U$ induce a sparse path in $G$, there must be a vertex $u_3 \in U \setminus \set{u_1,u_2}$ such that $\deg(u_3) = 1$ in $G - \set{u_1, u_2}$, and so $G - \set{u_1,u_2}$ and $G - \set{u_1,u_2,u_3}$ must have the same $\LS_+$-rank. Iterating this argument until we remove all vertices of $U$ from $G$, we see that $r_+(G - u_1) = r_+(G - U)$.

Since $|U| \geq 4$, $G-U$ has $3\ell  - |U| < 3(\ell-1)$ vertices, and so $r_+(G - u_1) = r_+(G- U) \leq \ell -2$. As a result, $r_+(G) \leq \ell -1$, and thus $G$ is not $\ell$-minimal.
\end{proof}

Thus, both graphs in Figure~\ref{figSparsePath} have $\LS_+$-rank at most $2$. (In fact, they have rank $2$ as they both contain $G_{2,1}$ as an induced subgraph.) Next, we show that if we $2$-stretch a vertex in $K_n$, and then $2$-stretch one of the three new vertices in the stretched graph, the resulting graph cannot have $\LS_+$-rank $3$.

\begin{proposition}\label{05doublestretch}
Let $n \geq 4$. Suppose $G_1 \in \S_2(K_n)$ is obtained by stretching vertex $n$ in $K_n$, and $G_2 \in \S_2^2(K_n)$ is obtained by stretching vertex $n_0, n_1$, or $n_2$ in $G_1$. Then $r_+(G_2) = 2$.
\end{proposition}

\begin{proof}
First, if $G_2$ is obtained from $G_1$ by stretching $n_0$, then $n_{00}, n_{01}, n_{02}$ all have degree $2$. Notice that $G_2 - n_{00}$ must be a perfect graph, and so $r_+(G_2 - n_{00}) \leq 1$, which implies $r_+(G_2) \leq 2$ in this case.

Otherwise, assume without loss of generality that $G_2$ is obtained from $G_1$ by stretching $n_1$, and that $\set{n_{12}, n_0} \in E(G_2)$. (Note that $\set{n_{11}, n_0}$ may or may not be an edge.) Now notice that $G_2 - n_{12}$ is a perfect graph, and thus, $r_+(G_2) \leq 2$ in this case as well. Finally, since $r_+(G_1) = 2$ (from Proposition~\ref{prop51}) and $G_2 \in \S(G_1)$, Proposition~\ref{propVertexStretch}(iii) implies that $r_+(G_2) \geq 2$.
\end{proof}

Thus, to obtain a graph with $\LS_+$-rank $3$ in $\S_2^2(K_n)$, it is necessary that we stretch two of the original vertices of $K_n$. (That is not sufficient though, as shown for the graphs in Figure~\ref{figSparsePath}.)

Next, observe that if $G$ is an $\ell$-minimal graph, then $\STAB(G)$ must have a facet with $\LS_+$-rank $\ell$, and it is necessary that such a facet has full support. (Otherwise it follows from Lemma~\ref{lemfacet} that $G$ has a proper subgraph with the same $\LS_+$-rank, which implies the existence of a graph on fewer than $3\ell$ vertices with $\LS_+$-rank $\ell$.)  We provide more circumstantial evidence that $2$-stretching a number of original vertices of a complete graph is a promising approach for generating $\ell$-minimal graphs by showing that the stable set polytope of many of  these graphs have a full-support facet. Let  $H \in \S_2^{k}(K_{n})$ be a graph that is obtained from $2$-stretching the vertices $D \subseteq V(K_n)$ (and thus $|D| = k$). For every $v \in V(K_n)$, we define the vertices \emph{associated} with $v$ to be $v_0, v_1$, and $v_2$ if $v \in D$, and just the unstretched vertex $v$ if $v \not\in D$. Notice that every vertex in $H$ is associated with exactly one $v \in V(K_n)$. Then we have the following.

\begin{proposition}\label{propSkKlfacet}
Let $k, n$ be integers where $n \geq 3$ and $n \geq k \geq 0$. Suppose $H \in \S_2^{k}(K_{n})$ is obtained from $K_{n}$ by $2$-stretching the vertices $D \subseteq V(K_{n})$. Furthermore, suppose $H$ satisfies the following:
\begin{itemize}
\item[($*$)]  For every $v \in D$ and for every $j \in \set{1,2}$, there exists $u \in V(K_n) \setminus \set{v}$ where $v_{j}$ is not adjacent with any vertex associated with $u$.
\end{itemize}
Then $\sum_{i \in V(H)} x_i \leq k+1$ is a facet of $\STAB(H)$.
\end{proposition}

\begin{proof}
We prove our claim by induction on $k$. When $k=0$, $H = K_{n}$, and the claim obviously holds. Next, assume $1 \leq k \leq n$. Since $H \in \S_2^k(K_n)$, there exists $G \in \S_2^{k-1}(K_{n})$ where $H \in \S_2(G)$. (So there exists $v \in D$ where $H$ is obtained from $G$ by stretching $v$.)

Since $H$ satisfies ($*$) by assumption, so must $G$, and so it follows from the inductive hypothesis that $\sum_{i \in V(G)} x_i \leq k$ is a facet of $\STAB(G)$. To prove our claim, we make use of Proposition~\ref{propVertexStretch2} and show that $d_1=d_2=1$. To do so, let $A_1 \ce \Gamma_H(v_1) \setminus \set{v_0}, A_2 \ce \Gamma_H(v_2) \setminus \set{v_0}$, 
\begin{align*}
c_1 &\ce  \max\set{ \bar{e}^{\top}x : x \in \STAB(G), x_{i} = 0~\forall i \in  \set{v} \cup A_2}, \\
c_2 &\ce  \max\set{ \bar{e}^{\top}x : x \in \STAB(G), x_{i} = 0~\forall i \in  \set{v} \cup A_1}.
\end{align*}
Then it suffices to prove that $c_1=c_2=k$, which would then imply that $d_1=d_2=1$.

First, it is obvious that $c_1, c_2 \leq k$ since $\alpha(G) = k$. Next, consider $\Gamma_H(v_1)$. By the assumption that $H$ satisfies ($*$), at least one of the following must hold:
\begin{itemize}
\item
There exists $u \in V(K_n)$ where $u \not\in D$ (so $u \in V(H)$) and $u \not\in \Gamma_H(v_1)$. Then
\[
S \ce \set{u} \cup \set{ w_0 : w \in D \setminus \set{v}}
\]
is a stable set that gives $c_1 = k$.
\item
There exists $u \in V(K_n), u \neq v$ where $u \in D$ (so $u_0, u_1. u_2 \in V(H)$) and $u_0, u_1, u_2 \not\in \Gamma_H(v_1)$. Then
\[
S \ce \set{u_1, u_2} \cup \set{ w_0 : w \in D \setminus \set{u,v}}
\]
is a stable set that gives $c_1=k$.
\end{itemize}
The same argument shows that $c_2=k$, and this finishes the proof. 
\end{proof}

We remark that the assumption of stretching only the original vertices of $K_{n}$ in Proposition~\ref{propSkKlfacet} is necessary, as shown in the following example.

\begin{example}\label{egSTABfacet}
Recall the graph $G_{2,2}$, introduced in Figure~\ref{figKnownEG} and reproduced below in Figure~\ref{figSTABfacet} (left). Observe that $G_{2,2} \in \S_2(K_4)$, and that $\bar{e}^{\top}x \leq 2$ is a facet of $\STAB(G_{2,2})$. Now, we $2$-stretch the vertex $4_2 \in V(G_{2,2})$ to obtain $H \in \S_2^2(K_4)$ as shown in Figure~\ref{figSTABfacet} (right). Observe that the subgraph of $H$ induced by vertices $1,2,3,4_1,4_0,4_{21}$ is isomorphic to $G_{2,1}$ from Figure~\ref{figKnownEG}. Thus, 
\begin{equation}\label{egSTABfaceteq1}
x_{1} + x_2 + x_3 + x_{4_1} + x_{4_0} + x_{4_{21}} \leq 2
\end{equation}
is valid for $\STAB(H)$. (In fact, one can show that it is a facet of $\STAB(G)$ using Proposition~\ref{propVertexStretch2}.) This implies that $\sum_{i \in V(H)} x_i \leq 3$, which is the sum of~\eqref{egSTABfaceteq1} and the edge inequality $x_{4_{20}} + x_{4_{22}} \leq 1$, is not a facet of $\STAB(H)$.
\end{example}

\def\x{270 - 360/4}
\def\z{360/4}
\def\y{0.7}
\def\sc{2}

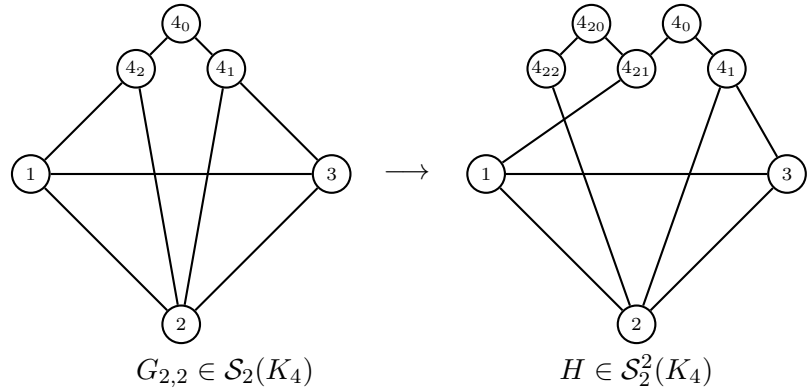
\begin{figure}[ht!]
\begin{center}
\begin{tabular}{cc}
\begin{tikzpicture}[scale=\sc, thick,main node/.style={circle, minimum size=5mm, inner sep=0.1mm,draw,font=\tiny\sffamily}]
\node[main node] at ({cos(\x+(0)*\z)},{sin(\x+(0)*\z)}) (1) {$1$};
\node[main node] at ({cos(\x+(1)*\z)},{sin(\x+(1)*\z)}) (2) {$2$};
\node[main node] at ({cos(\x+(2)*\z)},{sin(\x+(2)*\z)}) (3) {$3$};

\node[main node] at ({ \y* cos(\x+(3)*\z) + (1-\y)*cos(\x+(2)*\z)},{ \y* sin(\x+(3)*\z) + (1-\y)*sin(\x+(2)*\z)}) (4) {$4_1$};
\node[main node] at ({cos(\x+(3)*\z)},{sin(\x+(3)*\z)}) (5) {$4_0$};
\node[main node] at ({ \y* cos(\x+(3)*\z) + (1-\y)*cos(\x+(4)*\z)},{ \y* sin(\x+(3)*\z) + (1-\y)*sin(\x+(4)*\z)}) (6) {$4_2$};

 \path[every node/.style={font=\sffamily}]
(1) edge (2)
(2) edge (3)
(3) edge (1)
(4) edge (5)
(5) edge (6)
(4) edge (2)
(4) edge (3)
(6) edge (1)
(6) edge (2);

\node at (1.5,0) () {$\longrightarrow$};
\end{tikzpicture}
&

\begin{tikzpicture}[scale=\sc, thick,main node/.style={circle, minimum size=5mm, inner sep=0.1mm,draw,font=\tiny\sffamily}]
\node[main node] at ({cos(\x+(0)*\z)},{sin(\x+(0)*\z)}) (1) {$1$};
\node[main node] at ({cos(\x+(1)*\z)},{sin(\x+(1)*\z)}) (2) {$2$};
\node[main node] at ({cos(\x+(2)*\z)},{sin(\x+(2)*\z)}) (3) {$3$};

\node[main node] at ({ \y* cos(\x+(3)*\z) + (1-\y)*cos(\x+(2)*\z) - (1-\y)*cos(\x+(4)*\z)},{ \y* sin(\x+(3)*\z) + (1-\y)*sin(\x+(2)*\z)}) (4) {$4_1$};
\node[main node] at ({cos(\x+(3)*\z) - (1-\y)*cos(\x+(4)*\z)},{sin(\x+(3)*\z)}) (5) {$4_0$};
\node[main node] at ({ \y* cos(\x+(3)*\z) },{ \y* sin(\x+(3)*\z) + (1-\y)*sin(\x+(4)*\z)}) (6) {$4_{21}$};
\node[main node] at ({cos(\x+(3)*\z) + (1-\y)*cos(\x+(4)*\z)},{sin(\x+(3)*\z)})(7) {$4_{20}$};
\node[main node] at ({ \y* cos(\x+(3)*\z) - (1-\y)*cos(\x+(2)*\z) + (1-\y)*cos(\x+(4)*\z)},{ \y* sin(\x+(3)*\z) + (1-\y)*sin(\x+(2)*\z)})  (8) {$4_{22}$};

 \path[every node/.style={font=\sffamily}]
(1) edge (2)
(2) edge (3)
(3) edge (1)
(4) edge (5)
(5) edge (6)
(4) edge (2)
(4) edge (3)
(6) edge (1)
(6) edge (7)
(8) edge (7)
(8) edge (2);
\end{tikzpicture}\\
$G_{2,2} \in \S_2(K_4)$ &
$H \in \S_2^2(K_4)$
\end{tabular}
\end{center}
\caption{A graph in $H \in \S_2^2(K_4)$ (right) where $\bar{e}^{\top}x \leq \alpha(H)$ is not a facet of $\STAB(H)$ (see Example~\ref{egSTABfacet})}\label{figSTABfacet}
\end{figure}

\section{Existence of $4$-minimal graphs}\label{sec6}

\begin{figure}[ht!]
\begin{center}
\def\x{1}
\def\y{3}

\begin{tabular}{c}
\begin{tikzpicture}
[scale=0.8, thick,main node/.style={circle, minimum size=4mm, inner sep=0.1mm,draw,font=\tiny\sffamily}]
\node[main node] at ({cos(90)*\x},{sin(90)*\x}) (1) {$1$};
\node[main node] at ({cos(210)*\x},{sin(210)*\x}) (2) {$2$};
\node[main node] at ({cos(330)*\x},{sin(330)*\x}) (3) {$3$};

\node[main node] at ({cos(90)*\x + cos(90)*\y} ,{sin(90)*\x+ sin(90)*\y}) (11) {$6_0$};
\node[main node] at ({cos(210)*\x+ cos(90)*\y},{sin(210)*\x+ sin(90)*\y}) (12) {$6_2$};
\node[main node] at ({cos(330)*\x+ cos(90)*\y},{sin(330)*\x+ sin(90)*\y}) (10) {$6_1$};

\node[main node] at ({cos(90)*\x + cos(210)*\y} ,{sin(90)*\x+ sin(210)*\y}) (4) {$4_1$};
\node[main node] at ({cos(210)*\x+ cos(210)*\y},{sin(210)*\x+ sin(210)*\y}) (5) {$4_0$};
\node[main node] at ({cos(330)*\x+ cos(210)*\y},{sin(330)*\x+ sin(210)*\y}) (6) {$4_2$};

\node[main node] at ({cos(90)*\x + cos(330)*\y} ,{sin(90)*\x+ sin(330)*\y}) (9) {$5_2$};
\node[main node] at ({cos(210)*\x+ cos(330)*\y},{sin(210)*\x+ sin(330)*\y}) (7) {$5_1$};
\node[main node] at ({cos(330)*\x+ cos(330)*\y},{sin(330)*\x+ sin(330)*\y}) (8) {$5_0$};

 \path[every node/.style={font=\sffamily}]
(1) edge (2)
(1) edge (3)
(2) edge (3)
(4) edge (5)
(5) edge (6)
(7) edge (8)
(8) edge (9)
(10) edge (11)
(11) edge (12)
(4) edge[bend right=30] (9)
(7) edge[bend right=30] (12)
(10) edge[bend right=30] (6)
(1) edge (4)
(2) edge (4)
(2) edge (6)
(3) edge (6)
(2) edge (7)
(3) edge (7)
(3) edge (9)
(1) edge (9)
(3) edge (10)
(1) edge (10)
(1) edge (12)
(2) edge (12);
\end{tikzpicture}\\
\end{tabular}
\caption{An alternative drawing of $G_{4,1}$ to highlight its automorphisms}\label{figG41alt}
\end{center}
\end{figure}
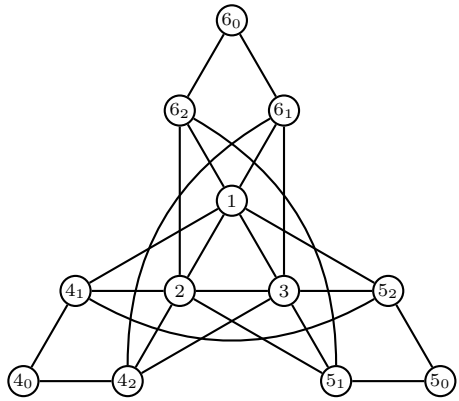

Recall the graph $G_{4,1}$ (Figure~\ref{figG41}), which was introduced in Section~\ref{sec1}. We show in this section that $r_+(G_{4,1})=4$, providing what we believe to be the first known example of a $4$-minimal graph (and the first advance in this direction since 2006~\cite{EscalanteMN06}). Notice that $G_{4,1}$ dovetails nicely with the structural results on $\ell$-minimal graphs developed in Sections~\ref{sec3} to~\ref{sec5} --- in particular, observe from its drawing in Figure~\ref{figG41} that $G_{4,1} \in \S_2^3(K_6)$, and is obtained from stretching three of the original vertices in $K_6$. On the other hand, while these structural results were helpful in our discovery of $G_{4,1}$ and putting the graph into perspective, the proof we present for $r_+(G_{4,1}) = 4$ is largely based on numerical certificates and the symmetries of $G_{4,1}$.

To that end, we point out two important automorphisms of $G_{4,1}$ that will be useful in simplifying our analysis of its $\LS_+$-rank. Consider the alternative drawing of $G_{4,1}$ in Figure~\ref{figG41alt}, and define the functions $f_1, f_2 : V(G_{4,1}) \to V(G_{4,1})$ as follows:
\begin{center}
\begin{tabular}{c|c|c|c|c|c|c|c|c|c|c|c|c}
$i$ & $1$ & $2$ & $3$ & $4_1$ & $4_0$ & $4_2$ & $5_1$ & $5_0$ & $5_2$ & $6_1$ & $6_0$ & $6_2$\\
\hline
$f_1(i)$ & $2$ & $3$ & $1$ & $5_1$ & $5_0$ & $5_2$ & $6_1$ & $6_0$ & $6_2$ & $4_1$ & $4_0$ & $4_2$\\
\hline
$f_2(i)$ & $1$ & $3$ & $2$ & $5_2$ & $5_0$ & $5_1$ & $4_2$ & $4_0$ & $4_1$ & $6_2$ & $6_0$ & $6_1$\\
\end{tabular}
\end{center}
 
Visually, $f_1$ corresponds to rotating the graph $G_{4,1}$ in Figure~\ref{figG41alt} counterclockwise by $\frac{2\pi}{3}$, and $f_2$ corresponds to reflecting the figure along the centre vertical line. Now we are ready to prove the main result of this section.

\begin{theorem}\label{thmG41}
The $\LS_+$-rank of $G_{4,1}$ is $4$.
\end{theorem}

\begin{proof}
For convenience, let $G \ce G_{4,1}$ throughout this proof. Since $G$ has $12$ vertices, by Theorem~\ref{thmNover3} it suffices to show that $r_+(G) \geq 4$. Consider the matrix $Y_0$ defined as follows:

{\scriptsize
\[
Y_0 \ce
\bbordermatrix{
&& 1 & 2 & 3 & 4_1 & 4_0 & 4_2 & 5_1 & 5_0 & 5_2 & 6_1 & 6_0 & 6_2 \cr
&100000 & 25340 & 25340 & 25340 & 16500 & 75020 & 16500 & 16500 & 75020 & 16500 & 16500 & 75020 & 16500\cr
&25340 & 25340 & 0 & 0 & 0 & 17502 & 7838 & 7838 & 17502 & 0 & 0 & 25340 & 0\cr
&25340 & 0 & 25340 & 0 & 0 & 25340 & 0 & 0 & 17502 & 7838 & 7838 & 17502 & 0\cr
&25340 & 0 & 0 & 25340 & 7838 & 17502 & 0 & 0 & 25340 & 0 & 0 & 17502 & 7838\cr
&16500 & 0 & 0 & 7838 & 16500 & 0 & 8073 & 589 & 15911 & 0 & 589 & 15419 & 1081\cr
&75020 & 17502 & 25340 & 17502 & 0 & 75020 & 0 & 15419 & 51150 & 15911 & 15911 & 51150 & 15419\cr
&16500 & 7838 & 0 & 0 & 8073 & 0 & 16500 & 1081 & 15419 & 589 & 0 & 15911 & 589\cr
&16500 & 7838 & 0 & 0 & 589 & 15419 & 1081 & 16500 & 0 & 8073 & 589 & 15911 & 0\cr
&75020 & 17502 & 17502 & 25340 & 15911 & 51150 & 15419 & 0 & 75020 & 0 & 15419 & 51150 & 15911\cr
&16500 & 0 & 7838 & 0 & 0 & 15911 & 589 & 8073 & 0 & 16500 & 1081 & 15419 & 589\cr
&16500 & 0 & 7838 & 0 & 589 & 15911 & 0 & 589 & 15419 & 1081 & 16500 & 0 & 8073\cr
&75020 & 25340 & 17502 & 17502 & 15419 & 51150 & 15911 & 15911 & 51150 & 15419 & 0 & 75020 & 0\cr
&16500 & 0 & 0 & 7838 & 1081 & 15419 & 589 & 0 & 15911 & 589 & 8073 & 0 & 16500 \cr
}.
\]}
Again, the columns of $Y_0$ are labelled by the vertices in $G$ they correspond to, with the rows of $Y_0$ following the same order of indexing.

We prove our claim by showing that $Y_0 \in \widehat{\LS}_+^3(G)$. First, one can check that $Y_0 \succeq 0$ (a $UV$-certificate is provided in Table~\ref{tabUV}). Moreover, observe that for all $i,j \in V(G)$,
\[
Y_0[i,j] = Y_0[f_1(i), f_1(j)] = Y_0[f_2(i), f_2(j)], 
\]
and thus the entries of $Y_0$ exhibit the same symmetries of the graph that are exposed by the automorphisms $f_1$ and $f_2$. Hence, to show that $Y_0 \in \widehat{\LS}_+^3(G)$, it suffices to verify the conditions $Y_0e_i, Y_0(e_0-e_i) \in \cone(\LS_+^2(G))$ for $i \in \set{1, 4_1, 6_0}$, since for every other vertex $j$ there is an automorphism of $G$ that would map $j$ to one of these three vertices.

Next, notice that
\begin{allowdisplaybreaks}
\begin{align*}
Y_0e_1 \leq{} &
17502 \begin{bmatrix} 1 \\ \chi_{\set{ 1, 4_0 ,5_0 , 6_0}} \end{bmatrix} + 
7838 \begin{bmatrix} 1 \\ \chi_{\set{ 1 , 4_2	, 5_1 , 6_0}} \end{bmatrix} 
,\\
Y_0e_{4_1} \leq{} &
7838 \begin{bmatrix} 1 \\ \chi_{\set{ 3 , 4_1 , 5_0 , 6_0}} \end{bmatrix} + 
589 \begin{bmatrix} 1 \\ \chi_{\set{ 4_1 , 4_2 , 5_1 , 6_0}} \end{bmatrix} + 
6992 \begin{bmatrix} 1 \\ \chi_{\set{ 4_1 , 4_2 , 5_0 , 6_0}} \end{bmatrix}\\
 & +
492 \begin{bmatrix} 1 \\ \chi_{\set{ 4_1 , 4_2 , 5_0 , 6_2}} \end{bmatrix} + 
589 \begin{bmatrix} 1 \\ \chi_{\set{ 4_1 , 5_0 , 6_1 , 6_2}} \end{bmatrix}
,\\
Y_0(e_0-e_{6_0}) \leq{} &
 7366 \begin{bmatrix} 1 \\ \chi_{\set{ 2 , 4_0 , 5_0 , 6_1}} \end{bmatrix} + 
 476 \begin{bmatrix} 1 \\ \chi_{\set{ 2 , 4_0 , 5_2 , 6_1}} \end{bmatrix} + 
 476 \begin{bmatrix} 1 \\ \chi_{\set{ 3 , 4_1 , 5_0 , 6_2}} \end{bmatrix} \\
 & + 7366 \begin{bmatrix} 1 \\ \chi_{\set{ 3 , 4_0 , 5_0 , 6_2}} \end{bmatrix} + 
 605 \begin{bmatrix} 1 \\ \chi_{\set{ 4_1 , 4_2 , 5_0 , 6_2}} \end{bmatrix} + 
 605 \begin{bmatrix} 1 \\ \chi_{\set{ 4_0 , 5_1 , 5_2 , 6_1}} \end{bmatrix}\\
 & +  8058 \begin{bmatrix} 1 \\ \chi_{\set{ 4_0 , 5_0 , 6_1 , 6_2}} \end{bmatrix} + 
28 \begin{bmatrix} 1 \\ 0 \end{bmatrix}.
\end{align*}
\end{allowdisplaybreaks}

Notice that all incidence vectors above correspond to stable sets in $G$. Since $\cone(\STAB(G))$ is a lower-comprehensive convex cone, we obtain that $Y_0e_1, Y_0e_{4_1}, Y_0(e_0-e_{6_0}) \in \cone(\STAB(G)) \subseteq \cone(\LS_+^2(G))$. The details for $Y_0(e_0 - e_1), Y_0e_{6_0},  Y_0(e_0 - e_{4_1}) \in \cone(\LS_+^2(G))$ are provided, respectively, in the proofs of Lemmas~\ref{lemG412},~\ref{lemG411}, and~\ref{lemG413} in Appendix~\ref{secA1}.

Finally, let $\bar{x}$ be the vector such that $Y_0e_0 = 100000 \begin{bmatrix} 1\\ \bar{x} \end{bmatrix}$. Since $Y_0 \in \widehat{\LS}_+^3(G)$, we have $\bar{x} \in \LS_+^3(G)$. Thus, we see that 
\[
\a_{\LS_+^3}(G) \geq \bar{e}^{\top} \bar{x} = 4.0008 > 4 = \a(G).
\]
This shows $\bar{x} \not\in \STAB(G)$ and thus $r_+(G) \geq 4$.
\end{proof}

Notice that $G_{4,1}$ contains $24$ edges. Also, notice that $G_{4,1} \in \S_2^3(K_6)$ and satisfies ($*$) from Proposition~\ref{propSkKlfacet}. Consequently, every graph $G \in \S_2^3(K_6)$ which is a subgraph of $G_{4,1}$ also satisfies ($*$), and so Proposition~\ref{propSkKlfacet} assures that the inequality $\bar{e}^{\top}x \leq 4$ is a facet of $\STAB(G)$. Thus, by Lemma~\ref{lem05subgraph}(ii) it follows that every graph in $\S_2^3(K_6)$ which is a subgraph of $G_{4,1}$ (which can have as few as $21$ edges) also has $\LS_+$-rank $4$, giving more examples of $4$-minimal graphs. The six non-isomorphic proper subgraphs of $G_{4,1}$ that belong to $\S_2^3(K_6)$ are listed in Figure~\ref{figG41subs}.

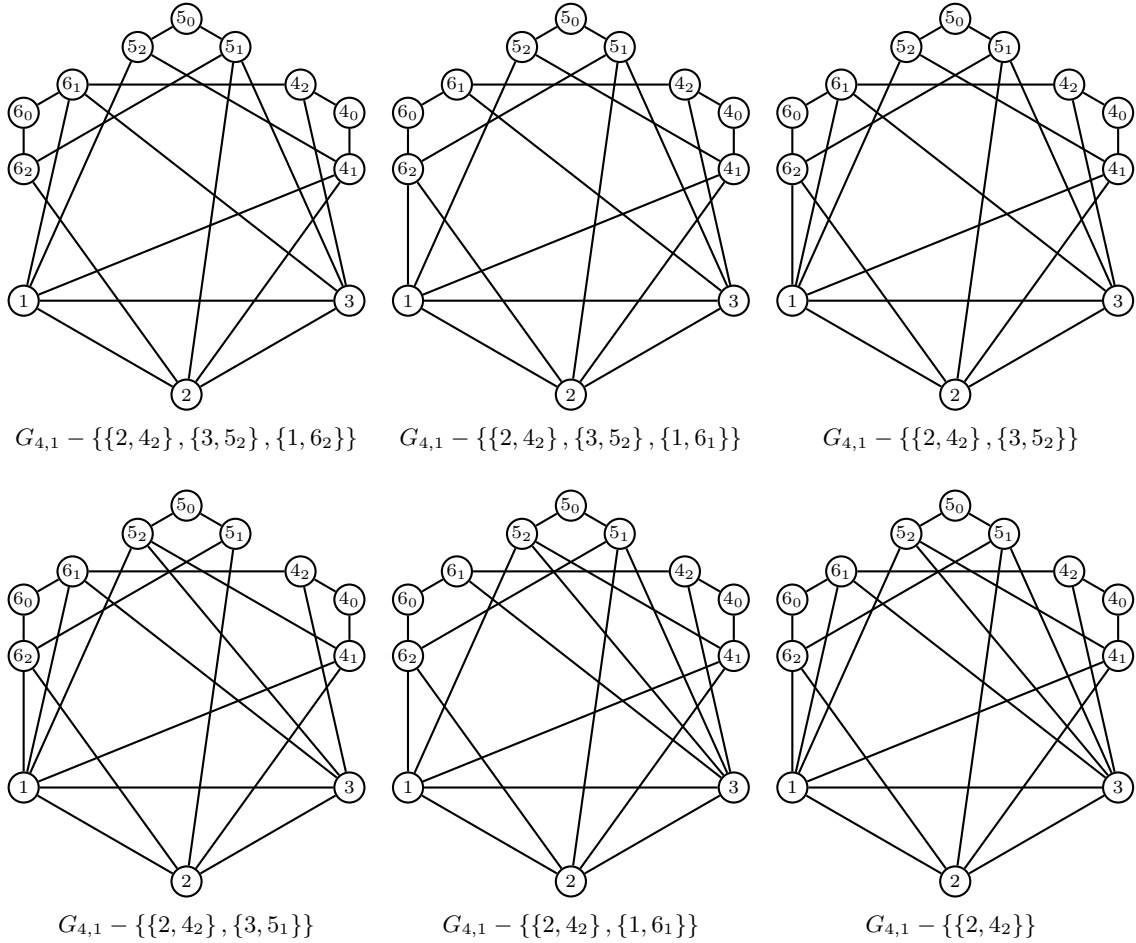
\begin{figure}[ht!]
\begin{center}
\def\x{270 - 360/6}
\def\z{360/6}
\def\y{0.7}
\def\sc{2.5}

\def\vp{
\node[main node] at ({cos(\x+(0)*\z)},{sin(\x+(0)*\z)}) (1) {$1$};
\node[main node] at ({cos(\x+(1)*\z)},{sin(\x+(1)*\z)}) (2) {$2$};
\node[main node] at ({cos(\x+(2)*\z)},{sin(\x+(2)*\z)}) (3) {$3$};

\node[main node] at ({ \y* cos(\x+(3)*\z) + (1-\y)*cos(\x+(2)*\z)},{ \y* sin(\x+(3)*\z) + (1-\y)*sin(\x+(2)*\z)}) (4) {$4_1$};
\node[main node] at ({cos(\x+(3)*\z)},{sin(\x+(3)*\z)}) (5) {$4_0$};
\node[main node] at ({ \y* cos(\x+(3)*\z) + (1-\y)*cos(\x+(4)*\z)},{ \y* sin(\x+(3)*\z) + (1-\y)*sin(\x+(4)*\z)}) (6) {$4_2$};

\node[main node] at ({ \y* cos(\x+(4)*\z) + (1-\y)*cos(\x+(3)*\z)},{ \y* sin(\x+(4)*\z) + (1-\y)*sin(\x+(3)*\z)}) (7) {$5_1$};
\node[main node] at ({cos(\x+(4)*\z)},{sin(\x+(4)*\z)}) (8) {$5_0$};
\node[main node] at ({ \y* cos(\x+(4)*\z) + (1-\y)*cos(\x+(5)*\z)},{ \y* sin(\x+(4)*\z) + (1-\y)*sin(\x+(5)*\z)}) (9) {$5_2$};

\node[main node] at ({ \y* cos(\x+(5)*\z) + (1-\y)*cos(\x+(4)*\z)},{ \y* sin(\x+(5)*\z) + (1-\y)*sin(\x+(4)*\z)}) (10) {$6_1$};
\node[main node] at ({cos(\x+(5)*\z)},{sin(\x+(5)*\z)}) (11) {$6_0$};
\node[main node] at ({ \y* cos(\x+(5)*\z) + (1-\y)*cos(\x+(6)*\z)},{ \y* sin(\x+(5)*\z) + (1-\y)*sin(\x+(6)*\z)}) (12) {$6_2$};

}
\def\ep{ 
\path[every node/.style={font=\sffamily}]
(1) edge (2)
(1) edge (3)
(2) edge (3)
(4) edge (5)
(5) edge (6)
(7) edge (8)
(8) edge (9)
(10) edge (11)
(11) edge (12)
(4) edge (9)
(7) edge (12)
(10) edge (6)
(1) edge (4)
(3) edge (6)
(2) edge (7)
(1) edge (9)
(3) edge (10)
(2) edge (12);
}

\begin{tabular}{ccc}
\begin{tikzpicture}[scale=\sc, thick,main node/.style={circle, minimum size=4mm, inner sep=0.1mm,draw,font=\tiny\sffamily}]
\vp \ep
\path[every node/.style={font=\sffamily}]
(2) edge (4)
(3) edge (7)
(1) edge (10)
;
\end{tikzpicture}
&

\begin{tikzpicture}[scale=\sc, thick,main node/.style={circle, minimum size=4mm, inner sep=0.1mm,draw,font=\tiny\sffamily}]
\vp \ep
\path[every node/.style={font=\sffamily}]
(2) edge (4)
(3) edge (7)
(1) edge (12);
\end{tikzpicture}
&

\begin{tikzpicture}[scale=\sc, thick,main node/.style={circle, minimum size=4mm, inner sep=0.1mm,draw,font=\tiny\sffamily}]
\vp \ep
\path[every node/.style={font=\sffamily}]
(2) edge (4)
(3) edge (7)
(1) edge (10)
(1) edge (12);
\end{tikzpicture}
\\
{\footnotesize $G_{4,1} - \set{ \set{2,4_2}, \set{3,5_2}, \set{1,6_2}}$ }&
{\footnotesize $G_{4,1} - \set{ \set{2,4_2}, \set{3,5_2}, \set{1,6_1}}$} &
{\footnotesize $G_{4,1} - \set{ \set{2,4_2}, \set{3,5_2}}$ }
\\
\\
\begin{tikzpicture}[scale=\sc, thick,main node/.style={circle, minimum size=4mm, inner sep=0.1mm,draw,font=\tiny\sffamily}]
\vp \ep
\path[every node/.style={font=\sffamily}]
(2) edge (4)
(3) edge (9)
(1) edge (10)
(1) edge (12);
\end{tikzpicture}
&

\begin{tikzpicture}[scale=\sc, thick,main node/.style={circle, minimum size=4mm, inner sep=0.1mm,draw,font=\tiny\sffamily}]
\vp \ep
\path[every node/.style={font=\sffamily}]
(2) edge (4)
(3) edge (7)
(3) edge (9)
(1) edge (12);
\end{tikzpicture}
&

\begin{tikzpicture}[scale=\sc, thick,main node/.style={circle, minimum size=4mm, inner sep=0.1mm,draw,font=\tiny\sffamily}]
\vp \ep
\path[every node/.style={font=\sffamily}]
(2) edge (4)
(3) edge (7)
(3) edge (9)
(1) edge (10)
(1) edge (12);
\end{tikzpicture}\\
{\footnotesize $G_{4,1} - \set{ \set{2,4_2}, \set{3,5_1}}$ }&
{\footnotesize $G_{4,1} - \set{ \set{2,4_2}, \set{1,6_1}}$} &
{\footnotesize $G_{4,1} - \set{ \set{2,4_2}}$ }

\end{tabular}

\caption{The six non-isomorphic proper subgraphs of $G_{4,1}$ that belong to $S_2^3(K_6)$ (and thus are $4$-minimal)}\label{figG41subs}
\end{center}
\end{figure}


Moreover, the fact that $G_{4,1}$ is $4$-minimal also provides some new examples of $3$-minimal graphs.

\begin{corollary}
Let $G_{3,2} \ce G_{4,1} \ominus 6_0$. Then $G_{3,2}$ is a $3$-minimal graph.
\end{corollary}

\begin{proof}
Since $G_{4,1}$ is $4$-minimal, there exists vertex $i \in V(G_{4,1})$ where $G \ominus i$ has $\LS_+$-rank $3$ and $\deg(i) = 2$. Thus, $i \in \set{4_0, 5_0, 6_0}$. Now observe that $G_{4,1} \ominus 4_0, G_{4,1} \ominus 5_0, G_{4,1} \ominus 6_0$ are all isomorphic to each other. Thus, $G_{3,2}$ is $3$-minimal.
\end{proof}

By Lemma~\ref{lem05subgraph}(ii) again, every graph in $\S_2^2(K_5)$ that is a subgraph of $G_{3,2}$ is $3$-minimal. Figure~\ref{figG32subs} illustrates $G_{3,2}$ (top left) and its five non-isomorphic proper subgraphs that belong to $\S_2^2(K_5)$. Notice that one of these graphs (top right of Figure~\ref{figG32subs}) is isomorphic to $G_{3,1}$, the first $3$-minimal graph discovered in~\cite{EscalanteMN06}.

\begin{figure}[ht!]
\begin{center}
\def\x{270 - 360/5}
\def\z{360/5}
\def\y{0.7}
\def\sc{2}
\def\vp{
\node[main node] at ({cos(\x+(0)*\z)},{sin(\x+(0)*\z)}) (1) {$1$};
\node[main node] at ({cos(\x+(1)*\z)},{sin(\x+(1)*\z)}) (2) {$2$};
\node[main node] at ({cos(\x+(2)*\z)},{sin(\x+(2)*\z)}) (3) {$3$};

\node[main node] at ({ \y* cos(\x+(3)*\z) + (1-\y)*cos(\x+(2)*\z)},{ \y* sin(\x+(3)*\z) + (1-\y)*sin(\x+(2)*\z)}) (4) {$4_1$};
\node[main node] at ({cos(\x+(3)*\z)},{sin(\x+(3)*\z)}) (5) {$4_0$};
\node[main node] at ({ \y* cos(\x+(3)*\z) + (1-\y)*cos(\x+(4)*\z)},{ \y* sin(\x+(3)*\z) + (1-\y)*sin(\x+(4)*\z)}) (6) {$4_2$};

\node[main node] at ({ \y* cos(\x+(4)*\z) + (1-\y)*cos(\x+(3)*\z)},{ \y* sin(\x+(4)*\z) + (1-\y)*sin(\x+(3)*\z)}) (7) {$5_1$};
\node[main node] at ({cos(\x+(4)*\z)},{sin(\x+(4)*\z)}) (8) {$5_0$};
\node[main node] at ({ \y* cos(\x+(4)*\z) + (1-\y)*cos(\x+(5)*\z)},{ \y* sin(\x+(4)*\z) + (1-\y)*sin(\x+(5)*\z)}) (9) {$5_2$};
}
\def\ep{ 
\path[every node/.style={font=\sffamily}]
(1) edge (2)
(1) edge (3)
(2) edge (3)
(4) edge (5)
(5) edge (6)
(7) edge (8)
(8) edge (9)
(4) edge (9)
(1) edge (4)
(3) edge (6)
(2) edge (7)
(1) edge (9);
}

\begin{tabular}{ccc}
\begin{tikzpicture}[scale=\sc, thick,main node/.style={circle, minimum size=4mm, inner sep=0.1mm,draw,font=\tiny\sffamily}]
\vp \ep
\path[every node/.style={font=\sffamily}]
(2) edge (4)
(2) edge (6)
(3) edge (7)
(3) edge (9)
;
\end{tikzpicture}
&
\begin{tikzpicture}[scale=\sc, thick,main node/.style={circle, minimum size=4mm, inner sep=0.1mm,draw,font=\tiny\sffamily}]
\vp \ep
\path[every node/.style={font=\sffamily}]
(2) edge (4)
(3) edge (7)
;
\end{tikzpicture}
&
\begin{tikzpicture}[scale=\sc, thick,main node/.style={circle, minimum size=4mm, inner sep=0.1mm,draw,font=\tiny\sffamily}]
\vp \ep
\path[every node/.style={font=\sffamily}]
(2) edge (4)
(3) edge (9)
;
\end{tikzpicture}
\\
{\footnotesize $G_{3,2}$ }&
{\footnotesize $G_{3,2} - \set{ \set{2,4_2}, \set{3,5_2}}$} &
{\footnotesize $G_{3,2} - \set{ \set{2,4_2}, \set{3,5_1}}$ }
\\
\\
\begin{tikzpicture}[scale=\sc, thick,main node/.style={circle, minimum size=4mm, inner sep=0.1mm,draw,font=\tiny\sffamily}]
\vp \ep
\path[every node/.style={font=\sffamily}]
(2) edge (4)
(3) edge (7)
(3) edge (9)
;
\end{tikzpicture}
&
\begin{tikzpicture}[scale=\sc, thick,main node/.style={circle, minimum size=4mm, inner sep=0.1mm,draw,font=\tiny\sffamily}]
\vp \ep
\path[every node/.style={font=\sffamily}]
(2) edge (6)
(3) edge (7)
;
\end{tikzpicture}
&
\begin{tikzpicture}[scale=\sc, thick,main node/.style={circle, minimum size=4mm, inner sep=0.1mm,draw,font=\tiny\sffamily}]
\vp \ep
\path[every node/.style={font=\sffamily}]
(2) edge (6)
(3) edge (7)
(3) edge (9)
;
\end{tikzpicture}
\\
{\footnotesize $G_{3,2}- \set{ \set{2,4_2}}$}&
{\footnotesize $G_{3,2} - \set{ \set{2,4_1}, \set{3,5_2}}$} &
{\footnotesize $G_{3,2} - \set{ \set{2,4_2}}$ }
\end{tabular}

\caption{The graph $G_{3,2}$ (top left) and its five non-isomorphic proper subgraphs that belong to $S_2^2(K_5)$ (and thus are $3$-minimal)}\label{figG32subs}
\end{center}
\end{figure}
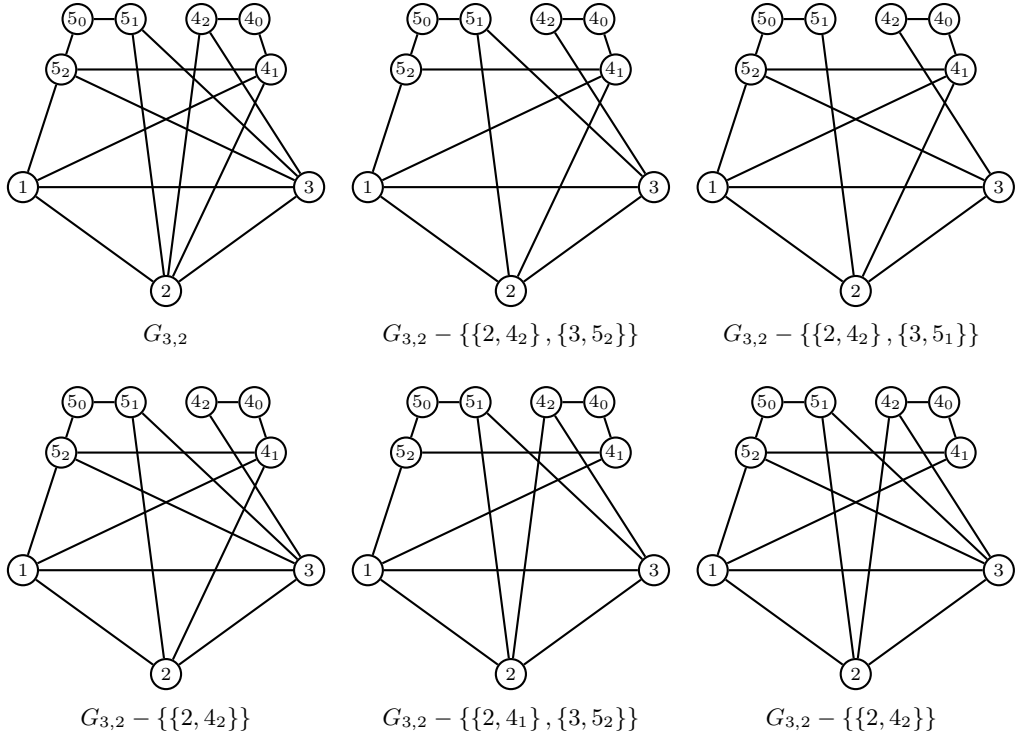


We close the section by showing that there are no $3$-minimal graphs with fewer edges than $G_{3,1}$.

\begin{proposition}\label{propSparsest3min}
Suppose $G$ is a $3$-minimal graph. Then $|E(G)| \geq 14$.
\end{proposition}

\begin{proof}
Since $G$ is $3$-minimal, there must exist vertex $v_0$ where $r_+(G \ominus v_0) = 2$. This implies that $|V(G \ominus v_0)| \geq 6$, and thus $\deg(v_0) \leq 2$. Since $\ell$-minimal graphs cannot have cut vertices, we see that $\deg(v_0) = 2$ and $|V(G \ominus v)|=6$, and so $G \ominus v$ is isomorphic to either $G_{2,1}$ (8 edges) or $G_{2,2}$ (9 edges).

Let $v_1, v_2$ be the two neighbours of $v_0$, and let $A \ce \set{v_0, v_1, v_2}$ and $B \ce V(G) \setminus A$. Observe that 
\begin{equation}\label{propSparsest3mineq1}
|E(G)| = \delta(A) + \delta(B) + \delta(A,B).
\end{equation}
 Since $|E(G)| \leq 13, \delta(A) = 2$, and $\delta(B) \geq 8$, we obtain $\delta(A,B) \leq 3$. Again, $G$ being $3$-minimal implies that $\deg(v_1), \deg(v_2) \geq 2$, and so we obtain that $2 \leq \delta( \set{v_1}, B) + \delta(\set{v_2}, B) \leq 3$. Thus, we may assume without loss of generality that $\delta( \set{v_1}, B)=1$, and let $u$ be the only neighbour of $v_1$ in $B$.

If $\delta(\set{v_2}, B) = 1$, then $u, v_1, v_0, v_2$ form a sparse path of length $3$ (with $\deg(v_1) = \deg(v_0) = \deg(v_2) = 2$), and Proposition~\ref{propSparsePath} implies that $G$ is not $3$-minimal. Now suppose $\delta(\set{v_2}, B) =2$. This means that $\delta(A,B) = 3$, and so from~\eqref{propSparsest3mineq1} we know that $|E(G)| =13, \delta(B) = 8$, and $G-A$ is indeed isomorphic to $G_{2,1}$ and not $G_{2,2}$.

Next, since $r_+(G) = 3$, we obtain that $r_+(G - u) \geq 2$. However, notice that $v_0$ is a cut vertex in $G-u$. Thus, if we let $A' \ce \set{u,v_1,v_0}$ and $B' \ce V(G) \setminus A'$, then we see that $r_+(G - A') \geq 2$. Since $G - A'$ has $6$ vertices, it must be isomorphic to $G_{2,1}$ or $G_{2,2}$. Thus, we see that $\delta(B') \geq 8$. Also, $\delta(A') = 2$ and 
\[
\delta(A',B') = \delta(\set{v_0}, B') + \delta(\set{u}, B')= 1 + (\deg(u)-1) = \deg(u).
\]
Since $13 = |E(G)| = \delta(A') + \delta(A',B') + \delta(B')$, we obtain that $\deg(u) = \delta(A',B') = 3$, and $\delta(B') = 8$. Thus, $G- A'$ is also isomorphic to $G_{2,1}$ and not $G_{2,2}$. For both $G-A$ and $G-A'$ to be isomorphic to $G_{2,1}$, $v_2$ must be adjacent to the two neighbours of $u$ in $G \ominus v_0$. Thus, $G$ is isomorphic to the graph shown in Figure~\ref{figSparsest3min}.

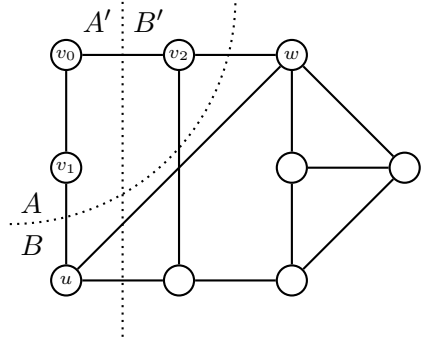
\begin{figure}[ht!]
\begin{center}
\begin{tikzpicture}
[scale=1.5, thick,main node/.style={circle, minimum size=4mm, inner sep=0.1mm,draw,font=\tiny\sffamily}]
\node[main node] at (3,1) (1) {$~$};
\node[main node] at (2,1) (2) {$~$};
\node[main node] at (2,2) (3) {$w$};

\node[main node] at (1,2) (4) {$v_2$};
\node[main node] at (0,2) (5) {$v_0$};
\node[main node] at (0,1) (6) {$v_1$};

\node[main node] at (0,0) (7) {$u$};
\node[main node] at (1,0) (8) {$~$};
\node[main node] at (2,0) (9) {$~$};

  \path[every node/.style={font=\sffamily}]
(1) edge (2)
(1) edge (3)
(2) edge (3)
(4) edge (5)
(5) edge (6)
(7) edge (8)
(8) edge (9)
(4) edge (3)
(4) edge (8)
(6) edge (7)
(7) edge (3)
(9) edge (2)
(9) edge (1)
;

\node[anchor=north east] at (0.5,2.5) () {$A'$};
\node[anchor=north west] at (0.5,2.5) () {$B'$};

\node[anchor=south west] at (-0.5,0.5) () {$A$};
\node[anchor=north west] at (-0.5,0.5) () {$B$};

\draw[dotted] (0.5,-0.5)--(0.5,2.5);
\draw[dotted] (-0.5, 0.5) edge[bend right = 45] (1.5,2.5);

\end{tikzpicture}
\caption{Illustrating the proof of Proposition~\ref{propSparsest3min}}\label{figSparsest3min}
\end{center}
\end{figure}

However, notice that $G - w$ has $\LS_+$-rank $1$, which contradicts $r_+(G) = 3$. This completes the proof.
\end{proof}

\section{Revisiting $H_k$ and constructing sparse graphs with high $\LS_+$-rank}\label{sec7}

In this section, we revisit the graphs $H_k$ defined in Section~\ref{sec1}, and obtain other related graphs with high $\LS_+$-ranks by applying some of our results on vertex stretching. First, we point out that the $\LS_+$-rank lower bound in Theorem~\ref{thmHk} also applies to some particular subgraphs of $H_k$. For every $k \geq 3$, define 
\[
H_k' \ce H_k - \set{1_0, 1_2, 2_0, 2_1}.
\]
\def\y{0.70}
\def\sc{2}
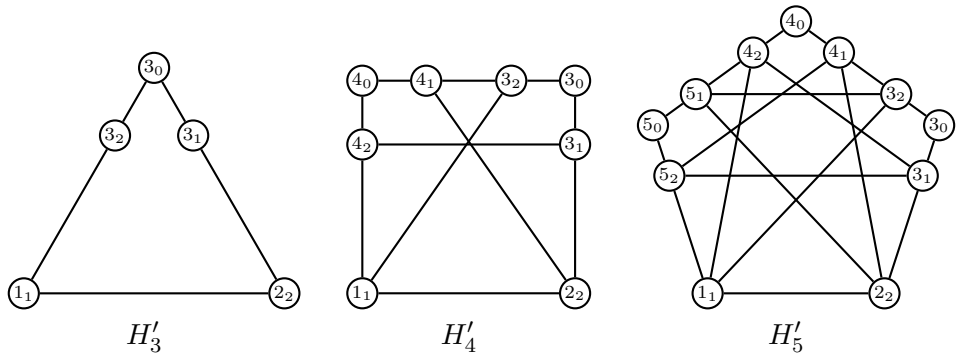
\begin{figure}[ht!]
\begin{center}
\begin{tabular}{ccc}

\def\x{270 - 180/3}
\def\z{360/3}
\begin{tikzpicture}[scale=\sc, thick,main node/.style={circle, minimum size=4mm, inner sep=0.1mm,draw,font=\tiny\sffamily}]
\node[main node] at ({cos(\x+(0)*\z)},{sin(\x+(0)*\z)}) (1) {$1_1$};
\node[main node] at ({cos(\x+(1)*\z)},{sin(\x+(1)*\z)}) (6) {$2_2$};

\node[main node] at ({ \y* cos(\x+(2)*\z) + (1-\y)*cos(\x+(1)*\z)},{ \y* sin(\x+(2)*\z) + (1-\y)*sin(\x+(1)*\z)}) (7) {$3_1$};
\node[main node] at ({cos(\x+(2)*\z)},{sin(\x+(2)*\z)}) (8) {$3_0$};
\node[main node] at ({ \y* cos(\x+(2)*\z) + (1-\y)*cos(\x+(3)*\z)},{ \y* sin(\x+(2)*\z) + (1-\y)*sin(\x+(3)*\z)}) (9) {$3_2$};

 \path[every node/.style={font=\sffamily}]
(8) edge (7)
(8) edge (9)
(1) edge (6)
(1) edge (9)
(7) edge (6);
\end{tikzpicture}

&

\def\x{270 - 180/4}
\def\z{360/4}

\begin{tikzpicture}[scale=\sc, thick,main node/.style={circle, minimum size=4mm, inner sep=0.1mm,draw,font=\tiny\sffamily}]

\node[main node] at ({cos(\x+(0)*\z)},{sin(\x+(0)*\z)}) (1) {$1_1$};
\node[main node] at ({cos(\x+(1)*\z)},{sin(\x+(1)*\z)}) (6) {$2_2$};

\node[main node] at ({ \y* cos(\x+(2)*\z) + (1-\y)*cos(\x+(1)*\z)},{ \y* sin(\x+(2)*\z) + (1-\y)*sin(\x+(1)*\z)}) (7) {$3_1$};
\node[main node] at ({cos(\x+(2)*\z)},{sin(\x+(2)*\z)}) (8) {$3_0$};
\node[main node] at ({ \y* cos(\x+(2)*\z) + (1-\y)*cos(\x+(3)*\z)},{ \y* sin(\x+(2)*\z) + (1-\y)*sin(\x+(3)*\z)}) (9) {$3_2$};

\node[main node] at ({ \y* cos(\x+(3)*\z) + (1-\y)*cos(\x+(2)*\z)},{ \y* sin(\x+(3)*\z) + (1-\y)*sin(\x+(2)*\z)}) (10) {$4_1$};
\node[main node] at ({cos(\x+(3)*\z)},{sin(\x+(3)*\z)}) (11) {$4_0$};
\node[main node] at ({ \y* cos(\x+(3)*\z) + (1-\y)*cos(\x+(4)*\z)},{ \y* sin(\x+(3)*\z) + (1-\y)*sin(\x+(4)*\z)}) (12) {$4_2$};

 \path[every node/.style={font=\sffamily}]
(8) edge (7)
(8) edge (9)
(11) edge (10)
(11) edge (12)
(1) edge (6)
(1) edge (9)
(1) edge (12)
(7) edge (6)
(7) edge (12)
(10) edge (6)
(10) edge (9);
\end{tikzpicture}

&

\def\x{270 - 180/5}
\def\z{360/5}

\begin{tikzpicture}[scale=\sc, thick,main node/.style={circle, minimum size=4mm, inner sep=0.1mm,draw,font=\tiny\sffamily}]
\node[main node] at ({cos(\x+(0)*\z)},{sin(\x+(0)*\z)}) (1) {$1_1$};
\node[main node] at ({cos(\x+(1)*\z)},{sin(\x+(1)*\z)}) (6) {$2_2$};

\node[main node] at ({ \y* cos(\x+(2)*\z) + (1-\y)*cos(\x+(1)*\z)},{ \y* sin(\x+(2)*\z) + (1-\y)*sin(\x+(1)*\z)}) (7) {$3_1$};
\node[main node] at ({cos(\x+(2)*\z)},{sin(\x+(2)*\z)}) (8) {$3_0$};
\node[main node] at ({ \y* cos(\x+(2)*\z) + (1-\y)*cos(\x+(3)*\z)},{ \y* sin(\x+(2)*\z) + (1-\y)*sin(\x+(3)*\z)}) (9) {$3_2$};

\node[main node] at ({ \y* cos(\x+(3)*\z) + (1-\y)*cos(\x+(2)*\z)},{ \y* sin(\x+(3)*\z) + (1-\y)*sin(\x+(2)*\z)}) (10) {$4_1$};
\node[main node] at ({cos(\x+(3)*\z)},{sin(\x+(3)*\z)}) (11) {$4_0$};
\node[main node] at ({ \y* cos(\x+(3)*\z) + (1-\y)*cos(\x+(4)*\z)},{ \y* sin(\x+(3)*\z) + (1-\y)*sin(\x+(4)*\z)}) (12) {$4_2$};

\node[main node] at ({ \y* cos(\x+(4)*\z) + (1-\y)*cos(\x+(3)*\z)},{ \y* sin(\x+(4)*\z) + (1-\y)*sin(\x+(3)*\z)}) (13) {$5_1$};
\node[main node] at ({cos(\x+(4)*\z)},{sin(\x+(4)*\z)}) (14) {$5_0$};
\node[main node] at ({ \y* cos(\x+(4)*\z) + (1-\y)*cos(\x+(5)*\z)},{ \y* sin(\x+(4)*\z) + (1-\y)*sin(\x+(5)*\z)}) (15) {$5_2$};

 \path[every node/.style={font=\sffamily}]
(8) edge (7)
(8) edge (9)
(11) edge (10)
(11) edge (12)
(14) edge (13)
(14) edge (15)
(1) edge (6)
(1) edge (9)
(1) edge (12)
(1) edge (15)
(7) edge (6)
(7) edge (12)
(7) edge (15)
(10) edge (6)
(10) edge (9)
(10) edge (15)
(13) edge (6)
(13) edge (9)
(13) edge (12);
\end{tikzpicture}
\\
$H_3'$ & $H_4'$ & $H_5'$ 
\end{tabular}
\caption{Several graphs in the family $H_k'$}\label{figH_k'}
\end{center}
\end{figure}

Figure~\ref{figH_k'} illustrates the graphs $H_k'$ for $k=3,4,5$. Notice that $H_k' \in \S_2^{k-2}(K_k)$ for all $k \geq 3$ --- this is apparent if one takes the drawings of $H_k'$ from Figure~\ref{figH_k'} and relabels the vertices $1_1$ and $2_2$ by $1$ and $2$ respectively. Then we have the following.

\begin{proposition}\label{propHk'}
For every $k \geq 3$, $r_+(H_k') \geq \frac{3k}{16}$.
\end{proposition}

\begin{proof}
For convenience, let $p \ce \left\lceil \frac{3k}{16} \right\rceil -1$ throughout this proof. Also, given $a,b \in \mR$, we define the vector $w_k(a,b) \in \mR^{V(H_k)}$ such that
\[
[w_k(a,b)]_j \ce \begin{cases}
a & \tn{if $j \in \set{i_1, i_2: i \in [k]}$;}\\
b & \tn{if $j \in \set{i_0 : i \in [k]}$.}
\end{cases}
\]
In~\cite{AuT24}, it was shown that there exists real numbers $a,b$ where $w_k(a,b)$ is contained in $\LS_+^p(H_k)$ and violates the inequality
\begin{equation}\label{lem62eq0}
w_k(k-1, k-2)^{\top} x \leq k(k-1),
\end{equation}
which is valid for $\STAB(H_k)$~\cite[Lemma 9(ii)]{AuT24}.

Now let $w_k(a,b)' \in \mR^{V(H_k')}$ be the vector obtained from $w_k(a,b)$ by removing the four entries that correspond to vertices which are not in $H_k'$. Then by Lemma~\ref{lemfacet}, we have $w_k(a,b)' \in \LS_+^p(H_k')$. On the other hand, the fact that $w_k(a,b)$ violates~\eqref{lem62eq0} implies that $(k-1)(2ka) + (k-2)(kb) > k(k-1)$, which implies that 
\[
\a_{\LS_+^p}(H_k') \geq \bar{e}^{\top} w_k(a,b)' = (2k-2)a + (k-2)b > k-1.
\]
However, since $H_k' \in \S_2^{k-2}(K_k)$, it follows from Lemma~\ref{05stretchalpha} that $\a(H_k') = k-1$. This implies that $w_k(a,b)' \in \LS_+^p(H_k') \setminus \STAB(H_k')$, and that $r_+(H_k') \geq p+1 \geq \frac{3k}{16}$.
\end{proof}

In fact, we can use the argument above to find many subgraphs of $H_k'$ for which the $\LS_+$-rank lower bound given in Proposition~\ref{propHk'} applies.

\begin{proposition}\label{propH_kSubgraph}
Let $G \in \S_2^{k-2}(K_k)$ be a subgraph of $H_k'$. Then $r_+(G) \geq \frac{3k}{16}$.
\end{proposition}

\begin{proof}
Again, let $p \ce \left\lceil \frac{3k}{16} \right\rceil -1$, and let $G \in \S_2^{k-2}(K_k)$ be a subgraph of $H_k'$.  Since $\a_{\LS_+^p}(H_k') > k-1$ (as shown in the proof of Proposition~\ref{propHk'},  Lemma~\ref{lem05subgraph}(ii) implies that $\a_{\LS_+^{p}}(G) > k-1$.  But then Lemma~\ref{05stretchalpha} implies that $\a(G) = k-1$. Thus, $r_+(G) \geq p+1 \geq \frac{3k}{16}$.
\end{proof}

Given a graph $G$, define the \emph{edge density} of $G$ to be $d(G) \ce \frac{|E(G)|}{ \binom{|V(G)|}{2}}$. For instance, $d(G) = 1$ for complete graphs, and $d(G) = 0$ for empty graphs. An interesting contrast that has emerged in the study of lift-and-project relaxations of the stable set polytope of graphs is that dense graphs tend to have high lift-and-project ranks with respect to operators that produce polyhedral relaxations, whereas graphs from both ends of the density spectrum tend to be of small lift-and-project ranks with respect to semidefinite operators. Thus, it is interesting to note that
\[
d(H_k') = \frac{k^2-k-1}{\binom{3k-4}{2}} = \frac{2}{9} + o(k).
\]
Moreover, it follows from Proposition~\ref{propH_kSubgraph} that the $\LS_+$-rank lower bound we showed for $H_k'$ also applies for many subgraphs of $H_k'$ with lower edge densities. For an example, given $k \geq 3$, we define the graph $H_k''$ where $V(H_k'') \ce V(H_k')$, with $E(H_k'')$ consisting of the following edges:
\begin{itemize}
\item[(i)]
$\set{1_1, 2_2}$;
\item[(ii)]
$\set{1_1, i_2}, \set{2_2, i_1}, \set{i_0, i_1}$, and $\set{i_0,i_2}$ for every $i \in \set{3, \ldots, k}$;
\item[(iii)]
$\set{i_2,j_1}$ for all $i,j \in \set{3,\ldots,k}$ where $(j-i)~\tn{mod}~(k-2) < \frac{k-2}{2}$;
\item[(iv)]
$\set{i_2,j_1}$ for all $i,j \in \set{3,\ldots,k}$ where $j-i = \frac{k-2}{2}$.
\end{itemize}
Observe that (iv) only contributes edges when $k$ is even. Also, for every $k \geq 3$, notice that $H_k''$ is a subgraph of $H_k'$, and that $H_k'' \in \S_2^{k-2}(K_k)$ (see Figures~\ref{figH_k'Subgraph} and~\ref{figH_k''Stretch}, respectively, for drawings of $H_5''$ and $H_6''$). Furthermore, $H_k''$ has the fewest edges among all graphs in $\S_2^{k-2}(K_k)$. To see this, suppose we start with a complete graph $K_k$ with vertex labels $1_1, 2_2, 3, 4, \ldots, k$, and $2$-stretch the vertices $3,4, \ldots, k$ to obtain a graph $G \in \S_2^{k-2}(K_k)$. If we define the sets $S_1 \ce \set{1_1}, S_2 \ce \set{2_2}$, and $S_i \ce \set{i_0, i_1, i_2}$ for all $i \in \set{3, \ldots, k}$, then there must be at least one edge in $G$ joining $S_i$ and $S_j$ for all distinct $i,j \in [k]$. To minimize the number of edges in $G$, one can ensure that the sets $A_1, A_2$ are disjoint in each vertex stretching operation. This would result in a graph with exactly one edge joining $S_i, S_j$ for all distinct $i,j \in [k]$, which is indeed the case for $H_k''$. 

It is easy to check that $|E(H_k'')| = \frac{k^2+3k-8}{2}$, and thus $d(H_k'') = \frac{1}{9} + o(k)$. Thus, we see that there are many subgraphs of $H_k'$ with edge densities between $\frac{1}{9}$ and $\frac{2}{9}$ for which the rank lower bound in Proposition~\ref{propH_kSubgraph} applies.

\def\x{270 - 180/5}
\def\z{360/5}
\def\y{0.70}
\def\sc{2}

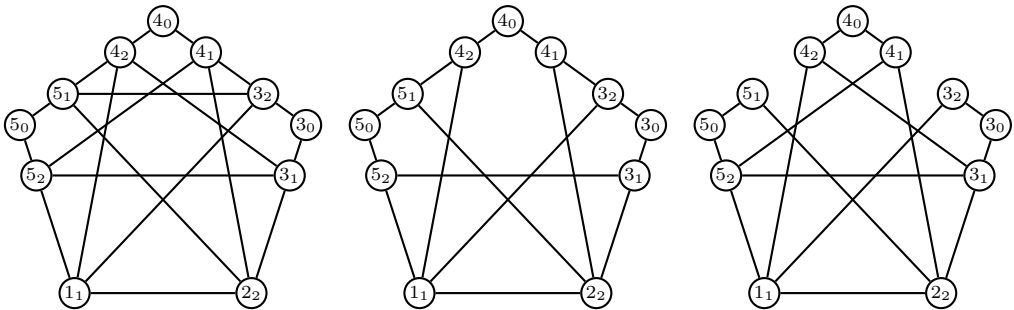
\begin{figure}[ht!]
\begin{center}
\begin{tabular}{ccc}

\begin{tikzpicture}[scale=\sc, thick,main node/.style={circle, minimum size=4mm, inner sep=0.1mm,draw,font=\tiny\sffamily}]
\node[main node] at ({cos(\x+(0)*\z)},{sin(\x+(0)*\z)}) (1) {$1_1$};
\node[main node] at ({cos(\x+(1)*\z)},{sin(\x+(1)*\z)}) (6) {$2_2$};

\node[main node] at ({ \y* cos(\x+(2)*\z) + (1-\y)*cos(\x+(1)*\z)},{ \y* sin(\x+(2)*\z) + (1-\y)*sin(\x+(1)*\z)}) (7) {$3_1$};
\node[main node] at ({cos(\x+(2)*\z)},{sin(\x+(2)*\z)}) (8) {$3_0$};
\node[main node] at ({ \y* cos(\x+(2)*\z) + (1-\y)*cos(\x+(3)*\z)},{ \y* sin(\x+(2)*\z) + (1-\y)*sin(\x+(3)*\z)}) (9) {$3_2$};

\node[main node] at ({ \y* cos(\x+(3)*\z) + (1-\y)*cos(\x+(2)*\z)},{ \y* sin(\x+(3)*\z) + (1-\y)*sin(\x+(2)*\z)}) (10) {$4_1$};
\node[main node] at ({cos(\x+(3)*\z)},{sin(\x+(3)*\z)}) (11) {$4_0$};
\node[main node] at ({ \y* cos(\x+(3)*\z) + (1-\y)*cos(\x+(4)*\z)},{ \y* sin(\x+(3)*\z) + (1-\y)*sin(\x+(4)*\z)}) (12) {$4_2$};

\node[main node] at ({ \y* cos(\x+(4)*\z) + (1-\y)*cos(\x+(3)*\z)},{ \y* sin(\x+(4)*\z) + (1-\y)*sin(\x+(3)*\z)}) (13) {$5_1$};
\node[main node] at ({cos(\x+(4)*\z)},{sin(\x+(4)*\z)}) (14) {$5_0$};
\node[main node] at ({ \y* cos(\x+(4)*\z) + (1-\y)*cos(\x+(5)*\z)},{ \y* sin(\x+(4)*\z) + (1-\y)*sin(\x+(5)*\z)}) (15) {$5_2$};

 \path[every node/.style={font=\sffamily}]
(1) edge (6)
(8) edge (7)
(8) edge (9)
(11) edge (10)
(11) edge (12)
(14) edge (13)
(14) edge (15)
(1) edge (9)
(1) edge (12)
(1) edge (15)
(6) edge (7)
(6) edge (10)
(6) edge (13)
(7) edge (12)
(7) edge (15)
(10) edge (9)
(10) edge (15)
(13) edge (9)
(13) edge (12);
\end{tikzpicture}
&

\begin{tikzpicture}[scale=\sc, thick,main node/.style={circle, minimum size=4mm, inner sep=0.1mm,draw,font=\tiny\sffamily}]
\node[main node] at ({cos(\x+(0)*\z)},{sin(\x+(0)*\z)}) (1) {$1_1$};
\node[main node] at ({cos(\x+(1)*\z)},{sin(\x+(1)*\z)}) (6) {$2_2$};

\node[main node] at ({ \y* cos(\x+(2)*\z) + (1-\y)*cos(\x+(1)*\z)},{ \y* sin(\x+(2)*\z) + (1-\y)*sin(\x+(1)*\z)}) (7) {$3_1$};
\node[main node] at ({cos(\x+(2)*\z)},{sin(\x+(2)*\z)}) (8) {$3_0$};
\node[main node] at ({ \y* cos(\x+(2)*\z) + (1-\y)*cos(\x+(3)*\z)},{ \y* sin(\x+(2)*\z) + (1-\y)*sin(\x+(3)*\z)}) (9) {$3_2$};

\node[main node] at ({ \y* cos(\x+(3)*\z) + (1-\y)*cos(\x+(2)*\z)},{ \y* sin(\x+(3)*\z) + (1-\y)*sin(\x+(2)*\z)}) (10) {$4_1$};
\node[main node] at ({cos(\x+(3)*\z)},{sin(\x+(3)*\z)}) (11) {$4_0$};
\node[main node] at ({ \y* cos(\x+(3)*\z) + (1-\y)*cos(\x+(4)*\z)},{ \y* sin(\x+(3)*\z) + (1-\y)*sin(\x+(4)*\z)}) (12) {$4_2$};

\node[main node] at ({ \y* cos(\x+(4)*\z) + (1-\y)*cos(\x+(3)*\z)},{ \y* sin(\x+(4)*\z) + (1-\y)*sin(\x+(3)*\z)}) (13) {$5_1$};
\node[main node] at ({cos(\x+(4)*\z)},{sin(\x+(4)*\z)}) (14) {$5_0$};
\node[main node] at ({ \y* cos(\x+(4)*\z) + (1-\y)*cos(\x+(5)*\z)},{ \y* sin(\x+(4)*\z) + (1-\y)*sin(\x+(5)*\z)}) (15) {$5_2$};

 \path[every node/.style={font=\sffamily}]
(1) edge (6)
(8) edge (7)
(8) edge (9)
(11) edge (10)
(11) edge (12)
(14) edge (13)
(14) edge (15)
(1) edge (9)
(1) edge (12)
(1) edge (15)
(6) edge (7)
(6) edge (10)
(6) edge (13)
(9) edge (10)
(12) edge (13)
(7) edge (15);
\end{tikzpicture}
&

\begin{tikzpicture}[scale=\sc, thick,main node/.style={circle, minimum size=4mm, inner sep=0.1mm,draw,font=\tiny\sffamily}]
\node[main node] at ({cos(\x+(0)*\z)},{sin(\x+(0)*\z)}) (1) {$1_1$};
\node[main node] at ({cos(\x+(1)*\z)},{sin(\x+(1)*\z)}) (6) {$2_2$};

\node[main node] at ({ \y* cos(\x+(2)*\z) + (1-\y)*cos(\x+(1)*\z)},{ \y* sin(\x+(2)*\z) + (1-\y)*sin(\x+(1)*\z)}) (7) {$3_1$};
\node[main node] at ({cos(\x+(2)*\z)},{sin(\x+(2)*\z)}) (8) {$3_0$};
\node[main node] at ({ \y* cos(\x+(2)*\z) + (1-\y)*cos(\x+(3)*\z)},{ \y* sin(\x+(2)*\z) + (1-\y)*sin(\x+(3)*\z)}) (9) {$3_2$};

\node[main node] at ({ \y* cos(\x+(3)*\z) + (1-\y)*cos(\x+(2)*\z)},{ \y* sin(\x+(3)*\z) + (1-\y)*sin(\x+(2)*\z)}) (10) {$4_1$};
\node[main node] at ({cos(\x+(3)*\z)},{sin(\x+(3)*\z)}) (11) {$4_0$};
\node[main node] at ({ \y* cos(\x+(3)*\z) + (1-\y)*cos(\x+(4)*\z)},{ \y* sin(\x+(3)*\z) + (1-\y)*sin(\x+(4)*\z)}) (12) {$4_2$};

\node[main node] at ({ \y* cos(\x+(4)*\z) + (1-\y)*cos(\x+(3)*\z)},{ \y* sin(\x+(4)*\z) + (1-\y)*sin(\x+(3)*\z)}) (13) {$5_1$};
\node[main node] at ({cos(\x+(4)*\z)},{sin(\x+(4)*\z)}) (14) {$5_0$};
\node[main node] at ({ \y* cos(\x+(4)*\z) + (1-\y)*cos(\x+(5)*\z)},{ \y* sin(\x+(4)*\z) + (1-\y)*sin(\x+(5)*\z)}) (15) {$5_2$};
 \path[every node/.style={font=\sffamily}]
(1) edge (6)
(8) edge (7)
(8) edge (9)
(11) edge (10)
(11) edge (12)
(14) edge (13)
(14) edge (15)
(1) edge (9)
(1) edge (12)
(1) edge (15)
(6) edge (7)
(6) edge (10)
(6) edge (13)
(7) edge (12)
(10) edge (15)
(7) edge (15);
\end{tikzpicture}
\end{tabular}
\caption{$H_5'$ (left), $H_5''$ (centre), and another subgraph of $H_5'$ in $\S_2^3(K_5)$ with the fewest possible edges (right)}\label{figH_k'Subgraph}
\end{center}
\end{figure}

Finally, we point out that we can further stretch the vertices of $H_k''$ to obtain very sparse graphs with arbitrarily high $\LS_+$-ranks. Given a graph $G$ and a vertex $v \in V(G)$, define
\[
w(v) \ce \begin{cases}
1 & \tn{if $\deg(v) \leq 3$;}\\
\deg(v) - 1 & \tn{if $\deg(v) \geq 4$.}
\end{cases}
\]
We also define $w(G) \ce \sum_{v \in V(G)} w(v)$. Then we have the following.

\begin{lemma}\label{lemStretchToDeg3}
For every graph $G$, there exists a graph $H$ that can be obtained from $G$ by a sequence of vertex stretching operations where $\deg(v) \leq 3$ for all $v \in V(H)$, and $|V(H)| \leq w(G)$.
\end{lemma}

\begin{proof}
First, if every vertex in $G$ has degree at most $3$, then $H=G$ suffices, so we now assume there exists $v \in V(G)$ with $\deg(v) \geq 4$. Next, we define
\begin{align*}
w_1(G) & \ce | \set{v \in V(G) : \deg(v) \geq 4}|,\\
w_2(G) & \ce \sum_{ i \in V(G)} \max\set{ \deg(i)-3, 0}.
\end{align*}
Then $w(G) = |V(G)| + w_1(G) + w_2(G)$ for every graph $G$. It is helpful to think of $w_2(G)$ as the total ``excess'' vertex degree in $G$, and $w_2(G) = 0$ if and only if $\deg(v) \leq 3$ for all $v \in V(G)$. Now notice that
\begin{itemize}
\item
If $v \in V(G)$ has $\deg(v) = 4$, we can $2$-stretch it with $|A_1| = |A_2|=2$. In this case, we obtain $H \in \S(G)$ with $|V(H)| = |V(G)| + 2$, $w_1(H) = w_1(G) - 1$, and $w_2(H) = w_2(G) -1$.
\item
If $v \in V(G)$ has $\deg(v) = 5$, we can $3$-stretch it with $|A_1| = |A_2|=2$ and $|A_3|=1$. In this case, we obtain $H \in \S(G)$ with $|V(H)| = |V(G)| + 3$, $w_1(H) = w_1(G) - 1$, and $w_2(H) = w_2(G) - 2$.
\item
If $v \in V(G)$ with $\deg(v) = p \geq 6$, we can $3$-stretch it with $|A_1| = |A_2|=2$ and $|A_3|=p-4$. In this case, we obtain $H \in \S(G)$ with $|V(H)| = |V(G)| + 3$, $w_1(H) \leq w_1(G)$, and $w_2(H) = w_2(G) - 3$. (More precisely, notice that $w_1(H) = w_1(G) - 1$ if $p =6$ and $w_1(H) = w_1(G)$ if $p \geq 7$.)
\end{itemize}

In all cases, we see that given a graph $G$ with $w_2(G) > 0$, we can apply a stretching operation to obtain $H \in \S(G)$ such that $w(H) \leq w(G)$ and $w_2(H) < w_2(G)$. Iterating this process would result in a graph $H$ with $w_1(H) = w_2(H) = 0$, which would satisfy $|V(H)| = w(H) \leq w(G)$.
\end{proof}

Then we have the following.

\begin{theorem}\label{thmStretchH_k''}
For every $k \geq 5$, there exists a graph $G$ on $k^2 - 4$ vertices such that $\deg(i) \leq 3$ for every $i \in V(G)$, and $r_+(G) \geq r_+(H_k'')$.
\end{theorem}

\begin{proof}
Given $k \geq 5$, consider the graph $H_k''$. Notice that $\deg(1_1) = \deg(2_2) = k-1$. Moreover, for every $i \in \set{3,\ldots,k}$, we have $\deg(i_0) = 2$, $\deg(i_1), \deg(i_2) \geq 3$, and $\deg(i_1) + \deg(i_2) = k+1$. 

Thus, using notation from the proof of Lemma~\ref{lemStretchToDeg3}, we obtain that $|V(H_k'')| = 3k-4$, $w_1(H_k'') \leq 2k-2$, and $w_2(H_k'') = k^2-5k+2$ (as each of $1_1, 2_2$ contributes $k-4$ to the sum, while $i_1$ and $i_2$ together contribute $k-5$ for every $i \in \set{3,\ldots,k}$). Therefore, $w(H_k'') \leq k^2-4$. Thus, we can apply Lemma~\ref{lemStretchToDeg3} to obtain a graph $G$ from stretching vertices of $H_k''$ where $|V(G)| \leq k^2-4$ and $\deg(v) \leq 3$ for all $v \in V(G)$. Since stretching a vertex cannot decrease the $\LS_+$-rank of a graph (Proposition~\ref{propVertexStretch}), the claim follows.
\end{proof}

Note that the bound $w_1(H_k'') \leq 2k-2$ is not tight for $k=5$ and $k=6$. In those cases, we can obtain a yet better bound as $w(H_5'') = 15$ and $w(H_6'') = 28$. Figure~\ref{figH_k''Stretch} illustrates $H_6''$ (left), and a stretched graph with $w(H_6'') = 28$ vertices which has maximum degree $3$ (right). Note that we suppressed the vertex labels in this figure to reduce cluttering.

\def\x{270 - 180/6}
\def\z{360/6}
\def\y{0.88}
\def\w{0.94}
\def\sc{2.5}

\begin{figure}[ht!]
\begin{center}
\begin{tabular}{cc}

\begin{tikzpicture}[scale=\sc, thick,main node/.style={circle, minimum size=1mm, fill=black, inner sep=0.1mm,draw,font=\tiny\sffamily}]
\node[main node] at ({cos(\x+(0)*\z)},{sin(\x+(0)*\z)}) (1) {};
\node[main node] at ({cos(\x+(1)*\z)},{sin(\x+(1)*\z)}) (6) {};

\node[main node] at ({ \y* cos(\x+(2)*\z) + (1-\y)*cos(\x+(1)*\z)},{ \y* sin(\x+(2)*\z) + (1-\y)*sin(\x+(1)*\z)}) (7) {};
\node[main node] at ({cos(\x+(2)*\z)},{sin(\x+(2)*\z)}) (8) {};
\node[main node] at ({ \y* cos(\x+(2)*\z) + (1-\y)*cos(\x+(3)*\z)},{ \y* sin(\x+(2)*\z) + (1-\y)*sin(\x+(3)*\z)}) (9) {};

\node[main node] at ({ \y* cos(\x+(3)*\z) + (1-\y)*cos(\x+(2)*\z)},{ \y* sin(\x+(3)*\z) + (1-\y)*sin(\x+(2)*\z)}) (10) {};
\node[main node] at ({cos(\x+(3)*\z)},{sin(\x+(3)*\z)}) (11) {};
\node[main node] at ({ \y* cos(\x+(3)*\z) + (1-\y)*cos(\x+(4)*\z)},{ \y* sin(\x+(3)*\z) + (1-\y)*sin(\x+(4)*\z)}) (12) {};

\node[main node] at ({ \y* cos(\x+(4)*\z) + (1-\y)*cos(\x+(3)*\z)},{ \y* sin(\x+(4)*\z) + (1-\y)*sin(\x+(3)*\z)}) (13) {};
\node[main node] at ({cos(\x+(4)*\z)},{sin(\x+(4)*\z)}) (14) {};
\node[main node] at ({ \y* cos(\x+(4)*\z) + (1-\y)*cos(\x+(5)*\z)},{ \y* sin(\x+(4)*\z) + (1-\y)*sin(\x+(5)*\z)}) (15) {};

\node[main node] at ({ \y* cos(\x+(5)*\z) + (1-\y)*cos(\x+(4)*\z)},{ \y* sin(\x+(5)*\z) + (1-\y)*sin(\x+(4)*\z)}) (16) {};
\node[main node] at ({cos(\x+(5)*\z)},{sin(\x+(5)*\z)}) (17) {};
\node[main node] at ({ \y* cos(\x+(5)*\z) + (1-\y)*cos(\x+(6)*\z)},{ \y* sin(\x+(5)*\z) + (1-\y)*sin(\x+(6)*\z)}) (18) {};

 \path[every node/.style={font=\sffamily}]
(1) edge (6)
(8) edge (7)
(8) edge (9)
(11) edge (10)
(11) edge (12)
(14) edge (13)
(14) edge (15)
(17) edge (16)
(17) edge (18)
(1) edge (9)
(1) edge (12)
(1) edge (15)
(1) edge (18)
(6) edge (7)
(6) edge (10)
(6) edge (13)
(6) edge (16)
(9) edge (10)
(9) edge (13)
(12) edge (13)
(12) edge (16)
(15) edge (16)
(18) edge (7);
\end{tikzpicture}

&

\begin{tikzpicture}[scale=\sc, thick,main node/.style={circle, minimum size=1mm, fill=black, inner sep=0.1mm,draw,font=\tiny\sffamily}]

\node[main node] at ({ \y* cos(\x+(0)*\z) + (1-\y)*cos(\x+(-1)*\z)},{ \y* sin(\x+(0)*\z) + (1-\y)*sin(\x+(-1)*\z)}) (1y) {};
\node[main node] at ({cos(\x+(0)*\z)},{sin(\x+(0)*\z)}) (1) {};
\node[main node] at ({ \y* cos(\x+(0)*\z) + (1-\y)*cos(\x+(1)*\z)},{ \y* sin(\x+(0)*\z) + (1-\y)*sin(\x+(1)*\z)}) (1b) {};

\node[main node] at ({\w* (cos(\x+(0)*\z)) + (1-\w)*(\y* cos(\x+(3)*\z) + (1-\y)*cos(\x+(4)*\z))},{\w*(sin(\x+(0)*\z)) + (1-\w)*( \y* sin(\x+(3)*\z) + (1-\y)*sin(\x+(4)*\z))}) (1c) {};

\node[main node] at ({ \y* cos(\x+(1)*\z) + (1-\y)*cos(\x+(0)*\z)},{ \y* sin(\x+(1)*\z) + (1-\y)*sin(\x+(0)*\z)}) (6y) {};
\node[main node] at ({cos(\x+(1)*\z)},{sin(\x+(1)*\z)}) (6) {};
\node[main node] at ({ \y* cos(\x+(1)*\z) + (1-\y)*cos(\x+(2)*\z)},{ \y* sin(\x+(1)*\z) + (1-\y)*sin(\x+(2)*\z)}) (6b) {};

\node[main node] at ({(-\w)* (cos(\x+(0)*\z)) - (1-\w)*(\y* cos(\x+(3)*\z) + (1-\y)*cos(\x+(4)*\z))},{\w*(sin(\x+(0)*\z)) + (1-\w)*( \y* sin(\x+(3)*\z) + (1-\y)*sin(\x+(4)*\z))}) (6c) {};

\node[main node] at ({ \y* cos(\x+(2)*\z) + (1-\y)*cos(\x+(1)*\z)},{ \y* sin(\x+(2)*\z) + (1-\y)*sin(\x+(1)*\z)}) (7) {};
\node[main node] at ({cos(\x+(2)*\z)},{sin(\x+(2)*\z)}) (8) {};
\node[main node] at ({ \y* cos(\x+(2)*\z) + (1-\y)*cos(\x+(3)*\z)},{ \y* sin(\x+(2)*\z) + (1-\y)*sin(\x+(3)*\z)}) (9) {};
\node[main node] at ({ (2*\y-1)* cos(\x+(2)*\z) + 2*(1-\y)*cos(\x+(3)*\z)},{ (2*\y-1)* sin(\x+(2)*\z) + 2*(1-\y)*sin(\x+(3)*\z)}) (9a) {};
\node[main node] at ({ (3*\y-2)* cos(\x+(2)*\z) + 3*(1-\y)*cos(\x+(3)*\z)},{ (3*\y-2)* sin(\x+(2)*\z) + 3*(1-\y)*sin(\x+(3)*\z)}) (9b) {};

\node[main node] at ({ \y* cos(\x+(3)*\z) + (1-\y)*cos(\x+(2)*\z)},{ \y* sin(\x+(3)*\z) + (1-\y)*sin(\x+(2)*\z)}) (10) {};
\node[main node] at ({cos(\x+(3)*\z)},{sin(\x+(3)*\z)}) (11) {};
\node[main node] at ({ \y* cos(\x+(3)*\z) + (1-\y)*cos(\x+(4)*\z)},{ \y* sin(\x+(3)*\z) + (1-\y)*sin(\x+(4)*\z)}) (12) {};
\node[main node] at ({ (2*\y-1)* cos(\x+(3)*\z) + 2*(1-\y)*cos(\x+(4)*\z)},{ (2*\y-1)* sin(\x+(3)*\z) + 2*(1-\y)*sin(\x+(4)*\z)}) (12a) {};
\node[main node] at ({ (3*\y-2)* cos(\x+(3)*\z) + 3*(1-\y)*cos(\x+(4)*\z)},{ (3*\y-2)* sin(\x+(3)*\z) + 3*(1-\y)*sin(\x+(4)*\z)}) (12b) {};

\node[main node] at ({ (3*\y-2)* cos(\x+(4)*\z) + 3*(1-\y)*cos(\x+(3)*\z)},{ (3*\y-2)* sin(\x+(4)*\z) + 3*(1-\y)*sin(\x+(3)*\z)}) (13y) {};
\node[main node] at ({ (2*\y-1)* cos(\x+(4)*\z) + 2*(1-\y)*cos(\x+(3)*\z)},{ (2*\y-1)* sin(\x+(4)*\z) + 2*(1-\y)*sin(\x+(3)*\z)}) (13z) {};
\node[main node] at ({ \y* cos(\x+(4)*\z) + (1-\y)*cos(\x+(3)*\z)},{ \y* sin(\x+(4)*\z) + (1-\y)*sin(\x+(3)*\z)}) (13) {};
\node[main node] at ({cos(\x+(4)*\z)},{sin(\x+(4)*\z)}) (14) {};
\node[main node] at ({ \y* cos(\x+(4)*\z) + (1-\y)*cos(\x+(5)*\z)},{ \y* sin(\x+(4)*\z) + (1-\y)*sin(\x+(5)*\z)}) (15) {};

\node[main node] at ({ (3*\y-2)* cos(\x+(5)*\z) + 3*(1-\y)*cos(\x+(4)*\z)},{ (3*\y-2)* sin(\x+(5)*\z) + 3*(1-\y)*sin(\x+(4)*\z)}) (16y) {};
\node[main node] at ({ (2*\y-1)* cos(\x+(5)*\z) + 2*(1-\y)*cos(\x+(4)*\z)},{ (2*\y-1)* sin(\x+(5)*\z) + 2*(1-\y)*sin(\x+(4)*\z)}) (16z) {};
\node[main node] at ({ \y* cos(\x+(5)*\z) + (1-\y)*cos(\x+(4)*\z)},{ \y* sin(\x+(5)*\z) + (1-\y)*sin(\x+(4)*\z)}) (16) {};
\node[main node] at ({cos(\x+(5)*\z)},{sin(\x+(5)*\z)}) (17) {};
\node[main node] at ({ \y* cos(\x+(5)*\z) + (1-\y)*cos(\x+(6)*\z)},{ \y* sin(\x+(5)*\z) + (1-\y)*sin(\x+(6)*\z)}) (18) {};

 \path[every node/.style={font=\sffamily}]
(1y) edge (1)
(1) edge (1b)
(6y) edge (6)
(6) edge (6b)
(9) edge (9a)
(9a) edge (9b)
(12) edge (12a)
(12a) edge (12b)
(13y) edge (13z)
(13z) edge (13)
(16y) edge (16z)
(16z) edge (16)
(1b) edge (6)
(8) edge (7)
(8) edge (9)
(11) edge (10)
(11) edge (12)
(14) edge (13)
(14) edge (15)
(17) edge (16)
(17) edge (18)
(1b) edge (9)
(1) edge (12)
(1y) edge (15)
(1y) edge (18)
(6b) edge (7)
(6b) edge (10)
(6) edge (13)
(6y) edge (16)
(9b) edge (10)
(9b) edge (13y)
(12b) edge (13)
(12b) edge (16y)
(15) edge (16y)
(18) edge (7);
\end{tikzpicture}
\end{tabular}
\caption{$H_6''$ (left), and a $28$-vertex graph with maximum degree $3$ obtained from stretching $H_6''$ (right)}\label{figH_k''Stretch}
\end{center}
\end{figure}
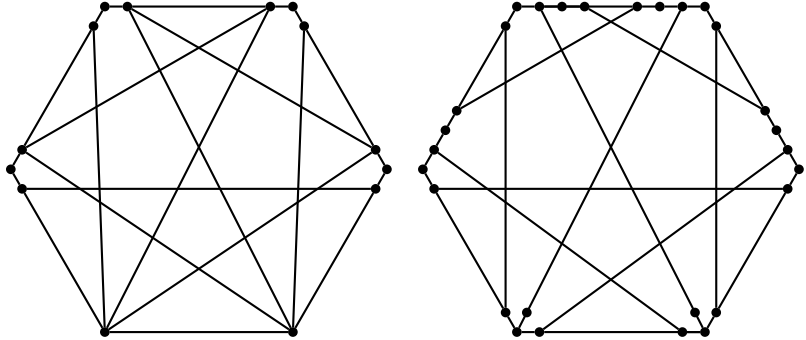


Also, since $r_+(H_k'') = \Theta(k)$, it follows from Theorem~\ref{thmStretchH_k''} that there exists a family of graphs $G$ with maximum degree $3$ where $r_+(G) = \Omega(\sqrt{|V(G)|})$. This bound asymptotically matches the previously known bound achieved by line graphs of odd cliques, whose vertex degrees grow without bound.

Finally, it was pointed to us by a reviewer that the family of graphs $H_k'$ coincide with
the family of graphs $G_k$ in~\cite[page 675]{DobreV2015} (also studied recently in~\cite[Section 5.5]{Vargas23}). It is very interesting that the families of graphs $H_k'$ have been considered as challenging instances for other but related convex relaxations of the stable set polytope. These graphs are also related to four graphs $G_8$, $G_{11}$, $G_{14}$ and $G_{17}$ considered as minimal obstructions in~\cite{PenaVZ2007} to the hierarchies discussed there which are related to the hierarchy
proposed in~\cite{deKlerkP2002}. These four graphs, as well as the graph $H_8$ recently studied in~\cite{LaurentV23} and ~\cite{Vargas23}, are related to our family $H_k''$ as they all are among graphs in $\S_2^{k-2}(K_k)$ with the fewest possible edges. These connections raise some more hope that some of our techniques and approaches in this paper may be useful for analyzing other convex relaxations of the stable set polytope.

\section{Some Future Research Directions}\label{sec8}

In this section, we mention some follow-up questions to our work in this manuscript that could lead to interesting future research.

\begin{problem}\label{pro1}
Is there an $\ell$-minimal graph $G$ in $\S_2^{\ell-1}(K_{\ell+2})$ for all $\ell \in \mN$?
\end{problem}

Results from~\cite{LiptakT03, EscalanteMN06} show that the answer is ``yes'' for $\ell \in \{1,2,3\}$. Our $4$-minimal graph $G_{4,1}$ shows that this is also true for $\ell = 4$. Does the pattern continue for larger $\ell$? And more importantly, how can we verify the $\LS_+$-rank of these graphs analytically, as opposed to primarily relying on specific numerical certificates?

\begin{problem}
Given $\ell \in \mN$, what are the maximum and minimum possible edge densities of $\ell$-minimal graphs?
\end{problem}

Given $\ell \in \mN$, let $d^+(\ell)$ (resp. $d^-(\ell)$) be the maximum (resp. minimum) possible edge density of an $\ell$-minimal graph. It was previously known that $d^+(1) = d^-(1) = 1$ (attained by the $3$-cycle), $d^+(2) = \frac{3}{5}$ ($G_{2,2}$), $d^-(2) = \frac{8}{15}$ ($G_{2,1}$), and $d^-(3) \leq \frac{7}{18}$ ($G_{3,1}$). In this work we showed that $d^-(3) = \frac{7}{18}$ (Proposition~\ref{propSparsest3min}) and $d^+(3) \geq \frac{4}{9}$ ($G_{3,2}$). For $\ell=4$, the discovery of $G_{4,1}$ and the other $4$-minimal graphs presented in Figure~\ref{figG41subs} show that $d^-(4) \leq \frac{7}{22}$ and $d^+(4) \geq \frac{4}{11}$. Can we prove tight bounds for $d^+(\ell)$ and/or $d^-(\ell)$ in general?

\begin{problem}
How many non-isomorphic $\ell$-minimal graphs are there for each $\ell \geq 1$?
\end{problem}

Given $\ell \in \mN$, let $c(\ell)$ denote the number of non-isomorphic $\ell$-minimal graphs. We know that $c(1)= 1$ (the triangle) and $c(2)=2$ ($G_{2,1}$ and $G_{2,2}$). We showed in Section~\ref{sec6} that $c(3) \geq 6$ ($G_{3,2}$ and its subgraphs in Figure~\ref{figG32subs}) and $c(4) \geq 7$ ($G_{4,1}$ and its subgraphs in Figure~\ref{figG41subs}). Does $c(\ell)$ grow without bound as $\ell$ increases? If so, at what rate asymptotically?

\begin{problem}\label{prosparse1}
What is the fastest growing function $f$ such that there exist graphs $G$ with maximum degree at most three and $r_+(G) = \Theta(f(|V(G)|))$?
\end{problem}

\begin{problem}\label{prosparse2}
What is the fastest growing function $f$ such that there exist cubic graphs $G$ with $r_+(G) = \Theta(f(|V(G)|))$?
\end{problem}

We proved in Section~\ref{sec7} that there exist very sparse graphs (maximum degree at most three) with $r_+(G) = \Omega(\sqrt{|V(G)|})$. Since all graphs $G$ with maximum degree at most two satisfy $r_+(G) \leq 1$, Problems~\ref{prosparse1} and~\ref{prosparse2} are really about the sparsest graphs with high $\LS_+$-ranks.

\begin{problem}
What can we say about the lift-and-project ranks of graphs for other positive semidefinite lift-and-project operators? To start with some concrete questions for this research problem, what are the solutions of Problems~\ref{pro1}-\ref{prosparse2} when we replace $\LS_+$ with $\Las, \BZ_+, \Theta_k$, or $\SA_+$? (For Problem~\ref{pro1}, we may have different sets $\S$, based on different graph operations, for different lift-and-project operators.)
\end{problem}

After $\LS_+$, many stronger semidefinite lift-and-project operators (such as $\Las$~\cite{Lasserre01},\\
$\BZ_+$~\cite{BienstockZ04}, $\Theta_k$~\cite{GouveiaPT10}, and $\SA_+$~\cite{AuT16}) have been proposed. While these stronger operators are capable of producing tighter relaxations than $\LS_+$, these SDP relaxations can also be more computationally challenging to solve. For instance, while the $\LS_+^k$-relaxation of a set $P \subseteq [0,1]^n$ involves $O(n^k)$ PSD constraints of order $O(n)$, the operators $\Las^k, \BZ_+^k$ and $\SA_+^k$ all impose one (or more) PSD constraint of order $\Omega(n^k)$ in their formulations. We briefly mentioned at the end of Section~\ref{sec3} that some of our tools for analyzing $\LS_+$ relaxations can be extended to these other operators, as well as pointed out some related work on other convex relations of the stable set polytope at the end of Section~\ref{sec7}. It would be interesting to determine the corresponding properties of graphs which are minimal with respect to the stronger lift-and-project operators.

\section*{Declarations}

{\bf Conflict of interest:} The authors declare that they have no conflict of interest.

\bibliographystyle{alpha}
\bibliography{ref}

\appendix

\section{Proofs of Lemmas~\ref{lemG412},~\ref{lemG411}, and~\ref{lemG413}}\label{secA1}

The following lemmas provide the deferred technical details from the proof of Theorem~\ref{thmG41}. To reduce cluttering, given $S \subseteq [n]$ we will let $\het{\chi}_S$ denote the vector $\begin{bmatrix} 1 \\ \chi_S \end{bmatrix} \in \mR^{n+1}$.

\begin{lemma}\label{lemG412}
Let $Y_0$ be as defined in the proof of Theorem~\ref{thmG41}. Then
\[
Y_0(e_0 - e_1) \in \cone(\LS_+^2(G_{4,1})).
\]
\end{lemma}

\begin{proof}
First, notice that $[Y_0(e_0 - e_1)]_{1} = 0$. Thus, let $G' \ce G_{4,1} - 1$ and $v$ be the restriction of $Y_0(e_0 - e_1)$ to the coordinates indexed by $\cone(\LS_+^2(G'))$. Then, by Lemma~\ref{lemfacet}, it suffices to show that $v \in \cone(\LS_+^2(G'))$. Consider the matrix
{\scriptsize
\[
Y_2 \ce 
\bbordermatrix{
&& 2 & 3 & 4_1 & 4_0 & 4_2 & 5_1 & 5_0 & 5_2 & 6_1 & 6_0 & 6_2 \cr
 & 74660 & 25340 & 25340 & 16500 & 57518 & 8662 & 8662 & 57518 & 16500 & 16500 & 49680 & 16500 \cr
 & 25340 & 25340 & 0 & 0 & 25340 & 0 & 0 & 17166 & 8174 & 8363 & 16977 & 0 \cr
 & 25340 & 0 & 25340 & 8174 & 17166 & 0 & 0 & 25340 & 0 & 0 & 16977 & 8363 \cr
 & 16500 & 0 & 8174 & 16500 & 0 & 8320 & 342 & 16158 & 0 & 0 & 14067 & 2433 \cr
 & 57518 & 25340 & 17166 & 0 & 57518 & 0 & 7678 & 41360 & 16158 & 16494 & 34971 & 14067 \cr
 & 8662 & 0 & 0 & 8320 & 0 & 8662 & 984 & 7678 & 342 & 0 & 8662 & 0 \cr
 & 8662 & 0 & 0 & 342 & 7678 & 984 & 8662 & 0 & 8320 & 0 & 8662 & 0 \cr
 & 57518 & 17166 & 25340 & 16158 & 41360 & 7678 & 0 & 57518 & 0 & 14067 & 34971 & 16494 \cr
 & 16500 & 8174 & 0 & 0 & 16158 & 342 & 8320 & 0 & 16500 & 2433 & 14067 & 0 \cr
 & 16500 & 8363 & 0 & 0 & 16494 & 0 & 0 & 14067 & 2433 & 16500 & 0 & 8137 \cr
 & 49680 & 16977 & 16977 & 14067 & 34971 & 8662 & 8662 & 34971 & 14067 & 0 & 49680 & 0 \cr
 & 16500 & 0 & 8363 & 2433 & 14067 & 0 & 0 & 16494 & 0 & 8137 & 0 & 16500 \cr
 }.
\]}
We claim that $Y_2 \in \widehat{\LS}_+^2(G')$. First, one can verify that $Y_2 \succeq 0$ (a $UV$-certificate is provided in Table~\ref{tabUV}). Also, notice that the function $f_2$ (restricted to $V(G')$) is an automorphism of $G'$. Moreover, observe that for all $i,j \in V(G')$, $Y_2[i,j] = Y_2[f_2(i), f_2(j)]$. Thus, by symmetry, it only remains to prove the conditions $Y_2e_i, Y_2(e_0-e_i) \in \cone(\LS_+(G'))$ for $i \in \set{2,4_1,4_0,4_2,6_1,6_0}$.

First, notice that
\begin{itemize}
\item
$[Y_2e_{4_0}]_{0} = [Y_2e_{4_0}]_{4_0}$, $[Y_2e_{4_0}]_{4_1} = [Y_2e_{4_0}]_{4_2} = 0$, and that the following matrix certifies that $Y_2e_{4_0}$ (with the entries corresponding to vertices $4_1, 4_0, 4_2$ removed) belongs to $\cone(\LS_+(G' \ominus 4_0))$. 

{\scriptsize
\[
Y_{21} \ce 
\bbordermatrix{
&& 2 & 3 & 5_1 & 5_0 & 5_2 & 6_1 & 6_0 & 6_2 \cr
 & 57518 & 25340 & 17164 & 7678 & 41360 & 16158 & 16496 & 34970 & 14068 \cr
 & 25340 & 25340 & 0 & 0 & 19057 & 6283 & 5860 & 19444 & 0 \cr
 & 17164 & 0 & 17164 & 0 & 17164 & 0 & 0 & 12010 & 5117 \cr
 & 7678 & 0 & 0 & 7678 & 0 & 7678 & 3125 & 4516 & 0 \cr
 & 41360 & 19057 & 17164 & 0 & 41360 & 0 & 10718 & 26585 & 10400 \cr
 & 16158 & 6283 & 0 & 7678 & 0 & 16158 & 5778 & 8385 & 3668 \cr
 & 16496 & 5860 & 0 & 3125 & 10718 & 5778 & 16496 & 0 & 8910 \cr
 & 34970 & 19444 & 12010 & 4516 & 26585 & 8385 & 0 & 34970 & 0 \cr
 & 14068 & 0 & 5117 & 0 & 10400 & 3668 & 8910 & 0 & 14068 \cr
}
\]}

\item
$[Y_2e_{6_0}]_{0} = [Y_2e_{6_0}]_{6_0}$, $[Y_2e_{6_0}]_{6_1} = [Y_2e_{6_0}]_{6_2} = 0$, and that the following matrix certifies that $Y_2e_{6_0}$ (with the entries corresponding to vertices $6_1, 6_0, 6_2$ removed) belongs to $\cone(\LS_+(G' \ominus 6_0))$. 

{\scriptsize
\[
Y_{22} \ce 
\bbordermatrix{
&& 2 & 3 & 4_1 & 4_0 & 4_2 & 5_1 & 5_0 & 5_2 & \cr
 & 49680 & 16977 & 16977 & 14068 & 34970 & 8662 & 8662 & 34970 & 14068 \cr
 & 16977 & 16977 & 0 & 0 & 16977 & 0 & 0 & 11129 & 5848 \cr
 & 16977 & 0 & 16977 & 5848 & 11129 & 0 & 0 & 16977 & 0 \cr
 & 14068 & 0 & 5848 & 14068 & 0 & 8220 & 442 & 13626 & 0 \cr
 & 34970 & 16977 & 11129 & 0 & 34970 & 0 & 7578 & 21344 & 13626 \cr
 & 8662 & 0 & 0 & 8220 & 0 & 8662 & 1084 & 7578 & 442 \cr
 & 8662 & 0 & 0 & 442 & 7578 & 1084 & 8662 & 0 & 8220 \cr
 & 34970 & 11129 & 16977 & 13626 & 21344 & 7578 & 0 & 34970 & 0 \cr
 & 14068 & 5848 & 0 & 0 & 13626 & 442 & 8220 & 0 & 14068 \cr
}
\]}

\item
$[Y_2(e_0-e_2)]_{2} =0$, and that the following matrix certifies that $Y_2(e_0-e_2)$ (with the entry corresponding to vertex $2$ removed) belongs to $\cone(\LS_+(G' - 2))$.

{\scriptsize
\[
Y_{23} \ce 
\bbordermatrix{
&& 3 & 4_1 & 4_0 & 4_2 & 5_1 & 5_0 & 5_2 & 6_1 & 6_0 & 6_2 \cr
 & 49320 & 25340 & 16500 & 32178 & 8662 & 8662 & 40354 & 8324 & 8137 & 32703 & 16500 \cr
 & 25340 & 25340 & 6118 & 19222 & 0 & 0 & 25340 & 0 & 0 & 19368 & 5972 \cr
 & 16500 & 6118 & 16500 & 0 & 8107 & 595 & 15905 & 0 & 2465 & 10494 & 6006 \cr
 & 32178 & 19222 & 0 & 32178 & 0 & 7688 & 24409 & 7769 & 5672 & 21928 & 10250 \cr
 & 8662 & 0 & 8107 & 0 & 8662 & 974 & 7688 & 555 & 0 & 5933 & 2729 \cr
 & 8662 & 0 & 595 & 7688 & 974 & 8662 & 0 & 8067 & 70 & 8592 & 0 \cr
 & 40354 & 25340 & 15905 & 24409 & 7688 & 0 & 40354 & 0 & 7763 & 24111 & 16243 \cr
 & 8324 & 0 & 0 & 7769 & 555 & 8067 & 0 & 8324 & 374 & 7950 & 257 \cr
 & 8137 & 0 & 2465 & 5672 & 0 & 70 & 7763 & 374 & 8137 & 0 & 8067 \cr
 & 32703 & 19368 & 10494 & 21928 & 5933 & 8592 & 24111 & 7950 & 0 & 32703 & 0 \cr
 & 16500 & 5972 & 6006 & 10250 & 2729 & 0 & 16243 & 257 & 8067 & 0 & 16500 \cr
}
\]}

\item
$[Y_2(e_0-e_{4_1})]_{4_1} =0$, and that the following matrix certifies that $Y_2(e_0-e_{4_1})$ (with the entry corresponding to vertex $4_1$ removed) belongs to $\cone(\LS_+(G' - 4_1))$.

{\scriptsize
\[
Y_{24} \ce 
\bbordermatrix{
&& 2 & 3 & 4_0 & 4_2 & 5_1 & 5_0 & 5_2 & 6_1 & 6_0 & 6_2 \cr
 & 58160 & 25340 & 17164 & 57518 & 342 & 8320 & 41360 & 16500 & 16496 & 35612 & 14068 \cr
 & 25340 & 25340 & 0 & 25228 & 0 & 0 & 19229 & 6068 & 5788 & 19552 & 0 \cr
 & 17164 & 0 & 17164 & 17063 & 0 & 0 & 17055 & 0 & 0 & 12187 & 4977 \cr
 & 57518 & 25228 & 17063 & 57518 & 0 & 7946 & 41199 & 16198 & 16378 & 35219 & 13979 \cr
 & 342 & 0 & 0 & 0 & 342 & 340 & 1 & 190 & 0 & 340 & 1 \cr
 & 8320 & 0 & 0 & 7946 & 340 & 8320 & 0 & 8190 & 3046 & 5274 & 0 \cr
 & 41360 & 19229 & 17055 & 41199 & 1 & 0 & 41360 & 0 & 10763 & 26612 & 10417 \cr
 & 16500 & 6068 & 0 & 16198 & 190 & 8190 & 0 & 16500 & 5653 & 8919 & 3601 \cr
 & 16496 & 5788 & 0 & 16378 & 0 & 3046 & 10763 & 5653 & 16496 & 0 & 9091 \cr
 & 35612 & 19552 & 12187 & 35219 & 340 & 5274 & 26612 & 8919 & 0 & 35612 & 0 \cr
 & 14068 & 0 & 4977 & 13979 & 1 & 0 & 10417 & 3601 & 9091 & 0 & 14068 \cr
}
\]}
\end{itemize}
Also, notice that $Y_{21}e_0 = Y_{21}(e_{5_0} + e_{5_2})$. Thus, if we let $Y'_{21}$ be the matrix obtained from $Y_{21}$ by removing the $0^{\tn{th}}$ row and column, then we see that $Y'_{21} \succeq 0 \Rightarrow Y_{21} \succeq 0$. The $UV$-certificates of $Y'_{21}, Y_{22}, Y_{23}$, and $Y_{24}$ are provided in Table~\ref{tabUV}.

Next, observe that

\begin{allowdisplaybreaks}
\begin{align*}
Y_2e_2 \leq{} &
 8291 \het{\chi}_{\set{ 2 , 4_0 , 5_0 , 6_1 }} + 
 8873 \het{\chi}_{\set{ 2 , 4_0 , 5_0 , 6_0 }} + 
 72 \het{\chi}_{\set{ 2 , 4_0 , 5_2 , 6_1 }} + 
 8104 \het{\chi}_{\set{ 2 , 4_0 , 5_2 , 6_0 }} 
,\\
Y_2e_{4_1} \leq{} &
 6365 \het{\chi}_{\set{ 3 , 4_1 , 5_0 , 6_0 }} + 
 1811 \het{\chi}_{\set{ 3 , 4_1 , 5_0 , 6_2 }} + 
 342 \het{\chi}_{\set{ 4_1 , 4_2 , 5_1 , 6_0 }} +
 7361 \het{\chi}_{\set{ 4_1 , 4_2 , 5_0 , 6_0 }} \\ 
 &+
 617 \het{\chi}_{\set{ 4_1 , 4_2 , 5_0 , 6_2 }} + 
 4 \het{\chi}_{\set{ 4_1 , 5_0 , 6_1 , 6_2 }} 
,\\
Y_2e_{4_2} \leq{} &
 642 \het{\chi}_{\set{ 4_1 , 4_2 , 5_1 , 6_0 }} + 
 7678 \het{\chi}_{\set{ 4_1 , 4_2 , 5_0 , 6_0 }} + 
 342 \het{\chi}_{\set{ 4_2 , 5_1 , 5_2 , 6_0 }} 
,\\
Y_2e_{6_1} \leq{} &
 6764 \het{\chi}_{\set{ 2 , 4_0 , 5_0 , 6_1 }} + 
 1599 \het{\chi}_{\set{ 2 , 4_0 , 5_2 , 6_1 }} + 
 4 \het{\chi}_{\set{ 4_1 , 5_0 , 6_1 , 6_2 }} +
 7300 \het{\chi}_{\set{ 4_0 , 5_0 , 6_1 , 6_2 }} \\ 
 &+
  833 \het{\chi}_{\set{ 4_0 , 5_2 , 6_1 , 6_2 }} 
,\\
Y_2(e_0-e_{4_0}) \leq{} &
 7254 \het{\chi}_{\set{ 3 , 4_1 , 5_0 , 6_0 }} + 
 1004 \het{\chi}_{\set{ 3 , 4_1 , 5_0 , 6_2 }} + 
 489 \het{\chi}_{\set{ 4_1 , 4_2 , 5_1 , 6_0 }} +
 6472 \het{\chi}_{\set{ 4_1 , 4_2 , 5_0 , 6_0 }} \\ 
 &+
 1291 \het{\chi}_{\set{ 4_1 , 4_2 , 5_0 , 6_2 }} + 
 137 \het{\chi}_{\set{ 4_1 , 5_0 , 6_1 , 6_2 }} + 
 495 \het{\chi}_{\set{ 4_2 , 5_1 , 5_2 , 6_0 }}
,\\
 Y_2(e_0-e_{4_2}) \leq{} &
 832 \het{\chi}_{\set{ 2 , 4_0 , 5_0 , 6_1 }} + 
 3414 \het{\chi}_{\set{ 2 , 4_0 , 5_0 , 6_0 }} + 
 496 \het{\chi}_{\set{ 2 , 4_0 , 5_2 , 6_1 }} + 
 919 \het{\chi}_{\set{ 2 , 4_0 , 5_2 , 6_0 }} \\ 
 &+
  5480 \het{\chi}_{\set{ 3 , 4_1 , 5_0 , 6_0 }} + 
 2094 \het{\chi}_{\set{ 3 , 4_1 , 5_0 , 6_2 }} + 
 2827 \het{\chi}_{\set{ 3 , 4_0 , 5_0 , 6_0 }} + 
 1634 \het{\chi}_{\set{ 3 , 4_0 , 5_0 , 6_2 }} \\ 
 &+
 700 \het{\chi}_{\set{ 4_1 , 5_0 , 6_1 , 6_2 }} + 
 526 \het{\chi}_{\set{ 4_0 , 5_1 , 5_2 , 6_1 }} + 
 1070 \het{\chi}_{\set{ 4_0 , 5_1 5_2 , 6_0 }} + 
 690 \het{\chi}_{\set{ 4_0 , 5_0 , 6_1 , 6_2 }} \\
 &+
 455 \het{\chi}_{\set{ 4_0 , 5_2 , 6_1 , 6_2 }} + 
 126 \het{\chi}_{\set{ 4_2 , 5_1 , 5_2 , 6_0 }} + 
\frac{44735}{57518}Y_2e_{4_0}
,\\
 Y_2(e_0-e_{6_1}) \leq{} &
 275 \het{\chi}_{\set{ 2 , 4_0 , 5_0 , 6_0 }} + 
 186 \het{\chi}_{\set{ 3 , 4_1 , 5_0 , 6_0 }} + 
 2333 \het{\chi}_{\set{ 3 , 4_1 , 5_0 , 6_2 }} + 
 186 \het{\chi}_{\set{ 3 , 4_0 , 5_0 , 6_0 }} \\
  &+
 5933 \het{\chi}_{\set{ 3 , 4_0 , 5_0 , 6_2 }} + 
 140 \het{\chi}_{\set{ 4_1 , 4_2 , 5_0 , 6_2 }} + 
 227 \het{\chi}_{\set{ 4_0 , 5_1 , 5_2 , 6_0 }} + 
\frac{48880}{49680}Y_2e_{6_0}
,\\
Y_2(e_0-e_{6_0}) \leq{} &
 7474 \het{\chi}_{\set{ 2 , 4_0 , 5_0 , 6_1 }} + 
 978 \het{\chi}_{\set{ 2 , 4_0 , 5_2 , 6_1 }} + 
 978 \het{\chi}_{\set{ 3 , 4_1 , 5_0 , 6_2 }} + 
 7474 \het{\chi}_{\set{ 3 , 4_0 , 5_0 , 6_2 }} \\
  &+
1454 \het{\chi}_{\set{ 4_1 , 5_0 , 6_1 , 6_2 }} + 
 5168 \het{\chi}_{\set{ 4_0 , 5_0 , 6_1 , 6_2 }} + 
 1454 \het{\chi}_{\set{ 4_0 , 5_2 , 6_1 , 6_2 }}.
\end{align*}
\end{allowdisplaybreaks}

Since all incidence vectors above correspond to stable sets in $G'$, and we already showed earlier that $Y_2e_{4_0}, Y_2e_{6_0} \in \cone(\LS_+(G'))$, we obtain that all the vectors above belong to $\cone(\LS_+(G'))$. Thus, we conclude that $Y_0(e_0-e_1) \in \cone(\LS_+^2(G_{4,1}))$.
\end{proof}

\begin{lemma}\label{lemG411}
Let $Y_0$ be as defined in the proof of Theorem~\ref{thmG41}. Then 
\[
Y_0e_{6_0} \in \cone(\LS_+^2(G_{4,1})).
\]
\end{lemma}

\begin{proof}
First, notice that $[Y_0e_{6_0}]_{0} = [Y_0e_{6_0}]_{6_0}$, and $[Y_0e_{6_0}]_{6_1} = [Y_0e_{6_0}]_{6_2} = 0$. Thus, let $G' \ce G_{4,1} \ominus 6_0$ and $v$ be the restriction of $Y_0e_{6_0}$ to the coordinates indexed by $\cone(\LS_+^2(G'))$. Then, by Lemma~\ref{lemfacet}, it suffices to show that $v \in \cone(\LS_+^2(G'))$. Consider the matrix
{\scriptsize
\[
Y_1 \ce 
\bbordermatrix{
&& 1 & 2 & 3 & 4_1 & 4_0 & 4_2 & 5_1 & 5_0 & 5_2 \cr
 & 75020 & 25340 & 17502 & 17502 & 15419 & 51150 & 15911 & 15911 & 51150 & 15419 \cr
 & 25340 & 25340 & 0 & 0 & 0 & 17400 & 7940 & 7940 & 17400 & 0 \cr
 & 17502 & 0 & 17502 & 0 & 0 & 17502 & 0 & 0 & 9571 & 7931 \cr
 & 17502 & 0 & 0 & 17502 & 7931 & 9571 & 0 & 0 & 17502 & 0 \cr
 & 15419 & 0 & 0 & 7931 & 15419 & 0 & 7488 & 396 & 14993 & 0 \cr
 & 51150 & 17400 & 17502 & 9571 & 0 & 51150 & 0 & 15485 & 27920 & 14993 \cr
 & 15911 & 7940 & 0 & 0 & 7488 & 0 & 15911 & 396 & 15485 & 396 \cr
 & 15911 & 7940 & 0 & 0 & 396 & 15485 & 396 & 15911 & 0 & 7488 \cr
 & 51150 & 17400 & 9571 & 17502 & 14993 & 27920 & 15485 & 0 & 51150 & 0 \cr
 & 15419 & 0 & 7931 & 0 & 0 & 14993 & 396 & 7488 & 0 & 15419 \cr
}.
\]}
We claim that $Y_1 \in \widehat{\LS}_+^2(G')$. First, one can verify that $Y_1 \succeq 0$ (a $UV$-certificate is provided in Table~\ref{tabUV}). Also, notice that the function $f_2$ (restricted to $V(G')$) is an automorphism of $G'$. Moreover, observe that for all $i,j \in V(G')$, $Y_1[i,j] = Y_1[f_2(i), f_2(j)]$. Thus, by symmetry, it only remains to prove the conditions $Y_1e_i, Y_1(e_0-e_i) \in \cone(\LS_+(G'))$ for $i \in \set{1,2,4_1,4_0,4_2}$.

First, notice that $[Y_1e_{4_0}]_{0} = [Y_1e_{4_0}]_{4_0}$, $[Y_1e_{4_0}]_{4_1} = [Y_1e_{4_0}]_{4_2} = 0$, and that the following matrix certifies that $Y_1e_{4_0}$ (with the entries corresponding to vertices $4_1, 4_0, 4_2$ removed) belongs to $\cone(\LS_+(G' \ominus 4_0))$. (See Table~\ref{tabUV} for a $UV$-certificate.)

{\scriptsize
\[
Y_{11} \ce \bbordermatrix{
&& 1 & 2 & 3 & 5_1 & 5_0 & 5_2 \cr
 & 51150 & 17400 & 17502 & 9571 & 15485 & 27920 & 14993 \cr
 & 17400 & 17400 & 0 & 0 & 7544 & 9856 & 0 \cr
 & 17502 & 0 & 17502 & 0 & 0 & 10450 & 7052 \cr
 & 9571 & 0 & 0 & 9571 & 0 & 9571 & 0 \cr
 & 15485 & 7544 & 0 & 0 & 15485 & 0 & 7941 \cr
 & 27920 & 9856 & 10450 & 9571 & 0 & 27920 & 0 \cr
 & 14993 & 0 & 7052 & 0 & 7941 & 0 & 14993 \cr
}
\]}
Now consider the following vectors:

{\footnotesize
\[
\begin{array}{rccccccccccl}
&& 1 & 2 & 3 & 4_1 & 4_0 & 4_2 & 5_1 & 5_0 & 5_2\\
z^{(1)} \ce [&51150 & 17400 & 17502 & 9571 & 0 & 51150 & 0 & 5485 & 27920 & 14933 & ]^{\top}\\ 
z^{(2)} \ce [&51150 & 17502 & 17400 & 9571 & 0 & 51150 & 0 & 14933 & 27920 & 15485 & ]^{\top}\\ 
z^{(3)} \ce [&57518 & 25340 & 0 & 17164 & 14068 & 34970 & 16496 & 16158 & 41360 & 7678 & ]^{\top}\\ 
z^{(4)} \ce [&49680 & 0 & 16977 & 16977 1 & 4068 3 & 4970 & 8662 & 8662 & 34970 & 14068 & ]^{\top} 
\end{array}
\]}

Notice that $z^{(1)} \in \cone(\LS_+(G'))$ follows from $Y_1e_{4_0} \in \cone(\LS_+(G'))$ as shown above. Then it follows from the symmetry of $G'$ that $z^{(2)} \in \cone(\LS_+(G'))$ as well. $z^{(3)}, z^{(4)} \in \cone(\LS_+(G'))$ follows respectively from $Y_2e_{4_0}, Y_2e_{6_0} \in \cone(\LS_+(G_{4,1} - 1))$, as shown in Lemma~\ref{lemG412}. Next, observe that

\begin{allowdisplaybreaks}
\begin{align*}
Y_1e_1 \leq{} &
17400 \het{\chi}_{\set{ 1 , 4_0 , 5_0 }} + 
7940 \het{\chi}_{\set{ 1 , 4_2 , 5_1 }} 
,\\
Y_1e_2 \leq{} &
9571	 \het{\chi}_{\set{ 2 , 4_0 , 5_0}} + 
7931 \het{\chi}_{\set{ 	2, 4_0 , 5_2}}
,\\
Y_1e_{4_1} \leq{} &
7931 \het{\chi}_{\set{ 3 ,4_1 , 5_0 }} + 
396 \het{\chi}_{\set{ 4_1 , 4_2 , 5_1}} + 
7092 \het{\chi}_{\set{ 4_1 , 4_2 , 5_0}}
,\\
Y_1e_{4_2} \leq{} &
 414 \het{\chi}_{\set{ 1 , 4_0 , 5_1 }} + 
 7523 \het{\chi}_{\set{ 2 , 4_0 , 5_0 }} + 
 7974 \het{\chi}_{\set{ 3 , 4_0 , 5_0 }} 
 ,\\
 Y_1(e_0-e_1) \leq{} &
 874 \het{\chi}_{\set{ 2 , 4_0 , 5_0 }} + 
 3160 \het{\chi}_{\set{ 2 , 4_0 , 5_2 }} + 
 3160 \het{\chi}_{\set{ 3 , 4_1 , 5_0 }} + 
 874 \het{\chi}_{\set{ 3 , 4_0 , 5_0 }} + 
 1100 \het{\chi}_{\set{ 4_1 , 4_2 , 5_0 }} \\
  &+
 1100 \het{\chi}_{\set{ 4_0 , 5_1 , 5_2 }} + 
 \frac{39412}{49680}z^{(4)}
 ,\\
Y_1(e_0-e_{2}) \leq{} &
 1729 \het{\chi}_{\set{ 1 , 4_0 , 5_1 }} + 
 6165 \het{\chi}_{\set{ 1 , 4_0 , 5_0 }} + 
 626 \het{\chi}_{\set{ 1 , 4_2 , 5_1 }} + 
 1749 \het{\chi}_{\set{ 1 , 4_2 , 5_0 }} \\
  &+
 4298 \het{\chi}_{\set{ 3 , 4_1 , 5_0 }} + 
 2999 \het{\chi}_{\set{ 3 , 4_0 , 5_0 }} + 
 1009 \het{\chi}_{\set{ 4_1 , 4_2 , 5_1 }} + 
 1751 \het{\chi}_{\set{ 4_1 , 4_2 , 5_0 }} \\
  &+
 1959 \het{\chi}_{\set{ 4_0 , 5_1 , 5_2 }} + 
 971 \het{\chi}_{\set{ 4_2 , 5_1 , 5_2 }} + 
 \frac{34235}{57518}z^{(3)}+ 
27\het{\chi}_{\emptyset}
,\\
Y_1(e_0-e_{4_1}) \leq{} &
 498 \het{\chi}_{\set{ 1 , 4_0 , 5_1 }} + 
 2500 \het{\chi}_{\set{ 1 , 4_0 , 5_0 }} + 
 639 \het{\chi}_{\set{ 1 , 4_2 , 5_1 }} + 
 6832 \het{\chi}_{\set{ 1 , 4_2 , 5_0 }} \\
  &+
 1613 \het{\chi}_{\set{ 2 , 4_0 , 5_0 }} + 
 1022 \het{\chi}_{\set{ 2 , 4_0 , 5_2 }} + 
 1421 \het{\chi}_{\set{ 3 , 4_0 , 5_0 }} + 
 478 \het{\chi}_{\set{ 4_0 , 5_1 , 5_2 }} \\ 
 &+
 952 \het{\chi}_{\set{ 4_2 , 5_1 , 5_2 }} + 
 \frac{20799}{51150}z^{(1)} +
 \frac{22819}{51150}z^{(2)} +
 28 \het{\chi}_{\emptyset}
 ,\\
Y_1(e_0-e_{4_0}) \leq{} &
 452 \het{\chi}_{\set{ 1 , 4_0 , 5_1 }} + 
 7504 \het{\chi}_{\set{ 2 , 4_0 , 5_0 }} + 
 7931 \het{\chi}_{\set{ 3 , 4_1 , 5_0 }}+ 
 7955 \het{\chi}_{\set{ 3 , 4_0 , 5_0 }} + 
28 \het{\chi}_{\emptyset}
,\\
Y_1(e_0-e_{4_2}) \leq{} &
 234 \het{\chi}_{\set{ 1 , 4_0 , 5_1 }} + 
 195 \het{\chi}_{\set{ 2 , 4_0 , 5_2 }} + 
 7935 \het{\chi}_{\set{ 3 , 4_1 , 5_0 }} + 
 93 \het{\chi}_{\set{ 3 , 4_0 , 5_0 }} \\
  &+
 \frac{46278}{51150} z^{(1)}+
 \frac{4354}{51150}z^{(2)}+
 20\het{\chi}_{\emptyset}.
\end{align*}
\end{allowdisplaybreaks}

Since all incidence vectors above correspond to stable sets in $G'$, we obtain that all the vectors above belong to $\cone(\LS_+(G'))$. Thus, we conclude that $Y_0e_{6_0} \in \cone(\LS_+^2(G_{4,1}))$.
\end{proof}

\begin{lemma}\label{lemG413}
Let $Y_0$ be as defined in the proof of Theorem~\ref{thmG41}. Then 
\[
Y_0(e_0 - e_{4_1}) \in \cone(\LS_+^2(G_{4,1})).
\]
\end{lemma}

\begin{proof}
For convenience, let $G \ce G_{4,1}$ throughout this proof. Using $Y_0e_{6_0} \in \cone(\LS_+^2(G))$ from Lemma~\ref{lemG411} and the symmetry of $G$, we know that the vector

{\footnotesize
\[
\begin{array}{rcccccccccccccl}
&& 1 & 2 & 3 & 4_1 & 4_0 & 4_2 & 5_1 & 5_0 & 5_2& 6_1 & 6_0 & 6_2\\
z \ce [& 75020 & 17502 & 25340 & 17502 &
 0 & 75020 & 0 &
 15419 & 51150 & 15911 & 
 15911 & 51150 & 15419
&]^{\top}
\end{array}
\]}
belongs to $\cone(\LS_+^2(G))$. Now observe that
\begin{align*}
Y_0(e_0-e_{4_1}) \leq{} & \frac{2}{3} z + \frac{1}{3} \Big(
 7726 \het{\chi}_{\set{ 1 , 4_0 , 5_1 , 6_0 }} + 
 17105 \het{\chi}_{\set{ 1 , 4_0 , 5_0 , 6_0 }} + 
 16187 \het{\chi}_{\set{ 1 , 4_2 , 5_0 , 6_0 }} \\
  &+
 8509 \het{\chi}_{\set{ 2 , 4_0 , 5_0 , 6_1 }} + 
 8324 \het{\chi}_{\set{ 2 , 4_0 , 5_0 , 6_0 }} + 
 8509 \het{\chi}_{\set{ 2 , 4_0 , 5_2 , 6_0 }} + 
 9486 \het{\chi}_{\set{ 3 , 4_0 , 5_0 , 6_0 }} \\ 
 &+
 8017 \het{\chi}_{\set{ 3 , 4_0 , 5_0 , 6_2 }} + 
 7403 \het{\chi}_{\set{ 4_0 , 5_0 , 6_1 , 6_2 }} + 
 9170 \het{\chi}_{\set{ 4_2 , 5_1 , 5_2 , 6_0 }} + 
24\het{\chi}_{\emptyset} \Big).
\end{align*}

Notice that all incidence vectors above correspond to stable sets in $G$. Since $\cone(\LS_+^{2}(G))$ is a lower-comprehensive convex cone, it follows that $Y_0(e_0-e_{4_1}) \in \cone(\LS_+^{2}(G))$.
\end{proof}

Finally, we provide in Table~\ref{tabUV} the $UV$-certificates of all PSD matrices used in Theorem~\ref{thmG41} and Lemmas~\ref{lemG412},~\ref{lemG411}, and~\ref{lemG413}.

\begin{center}
\begin{table}
\setlength\arraycolsep{1pt}
\def\arraystretch{0.95}
\resizebox{\textwidth}{!}{
\begin{tabular}{ccc}
 & $U$ & $V$ \\
\hline
$Y_0$ & 
$\begin{bmatrix}
-2 & 1 & 1 & 1 & 1 & 1 & 1 & 1 & 1 & 1 & 1 & 1 & 1 \\
0 & 114 & -8 & -107 & 152 & 27 & -114 & 202 & 411 & 348 & -353 & -440 & -234 \\
0 & -57 & 128 & -71 & -320 & -492 & -335 & 292 & 270 & 69 & 29 & 221 & 268 \\
-361 & -710 & -711 & -711 & -88 & 434 & -88 & -88 & 434 & -88 & -89 & 434 & -87 \\
0 & -135 & 521 & -387 & 477 & 230 & 31 & 395 & -171 & -771 & -872 & -60 & 741 \\
0 & 525 & -145 & -379 & 731 & -65 & -873 & -779 & -167 & 409 & 48 & 231 & 462 \\
0 & 0 & 0 & 0 & -984 & 0 & 984 & -984 & 0 & 984 & -984 & 0 & 984 \\
820 & -787 & -788 & -787 & 1222 & -569 & 1222 & 1221 & -569 & 1221 & 1222 & -569 & 1222 \\
0 & 1693 & -2745 & 1052 & -901 & 857 & -628 & 1238 & -329 & -652 & -337 & -529 & 1280 \\
0 & -2192 & -370 & 2562 & 909 & 116 & -1116 & 326 & -801 & 1102 & -1235 & 685 & 14 \\
0 & 1845 & -595 & -1250 & 308 & -990 & 1042 & 1707 & -2079 & 1128 & -2014 & 3069 & -2170 \\
0 & -378 & 1787 & -1409 & -2148 & 2972 & -1904 & 1341 & -2343 & 1854 & 808 & -629 & 50 \\
8215 & 2154 & 2154 & 2154 & 1246 & 6313 & 1246 & 1246 & 6313 & 1246 & 1246 & 6313 & 1246
\end{bmatrix}$
 & 
$\begin{bmatrix}
11050 & 1142 & 1601 & 781 & -196 & 621 & -196 & 624 & 621 & 624 & -557 & 621 & 165 \\
1142 & 15917 & -509 & -499 & -1882 & 3390 & -156 & -27 & 432 & -915 & -1711 & -214 & 1694 \\
1601 & -509 & 12805 & -677 & 958 & -1473 & 1373 & -587 & -70 & 59 & -379 & 235 & 227 \\
781 & -499 & -677 & 14384 & 1269 & 58 & 20 & 75 & 1698 & -1891 & -868 & 901 & 703 \\
-196 & -1882 & 958 & 1269 & 13455 & -739 & -38 & -706 & 1544 & -297 & -936 & 409 & -415 \\
621 & 3390 & -1473 & 58 & -739 & 11866 & -186 & 540 & -364 & 186 & -530 & 381 & 544 \\
-196 & -156 & 1373 & 20 & -38 & -186 & 8568 & -1358 & 220 & 469 & 156 & -1144 & -163 \\
624 & -27 & -587 & 75 & -706 & 540 & -1358 & 10158 & 1776 & -61 & 476 & -1550 & 172 \\
621 & 432 & -70 & 1698 & 1544 & -364 & 220 & 1776 & 11890 & -1026 & 1209 & 198 & 23 \\
624 & -915 & 59 & -1891 & -297 & 186 & 469 & -61 & -1026 & 13347 & -1927 & -193 & 1281 \\
-557 & -1711 & -379 & -868 & -936 & -530 & 156 & 476 & 1209 & -1927 & 11530 & -292 & 43 \\
621 & -214 & 235 & 901 & 409 & 381 & -1144 & -1550 & 198 & -193 & -292 & 9703 & 1337 \\
165 & 1694 & 227 & 703 & -415 & 544 & -163 & 172 & 23 & 1281 & 43 & 1337 & 8773
\end{bmatrix}$
\\
\\
$Y_{2}$ & 
$\begin{bmatrix}
-2 & 1 & 1 & 0 & 0 & 1 & 1 & 0 & 0 & 1 & 1 & 0\\
0 & 1 & -2 & 72 & 70 & -5 & 5 & -70 & -72 & 0 & 0 & 0\\
-67 & -50 & -51 & 82 & 87 & -35 & -35 & 88 & 82 & -57 & -58 & -56\\
-56 & -129 & -129 & 35 & -44 & -250 & -250 & -44 & 35 & 142 & 278 & 143\\
0 & -110 & 110 & -110 & 141 & 384 & -384 & -141 & 111 & 52 & 0 & -52\\
0 & 216 & -216 & -18 & 25 & 125 & -125 & -25 & 18 & -567 & 0 & 567\\
-144 & -423 & -423 & -424 & 255 & 253 & 252 & 255 & -424 & -48 & 263 & -48\\
362 & -556 & -556 & 538 & -221 & 531 & 531 & -221 & 538 & 546 & -248 & 546\\
0 & 928 & -928 & 442 & -429 & 492 & -492 & 429 & -442 & 429 & 0 & -429\\
-62 & 82 & 82 & -588 & 568 & -470 & -470 & 568 & -588 & 1034 & -1489 & 1034\\
0 & -1020 & 1021 & 1031 & -1020 & 432 & -432 & 1020 & -1031 & -371 & 0 & 371\\
3377 & 1203 & 1203 & 673 & 2686 & 338 & 338 & 2686 & 673 & 717 & 2294 & 717
\end{bmatrix}$
 & 
$\begin{bmatrix}
4918 & 41 & -26 & -35 & -29 & -89 & -233 & 38 & -35 & -130 & 576 & -9\\
41 & 6219 & 172 & 148 & 1281 & -511 & 448 & -1143 & 908 & -675 & 173 & -80\\
-26 & 172 & 4074 & -79 & -16 & 399 & -532 & 329 & 1079 & 54 & 115 & -922\\
-35 & 148 & -79 & 5725 & -112 & 974 & -28 & -20 & 201 & -548 & 953 & -344\\
-29 & 1281 & -16 & -112 & 4852 & -1238 & -471 & -31 & -79 & -27 & -22 & -1183\\
-89 & -511 & 399 & 974 & -1238 & 5976 & -273 & -691 & 12 & -533 & 26 & -469\\
-233 & 448 & -532 & -28 & -471 & -273 & 6481 & -948 & 934 & -803 & 289 & -295\\
38 & -1143 & 329 & -20 & -31 & -691 & -948 & 4677 & -53 & -1083 & 36 & -14\\
-35 & 908 & 1079 & 201 & -79 & 12 & 934 & -53 & 5504 & -279 & 953 & -613\\
-130 & -675 & 54 & -548 & -27 & -533 & -803 & -1083 & -279 & 4546 & 77 & -53\\
576 & 173 & 115 & 953 & -22 & 26 & 289 & 36 & 953 & 77 & 9521 & -142\\
-9 & -80 & -922 & -344 & -1183 & -469 & -295 & -14 & -613 & -53 & -142 & 4375\\
\end{bmatrix}$
\\
\\
$Y'_{21}$ & 
$\begin{bmatrix}
-44 & -1 & 0 & -20 & -1 & 0 & 0 & 0\\
233 & 22 & 465 & 229 & 20 & -410 & -389 & 86\\
385 & 793 & -241 & -149 & 305 & 244 & -460 & -545\\
439 & 137 & -877 & -360 & 468 & -700 & -18 & 959\\
1565 & -1333 & -1075 & 528 & -1287 & 508 & -788 & -327\\
655 & -1159 & 930 & -530 & 1923 & 1995 & -641 & 1022\\
-1375 & 1105 & -499 & 1255 & -949 & 1180 & -1505 & 1780\\
2466 & 1620 & 307 & 4002 & 775 & 986 & 3444 & 835\\
\end{bmatrix}$
 & 
$\begin{bmatrix}
3758 & -873 & -519 & -1 & 636 & 9 & -429 & -479\\
-873 & 7118 & -18 & 2774 & -91 & -423 & -1334 & -135\\
-519 & -18 & 1822 & 17 & 282 & -78 & 70 & -402\\
-1 & 2774 & 17 & 8485 & -2504 & -708 & 22 & 687\\
636 & -91 & 282 & -2504 & 6178 & -107 & 586 & 225\\
9 & -423 & -78 & -708 & -107 & 3903 & -635 & -904\\
-429 & -1334 & 70 & 22 & 586 & -635 & 5449 & 602\\
-479 & -135 & -402 & 687 & 225 & -904 & 602 & 5052
\end{bmatrix}$
\\
\\
$Y_{22}$ & 
$\begin{bmatrix}
5 & -3 & -3 & -1 & -1 & -3 & -3 & -1 & -1\\
-1 & -12 & 9 & 8011 & 8038 & -189 & 190 & -8035 & -8009\\
-7048 & -7604 & -7605 & 6746 & 7345 & -5664 & -5664 & 7345 & 6752\\
-1 & 8338 & -8339 & 10389 & -11032 & -29236 & 29235 & 11032 & -10389\\
5767 & 23052 & 23052 & 28635 & -19951 & -26290 & -26290 & -19952 & 28635\\
0 & 82739 & -82739 & 16909 & -15913 & 35612 & -35612 & 15913 & -16909\\
32757 & -40385 & -40385 & 61612 & -35034 & 61878 & 61878 & -35034 & 61613\\
0 & 55744 & -55744 & -88585 & 87206 & -48485 & 48485 & -87206 & 88585\\
239799 & 88332 & 88332 & 60363 & 177671 & 34664 & 34664 & 177671 & 60363
\end{bmatrix}$
 & 
$\begin{bmatrix}
245290 & 67346 & 43642 & 33983 & 37 & 8790 & 67640 & 11795 & 6716\\
67346 & 262681 & -47235 & -60196 & 1386 & -15331 & -35125 & -4652 & 21255\\
43642 & -47235 & 230858 & -23610 & -7313 & -78942 & 460 & 18674 & -1875\\
33983 & -60196 & -23610 & 293133 & -56794 & -5587 & 28035 & 3646 & 3783\\
37 & 1386 & -7313 & -56794 & 116143 & 891 & -8115 & -5728 & -26068\\
8790 & -15331 & -78942 & -5587 & 891 & 159853 & -29478 & -14768 & -1859\\
67640 & -35125 & 460 & 28035 & -8115 & -29478 & 217945 & -6899 & -36239\\
11795 & -4652 & 18674 & 3646 & -5728 & -14768 & -6899 & 124459 & 2908\\
6716 & 21255 & -1875 & 3783 & -26068 & -1859 & -36239 & 2908 & 120960
\end{bmatrix}$
\\
\\
$Y_{23}$ & 
$\begin{bmatrix}
-1 & -1 & 0 & 0 & -1 & 0 & 0 & 0 & -1 & 0 & 0\\
3 & 0 & -10 & -5 & -6 & -1 & 11 & -3 & -6 & 23 & 24\\
-1 & 7 & 26 & 29 & 7 & -50 & -61 & -4 & 4 & 18 & 20\\
-56 & -71 & 75 & 74 & -61 & -45 & 63 & 55 & -78 & -15 & -13\\
-20 & -29 & 71 & 77 & -26 & 125 & -34 & -186 & -29 & 4 & 12\\
3 & -6 & 219 & -154 & -525 & -3 & -17 & 16 & 490 & 160 & -165\\
71 & 395 & 245 & -200 & -264 & 256 & -195 & 248 & -296 & -219 & 283\\
-382 & 1053 & -281 & -46 & -387 & -664 & 309 & -650 & -488 & 157 & -530\\
-61 & 135 & -946 & 906 & -666 & -101 & 25 & -28 & 493 & -776 & 714\\
25 & -59 & 764 & -724 & 359 & -748 & 769 & -732 & 442 & -855 & 883\\
2625 & 1525 & 824 & 1782 & 403 & 373 & 2243 & 358 & 385 & 1800 & 822
\end{bmatrix}$
 & 
$\begin{bmatrix}
3268 & -126 & -24 & -112 & -237 & 198 & -265 & 583 & 192 & -200 & 631\\
-126 & 5487 & -1306 & -32 & 22 & -488 & 501 & -233 & -286 & 8 & 753\\
-24 & -1306 & 3523 & -177 & 4 & 802 & 29 & 9 & -228 & -99 & 75\\
-112 & -32 & -177 & 2815 & -3 & -4 & -57 & 928 & 47 & 597 & -258\\
-237 & 22 & 4 & -3 & 4247 & 119 & -1336 & 82 & -822 & -708 & 8\\
198 & -488 & 802 & -4 & 119 & 3240 & 189 & 165 & 337 & -120 & 68\\
-265 & 501 & 29 & -57 & -1336 & 189 & 4825 & 96 & -29 & 796 & 459\\
583 & -233 & 9 & 928 & 82 & 165 & 96 & 3654 & 262 & 374 & -122\\
192 & -286 & -228 & 47 & -822 & 337 & -29 & 262 & 4115 & -118 & 199\\
-200 & 8 & -99 & 597 & -708 & -120 & 796 & 374 & -118 & 4024 & -139\\
631 & 753 & 75 & -258 & 8 & 68 & 459 & -122 & 199 & -139 & 4828
\end{bmatrix}$
\\
\\
$Y_{24}$ & 
$\begin{bmatrix}
-23 & -1 & -1 & 0 & 0 & -1 & -1 & -1 & 0 & 0 & -1\\
20 & 7 & 6 & -12 & -25 & 0 & -15 & -19 & 0 & 5 & 10\\
-4 & -1 & 2 & 44 & 43 & -2 & -42 & -42 & -2 & -2 & 0\\
72 & 31 & 31 & -119 & 136 & 21 & 3 & 6 & 13 & 7 & 20\\
-41 & -70 & 16 & -41 & -1 & -164 & -61 & 26 & 192 & 188 & -6\\
-19 & -131 & -304 & -36 & 25 & 107 & 83 & -101 & -95 & 203 & 256\\
-13 & -170 & -62 & -28 & 16 & 369 & 166 & -186 & 294 & 26 & -360\\
82 & -767 & 605 & 68 & 15 & 358 & -253 & 330 & -375 & 269 & 44\\
33 & 510 & -537 & 20 & 15 & 373 & -620 & 652 & -252 & 506 & -562\\
-222 & 245 & 337 & -215 & -3 & -333 & 402 & -620 & -773 & 659 & -567\\
2058 & 955 & 634 & 2045 & 7 & 242 & 1551 & 501 & 533 & 1338 & 451
\end{bmatrix}$
 & 
$\begin{bmatrix}
3039 & 256 & -4 & 11 & 33 & 616 & 259 & 212 & 195 & 517 & 213\\
256 & 2748 & -376 & -1 & -168 & -556 & 192 & -60 & 147 & -215 & -197\\
-4 & -376 & 2699 & -164 & 9 & -114 & -33 & -861 & 7 & 513 & -52\\
11 & -1 & -164 & 3056 & -981 & -238 & 125 & -117 & 651 & -6 & -116\\
33 & -168 & 9 & -981 & 2948 & -11 & -251 & 5 & -464 & 102 & -130\\
616 & -556 & -114 & -238 & -11 & 3982 & 71 & 516 & 145 & 253 & -36\\
259 & 192 & -33 & 125 & -251 & 71 & 4861 & -1304 & 80 & 1182 & 218\\
212 & -60 & -861 & -117 & 5 & 516 & -1304 & 3960 & 18 & -14 & 8\\
195 & 147 & 7 & 651 & -464 & 145 & 80 & 18 & 2459 & 90 & -12\\
517 & -215 & 513 & -6 & 102 & 253 & 1182 & -14 & 90 & 5059 & 1081\\
213 & -197 & -52 & -116 & -130 & -36 & 218 & 8 & -12 & 1081 & 2689
\end{bmatrix}$
\\
\\
$Y_{1}$ &
$\begin{bmatrix}
-3 & 1 & 1 & 1 & 1 & 1 & 1 & 1 & 1 & 1 \\
0 & 0 & 312 & -312 & -663 & -1314 & -1080 & 1080 & 1313 & 665 \\
639 & 262 & 1834 & 1834 & -1199 & -1101 & 1007 & 1007 & -1100 & -1201 \\
-573 & -2813 & -954 & -953 & -1543 & 1304 & 1872 & 1872 & 1305 & -1544 \\
0 & 0 & 1938 & -1939 & 3269 & 304 & -1822 & 1822 & -303 & -3269 \\
0 & 0 & 6254 & -6254 & -1331 & -1182 & 4069 & -4069 & 1182 & 1331 \\
-2501 & 2818 & 1055 & 1055 & -5112 & 3150 & -4138 & -4138 & 3151 & -5112 \\
-123 & -9594 & 4954 & 4954 & 2208 & 1240 & -3751 & -3751 & 1240 & 2208 \\
0 & 0 & 3904 & -3904 & -6195 & 8605 & -5528 & 5528 & -8605 & 6195 \\
22183 & 7784 & 5318 & 5318 & 4055 & 15768 & 4304 & 4304 & 15768 & 4055
\end{bmatrix}$
 & 
$\begin{bmatrix}
15782 & 3040 & 944 & 1517 & -961 & 810 & -1074 & -1074 & 3245 & -256 \\
3040 & 28790 & -3737 & -924 & -1274 & 561 & -1625 & -1625 & 294 & -3563 \\
944 & -3737 & 23948 & -126 & -1931 & 2745 & 203 & -3361 & -4007 & 3326 \\
1517 & -924 & -126 & 21978 & 6048 & -1446 & -7055 & 153 & 828 & -318 \\
-961 & -1274 & -1931 & 6048 & 37807 & -12613 & -459 & -3449 & 188 & 3617 \\
810 & 561 & 2745 & -1446 & -12613 & 36247 & 852 & -2975 & 6408 & -2004 \\
-1074 & -1625 & 203 & -7055 & -459 & 852 & 18359 & 322 & -974 & 2597 \\
-1074 & -1625 & -3361 & 153 & -3449 & -2975 & 322 & 18359 & 1369 & 1267 \\
3245 & 294 & -4007 & 828 & 188 & 6408 & -974 & 1369 & 32772 & -4348 \\
-256 & -3563 & 3326 & -318 & 3617 & -2004 & 2597 & 1267 & -4348 & 27264
\end{bmatrix}$
\\
\\
$Y_{11}$ &
$\begin{bmatrix}
-2 & 1 & 1 & 1 & 1 & 1 & 1 \\
1218 & 1284 & 1319 & 2378 & -1339 & -3175 & -1273 \\
11 & 2338 & -2256 & -81 & -4392 & 40 & 4398 \\
-981 & 5772 & 5824 & -8802 & -2137 & -362 & -2372 \\
169 & -10840 & 10788 & -537 & -4810 & -823 & 6490 \\
2335 & 1031 & -1134 & -5111 & 10842 & -11453 & 9705 \\
26748 & 9331 & 9439 & 5079 & 7256 & 15899 & 7008
\end{bmatrix}$
 & 
$\begin{bmatrix}
32150 & -109 & 12752 & 4226 & 506 & -460 & 13236 \\
-109 & 21193 & 568 & -2619 & -6911 & -795 & -2712 \\
12752 & 568 & 20891 & 130 & -600 & 863 & -5224 \\
4226 & -2619 & 130 & 19072 & -4217 & 463 & -1160 \\
506 & -6911 & -600 & -4217 & 20400 & 2412 & -2691 \\
-460 & -795 & 863 & 463 & 2412 & 32551 & -3917 \\
13236 & -2712 & -5224 & -1160 & -2691 & -3917 & 37492
\end{bmatrix}$\\
\end{tabular}}

\caption{$UV$-certificates for matrices in the proofs of Theorem~\ref{thmG41} and Lemmas~\ref{lemG412},~\ref{lemG411}, and~\ref{lemG413}}\label{tabUV}
\end{table}
\end{center}

\end{document}